%% file: thesis.tex
\documentclass[12pt,Bold,Justify]{mcgilletdclass}
\usepackage[dvips,final]{graphicx}

\usepackage{epsfig}
\usepackage{amsfonts}
\usepackage{listings}
\usepackage{amsmath}

\usepackage{amsfonts, amssymb, amsmath, amsthm, latexsym}   
\usepackage{epic,eepic}
\usepackage{graphicx}
\usepackage{xcolor}
\usepackage{pstricks}				
\input{include/Qcircuit.tex}		
\input{include/header.tex}		
\input{include/graphviz.tex} 	

\SetTitle{\huge{	Distributed Compression \\ and\\ Squashed Entanglement}}
\SetAuthor{Ivan Savov}%
\SetDegreeType{Master of Science}%
\SetDepartment{Physics Department}%
\SetUniversity{McGill University}%
\SetUniversityAddr{Montreal, Quebec}%
\SetThesisDate{\today}
\SetRequirements{	A thesis submitted to McGill University in partial fulfillment of the \\
                  	requirements of the degree of Master of Science.}%
\SetCopyright{\copyright Ivan Savov, 2007}%

\makeindex

\begin{document}

\maketitle%

\begin{romanPagenumber}{2}%
	
	\SetDedicationName{\MakeUppercase{Dedication}}%
	\SetDedicationText{ \begin{center}	
						To my parents, the only ``system'' I have respect for. 
						\end{center}
					   }%
	\Dedication%
	
	\SetAcknowledgeName{\MakeUppercase{Acknowledgements}}%
	\SetAcknowledgeText{
			This work would not have been possible without the help and guidance of my supervisor Prof. Patrick Hayden. 
			His outstanding pedagogical abilities and attention to detail have helped shape my understanding of 
			the field of quantum information science at a world class level.
			In addition, I would like to thank Fr\'ed\'eric Dupuis, S\'ebastien Gambs and Omar Khalid for the 
			many fruitful discussions about information theory and their help with the preparation of this
			manuscript. 
			I also owe many thanks to Prof. David Avis and Leonid Chindelevitch for their assistance with some of 
			the most difficult parts in this work.
			There are many other people who deserve an honorable mention for either directly or indirectly
			influencing me: Claude Cr\'epeau, Aram Harrow, Debbie Leung, Jonathan Oppenheim and Andreas Winter.
			Last but not least, I want to thank my parents for cultivating in me the love of science and knowledge.
	}%
	\Acknowledge%

			\SetAbstractEnName{\MakeUppercase{Abstract}}%
			\SetAbstractEnText{ 
				A single quantum state can be shared by many distant parties. 
				In this thesis, we try to characterize the information contents of such 
				distributed states by defining the multiparty information and the multiparty squashed 
				entanglement, two steps toward a general theory of multiparty quantum information.
				\newline \indent
				As a further step in that direction, we partially solve the multiparty distributed  
				compression problem where multiple parties use quantum communication to faithfully 
				transfer their shares of a state to a common receiver.
				We build a protocol for multiparty distributed compression based on the fully quantum 
				Slepian-Wolf protocol and prove both inner and outer bounds on the achievable rate region.
				We relate our findings to previous results in information theory and discuss some 
				possible applications.
				 }

			\AbstractEn%

			\SetAbstractFrName{\MakeUppercase{ABR\'{E}G\'{E}}}%
			\SetAbstractFrText{ 
					Un \'etat quantique peut \^etre partag\'e entre
					plusieurs entit\'es qui sont spatialement s\'epar\'es. Dans
					ce m\'emoire, nous essayons de caract\'eriser
					l'information quantique contenue dans de tels \'etats
					distribu\'es en d\'efinissant et utilisant les notions
					d'information multipartie (multiparty information) et
					d'intrication ``\'ecras\'ee'' multipartie (multiparty
					squashed entanglement). Il s'agit de premiers pas 
					vers une th\'eorie g\'en\'erale de l'information quantique
					multipartie. 
					\newline \indent Nous faisons aussi un autre pas dans cette direction
					en \'etudiant le probl\`eme de la compression distribu\'ee
					d'information quantique. En particulier, nous
					proposons un protocole de compression distribu\'ee bas\'e
					sur la version quantique du protocole de Slepian et Wolf
					et analysons ses caract\'eristiques. Nous discutons
					aussi la relation entre nos r\'esultats et les travaux
					pr\'ec\'edents dans la th\'eorie de l'information   
					et soulignons quelques applications possibles de notre
					protocole.		
								}
					%
			\AbstractFr%

\TOCHeading{\MakeUppercase{Table of Contents}}%
\LOTHeading{\MakeUppercase{List of Tables}}%
\LOFHeading{\MakeUppercase{List of Figures}}%
\tableofcontents %

	%
\listoffigures %

\end{romanPagenumber}

\mainmatter %

\chapter{Introduction}

	Information theory is one of the most important mathematical theories developed in the last century. 
	It finds applications in communications engineering, computer science, physics, economics,
	neuroscience and many other fields of modern science.
	Of particular interest are the recent developments in quantum information theory (QIT),%
	\index{QIT@QIT: quantum information theory} a discipline which studies the limits that the laws of 
	quantum mechanics impose on our ability to store, manipulate and transmit information.
	All information is physical; whether it be the magnetic domains of a hard disk platter, 
	the reflective bumps on the surface of a DVD or the charge of the capacitors in a stick of
	RAM, that which we intuitively refer to as information must be stored in some physical system \cite{landauer}.
	Thus, the incursion of quantum physics into information theory is inevitable if we want to understand 
	the information properties of quantum systems like single photons and superconducting loops.
	
	Modern quantum information theory has elaborated a paradigm in which a set of spatially localized parties 
	try to accomplish a communication task by using communication resources like channels, states and quantum 
	entanglement \cite{BBPS,DHW04,DHW05b,FQSW}.
    Such an approach is now possible because of the substantial body of results characterizing 
    quantum communication channels \cite{H98,SW97,BSST99,D03} and the truly quantum resource of shared 
    entanglement \cite{B96,DW05,PV07}.
	%
	In this new quantum paradigm of information theory, many classical results need to be revisited in the light
	of the peculiar properties of quantum information. 
	
	
    In classical information theory, distributed compression is the search for the optimal rates 
    at which two parties Alice and Bob can compress and transmit information faithfully to a third party Charlie.
    If the senders are allowed to communicate among themselves then they can obviously use the
    correlations between their sources to achieve better rates.
    The more interesting problem is to ask what rates can be achieved if no communication is allowed between
    the senders.
    The classical version of this problem was solved by Slepian and Wolf~\cite{SW73}.
    The quantum version of this problem was first approached in \cite{ADHW04,HOW05} and more recently in \cite{FQSW},
    which describes the fully quantum Slepian-Wolf (FQSW) protocol and partially solves the distributed compression
    problem for two senders.
    
    \index{FQSW@FQSW: fully quantum Slepian-Wolf} 
	In this thesis, we analyze the multiparty scenario of distributed compression where many senders, 
	Alice $1$ through Alice $m$, send quantum information to a single receiver, Charlie. 
	We will describe the multiparty FQSW protocol and exhibit a set of achievable rates for this protocol.
	We also derive an outer bound on the possible rates for \emph{all} distributed compression protocols based on the
	multiparty squashed entanglement.
    
        
    The multiparty squashed entanglement (independently discovered by Yang, et al. \cite{multisquash}) is
	a generalization of the squashed entanglement defined by Christandl and Winter \cite{CW04}
	and has very desirable properties as a measure of multiparty entanglement.
    While there exist several measures for bipartite entanglement with useful properties and applications
    \cite{BBPS, HHT, Ra99, VP98}, the theory of multiparty entanglement, despite considerable effort
    \cite{LSSW,DCT99,CKW00,BPRST99}, remains comparatively undeveloped.
    Multiparty entanglement is fundamentally more complicated because it cannot be described by a single number
    even for pure states.
    We can, however, define \emph{useful} entanglement measures for particular applications, and
    the multiparty squashed entanglement is one such measure well-suited to application in the distributed 
    compression problem.

	The main results of this thesis are contained in Chapters \ref{chapter:multiparty-quantum-information} and 
	\ref{chapter:multiparty-distributed-compression}.
	Chapter~\ref{chapter:multiparty-quantum-information} presents our original results on the multiparty
	generalization of squashed entanglement.
	Chapter~\ref{chapter:multiparty-distributed-compression} deals with the multiparty distributed compression
	problem and proves inner and outer bounds on the rate region.
	Before we get there, however, we will introduce some background material on classical and quantum information 
	theory in Chapter~\ref{chapter:background}. 
	In Chapter~\ref{chapter:results-in-QIT}, we describe some important recent results of quantum information 
	theory which form the basic building blocks for our results.
	Finally, in Chapter~\ref{chapter:applications-to-black-holes} we take a look at some possible applications 
	of the distributed compression results to the black hole information paradox.
	The dependency graph for the sections in this thesis is shown in Figure~\ref{fig:dependency} on the next page.
	

	Most of the original results in Chapters~\ref{chapter:multiparty-quantum-information} and
	\ref{chapter:multiparty-distributed-compression} appear in a paper \cite{AHS07} co-authored with 
	Prof. David Avis and Prof. Patrick Hayden to which the author has made substantial contributions.


	    \begin{figure}[ht]  
	    \begin{center}
	    	\includegraphics[scale=0.7]{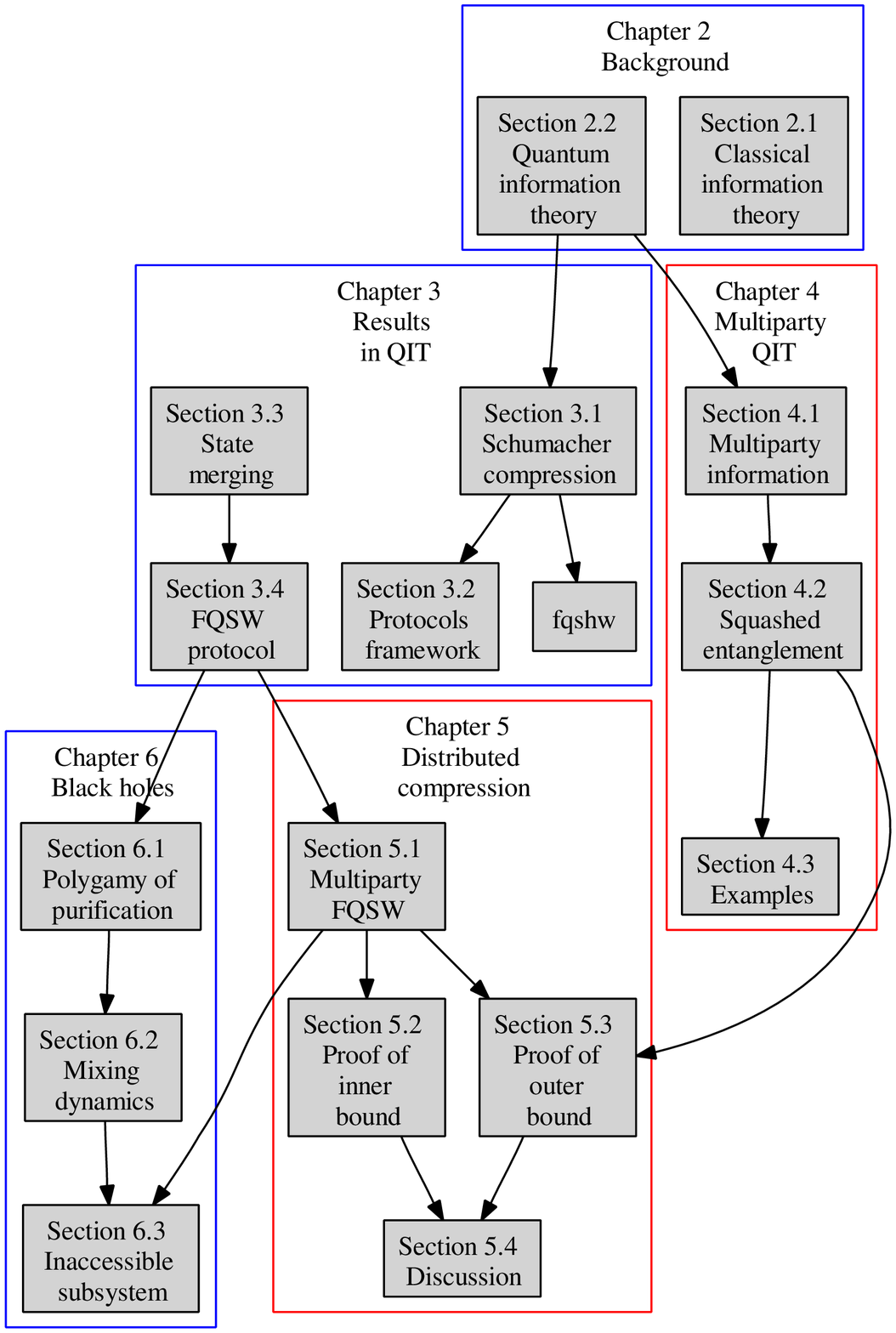}
	    \end{center}
			\FigureCaptionOpt{  Dependency graph for the parts of this thesis.}{
								Dependency graph for the parts of this thesis.}	             
	        \label{fig:dependency}
	    \end{figure}

\chapter{Background}	\label{chapter:background}

	In this chapter, we will present background concepts from classical information theory and their analogues in
	quantum information theory. 
	These concepts form the basic building blocks with which we will construct all subsequent results.
	Our coverage of the information theoretic topics is far from exhaustive; it serves to introduce a minimum 
	prerequisite structure that can support the rest of the exposition. For an in-depth view of classical and 
	quantum information theory the reader is referred to the classics in the fields: \cite{CT91} and \cite{NC04}
	respectively.

	\bigskip
	\bigskip	
	\section{Classical information theory}
	
	In 1948 Claude Shannon published a seminal paper \cite{S48} titled ``A mathematical theory of communication'' which
	set the stage for what has become one of the most fruitful modern mathematical theories. 
	The field of \emph{information theory} was born out of the need of communication engineers to quantify the
	information carrying capacities of channels and the theoretical aspects of data compression.
	
		\bigskip
		\subsection{Foundations}
			
			At the root of Shannon's information theory is the simplifying assumption that information ultimately boils 
			down to the statistics of the symbols used to express it.
			Thus, another name for information theory could be \emph{information statistics}.
			By focusing solely on the statistics of the symbols, which can be described by mathematical equations 
			and axioms 
			in the spirit of Hilbert's program~\cite{cyborgscience}, we can dispense with the difficult 
			semantical questions related to humans.
			
			\index{information!source} 			
			We say that information is produced by a \emph{source}, which is a random variable $X$
			that takes on values from an alphabet $\cX=\{\alpha^1,\alpha^2,\ldots,\alpha^{|\cX|} \}$  
			according to some probability distribution $\Pr\{X=x\}=p(x)$.
			
			\begin{example}
				Let $X$ be the outcome of a coin flip. We will denote the alphabet $\cX=\{\text{`H'}, \text{`T'}\}$.
				If the coin is fair, then 
				\begin{eqnarray*}
					\Pr\{X=\alpha^1\} &\equiv& \Pr\{X=\text{`H'}\} \ =\  0.5, \\
					\Pr\{X=\alpha^2\} &\equiv& \Pr\{X=\text{`T'}\} \ =\  0.5.
				\end{eqnarray*}
				In this case, all outcomes are equally likely and it is maximally difficult to guess
				the result of the coin flip.
			\end{example}

			\begin{example}
				Suppose the Canadian border control center receives an hourly status message $M$ from a distant outpost. 
				The possible messages are:
				\begin{itemize}
					\item	No one has attacked, which occurs 99.7\% of the time: \mbox{$\Pr\{M\!=\!\alpha^0\}\!=0.997$}
					\item	\joke{Neighbour 1 has}{The Americans have} attacked: $\Pr\{M=\alpha^1\}=0.002$
					\item	\joke{Neighbour 2 has}{The Russians have} attacked: $\Pr\{M=\alpha^2\}=0.001$ \joke{}{
							\footnote{The quoted probabilities may not reflect the current 
									  geo-political balance of power.}}
				\end{itemize}
				In this scenario, one of the outcomes, $\alpha^0$, is much more likely than all the others.
				If we were to shut down the remote outpost and instead guess $M=\alpha^0$ every hour, we would only
				be wrong 0.3\% of the time! Of course, implementing such an approximate border defense system 
				is a silly idea, but in other situations an approximate result is just as good as the exact one.
			\end{example}
			
		 	Do we learn more information from one outpost message $M$, or from one coin flip $X$? 
		 	Using information theory, we should be able to \emph{quantify} the amount of information produced by each
		 	source.
		
		\bigskip
		\subsection{Shannon entropy}
			
			According to Shannon, there exists a single function sufficient to quantify the information content of a
			source. \index{information!content}
			This function is the key building block in all of information theory.
			
			\begin{definition}[Shannon entropy] 
				Given a statistical source $X$, over the alphabet $\cX$ with probability function $p(x)$, the quantity 
				\be
					H(X)= - \sum_{x \in \cX}  p(x) \log_2 p(x)
				\ee
				is the \emph{Shannon entropy} of the source.
			\end{definition}
			\index{entropy!shannon-entropy@$H(X)$: Shannon entropy}
			
			The entropy of an unknown source measures our \emph{uncertainty} about it and therefore, it measures
			how much information we learn, on average, when we look at a symbol from that source. 
			The entropy is typically measured in \emph{bits} since we use the base-2 logarithm in the calculation.
			
			The quantity $-\!\sum_{x \in \cX}  p(x) \log p(x)$ also appears in thermodynamics where it is known
			as the Boltzmann-Gibbs entropy function. It is used to denote the logarithm of the number of available
			microstates that are consistent with certain macroscopic constraints~\cite{callen,M92}. 
			Together the entropy, energy, volume, pressure and temperature form the macroscopic description
			of a given thermodynamical system.
			

			We now revisit the coin flip and outpost message scenarios from the previous examples.
			
			\begin{example} 
				The entropy of the balanced coin flip is:
				\be \nonumber
					H(X) = - 0.5\log0.5 - 0.5\log0.5 = -\log0.5 = \log2 = 1 \ [\textup{bit}].
				\ee
				In other words, we learn one bit of information every time we flip the coin.
				On the other hand, the entropy of an outpost message is only
				\be \nonumber
					H(M) = - 0.997\log0.997 - 0.002\log0.002 - 0.001\log0.001= 0.03222 \ [\textup{bits}].
				\ee
				Therefore, every coin flip carries about 30 times more information than a message from the 
				distant outpost.
			\end{example}

		 	The true power of the information theoretic approach becomes apparent when we try to describe 
		 	very long strings of symbols produced \emph{independently} by the same source.
		 	Consider a source $X$ which is used $n$ times to produce the sequence $X_1,X_2,\ldots,X_n$.
		 	We will denote the entire sequence as $X^n$ with a superscript. 
			We assume that the random variables $X_i$ are independent and identically distributed (\emph{i.i.d.})%
			\index{i.i.d.@i.i.d.: identical, independently distributed} according to $p(x)$.
						
			We can write down the probability of a given string $x^n=x_1,x_2,\ldots,x_n$ occurring as
			\bea
				\Pr\{X^n = x^n\}  	&=& p(x_1,x_2,\ldots,x_n)  \nonumber \\
									&=& p(x_1)p(x_2)\cdots p(x_n) \label{prodProbs}
			\eea
			since the $X_i$'s are independent.
			
			Next we ask the important question: 
			\begin{quote}
			\hspace{-2mm}``How often does the symbol $\alpha^i$ occur, on average, in a sequence \hspace{4mm} of 
				  $n$ uses 
				  of the source $(X_1,\ldots,X_n)$?''
			\end{quote}
			Because every one of the symbols in the sequence has $\Pr\{X = \alpha^i\}=p(\alpha^i)$, the overall 
			number of  $\alpha^i$'s in a string of length $n$ is going to be approximatively $np(\alpha^i)$. 
			Therefore, on average, the probability of a string $x_1,x_2,\ldots,x_n$ is 
			\begin{eqnarray}\!\!\!\!\!\!\!\!\!\!\!
			p(x_1,\ldots,x_n) 	&=&	p(x_1)p(x_2)\cdots p(x_n) \label{productOfProbs}\\
						&\approx&		\underbrace{p(\alpha^1)\cdots p(\alpha^1)}_{np(\alpha^1) \textup{ times}}
										\underbrace{p(\alpha^2)\cdots p(\alpha^2)}_{np(\alpha^2) }\cdots 
										\underbrace{p(\alpha^{|\cX|})\cdots p(\alpha^{|\cX|})}_{np(\alpha^{|\cX|}) }
										\label{informalApprox}\\ 
						&=&				\prod_{x \in \cX} \Bigl[ p(x) 			\Bigr]^{np(x)} \nonumber \\ 
						&=&				\prod_{x \in \cX} \Bigl[ 2^{\log_2 p(x)} \Bigr]^{np(x)} \nonumber \\ 
						&=&				2^{n \left[\sum_{x \in \cX} p(x) \log_2 p(x)\right] } \nonumber \\
						&=&				2^{\neg nH(X)}. \label{allTheSame} 
			\end{eqnarray}

			By going from equation \eqref{productOfProbs} to \eqref{informalApprox}, we have made a crucial change 
			in our point of view: 
			instead of taking into account the individual symbols $x_i$ of the sequence, we focus on the 
			global count of the symbol's occurrences. 
			In other words, we abandon the microscopic description of the string and trade it for a macroscopic one
			in the spirit of thermodynamics.
			At first, it is difficult to believe that the typical sequences
			all 
			have the same constant probability of occurrence, but we will see in the next section 
			that this intuitive argument can be made rigorous. 
		
		\bigskip
		\subsection{Typical sets} \label{section:typical-sets}

			Much of information theory is based on the concept of typical sequences.
			In the i.i.d. regime, \index{i.i.d.@i.i.d.: identical, independently distributed} 
			nearly all of the sequences produced by the source have the same probability of occurrence.
			Consider the following theorem which makes precise our earlier argument.
			\index{AEP@AEP: asymptotic equipartition theorem}
			
			\begin{theorem}[Asymptotic equipartition theorem] \label{thm:AEP}
				Let $X_1,X_2,\ldots,X_n$ be a sequence of independent random variables distributed according
				to $p(x)$, then
				\be
					\lim_{n \to \infty}
					\Pr\left\{ 
							\left\vert \neg{\tfrac{1}{n}\log p(X_1, X_2, \ldots, X_n)} - H(X) \right\vert 
							> \epsilon 
						\right\}
					= 0, \qquad \forall \epsilon > 0.
				\ee
			\end{theorem}
		
			\noindent In other words, for large enough $n$, the probability that a sequence occurs approaches 
			$2^{-nH(X)}$ --- a constant value. 
			The result can also be interpreted in a different manner: sequences that have probability different 
			from $2^{-nH(X)}$ are not likely to occur.
			Using this insight, we can partition the space of all possible sequences, $\cX^n$, into two sets. 
			The set of sequences that have probability of occurrence close to $2^{-nH(X)}$ and those that do not. 
			We will call the former the set of typical sequences. \index{typical!set}
			
			\begin{definition}[Typical set]
				The set of \emph{entropy typical sequences} with respect to $p(x)$ is the set of all sequences $x_1,x_2,\ldots,x_n$ satisfying:
				\be
					2^{-n\left( H(X) + \epsilon \right)} 
						\leq p(x_1,x_2,\ldots,x_n) \leq 
								2^{-n\left( H(X) - \epsilon \right)}
				\ee
				we will denote this set $\typ$.
			\end{definition}
			
			{\renewcommand{\labelenumi}{(\roman{enumi})\ \ \ } 
			\noindent The typical set, $\typ$, has the following properties:
			\begin{enumerate}
				\item $\Pr\{X^n \in \typ \} \geq 1-\delta
						\qquad \quad \forall \epsilon, \delta \text{\ and $n$ sufficiently large.}$
				\item $| \typ | \leq 2^{n[H(X)+\epsilon]}
						\!\qquad \qquad \quad \forall \epsilon \text{\ and $n$ sufficiently large.}$
						
			\end{enumerate}
			}
			Property (i) is a consequence of the asymptotic equipartition theorem 
			and says that for large $n$, most of the sequences that come out of the source will be typical. 
			Property~(ii) is a bound on the size of the typical set which follows from the fact that all 
			typical sequences occur with the same probability.
			
			The bound on the size of the typical set is at the root of our ability to compress information.
			
		\bigskip
		\subsection{Compression}	\label{subsection:shannon-compression}
			
			Compression, also referred to as \emph{source coding}, is our ability to encode a given source string
			into a shorter string while preserving most of the information contained therein.
			\index{compression!clas@classical, Shannon}
			More generally, we talk about a \emph{compression rate} which can be achieved for a given source $X$.
			\index{rate!compression}
			\begin{definition}[Compression rate]
				We say a compression rate $R$ is achievable if for all $\epsilon>0$, there exists $N(\epsilon)$
				such that for $n > N(\epsilon)$, there exist maps:
				\bea
					E_n:\cX^n	 	& \rightarrow & 	\m \qquad |\m| = 2^{nR} \\
					D_n:\m 			& \rightarrow &		\cX^n 
				\eea
				such that 
				\be
					\Pr\{X^n \neq Y^n\} < \epsilon 
				\ee
				where $Y^n = \left(D_n \circ E_n\right)X^n$.
			\end{definition}
			
			Shannon's compression theorem \cite{S48} provides a bound the compression rates that are 
			achievable for a given source $X$.
			\begin{theorem}[Shannon source coding] 	\label{thm:shannon-source-coding}
				Let $X^n \equiv X_1,X_2,\ldots X_n$ be a sequence of symbols i.i.d. $\sim p(x)$, 
				then any compression rate $R$ which satisfies 
				\be
					R > H(X),
				\ee
				is achievable for $n$ sufficiently large.
			\end{theorem}
			
			The idea behind Shannon compression is very simple.
			We begin by indexing the set of typical sequences $\typ$ in some order.
			We know that the size of $\typ$ is
			\be
				\left\vert \typ \right\vert \leq 2^{n[H(X)+\epsilon]}.
			\ee
			therefore labels of length $\left\lceil H(X)+\epsilon \right\rceil$ bits will be sufficient to 
			index the typical sequences.
			The encoding operation $E_n$ for a given string $x^n$ consists of:
			\begin{itemize}
				\item Recording the index of $x^n$ if $x^n \in \typ$ and
				\item Rejecting the string and recording ``error'' if $x^n \notin \typ$.
			\end{itemize}
			The decoding operation $D_n$ simply takes the index record and replaces it with the original string.
			
			Because of Property~(i) of the set of typical sequences, we know that the ``error'' condition 
			will occur rarely:
			\be
				\Pr\{x^n \notin \typ \} < \delta 
					\qquad \quad \forall \epsilon, \delta \text{\ and $n$ sufficiently large.}
			\ee
			This guarantees the low-error condition $\Pr\{X^n \neq Y^n\} < \delta$ for any $\delta > 0$.
			Shannon's coding theorem holds since the rate $R = \left\lceil H(X)+\epsilon \right\rceil$ is achievable 
			for any $\epsilon$ provided $n$ is large enough.

		\bigskip
		\subsection{Multiple sources}	\label{subsection:multiple-sources}

			When we consider situations involving more than one source, some new information theoretic 
			quantities become relevant.
			Consider now two sources $X$ and $Y$ distributed jointly according to $p(x,y)$. We will denote 
			the marginals $p(x) = \sum_y p(x,y)$ and $p(y) = \sum_x p(x,y)$.
		
			Fist we define the quantity 
			\bea 	\index{entropy!conditional}
				H(X|Y)		&=&		- \sum_{x,y} p(x,y) \log \frac{p(x,y)}{p(y)}\\
							&=&		\sum_{y} p(y) H(X|Y=y) \\
							&=&		H(XY) - H(Y)	\label{conditional-classical-formula}
			\eea
			which is known as \emph{conditional entropy}. 
			The conditional entropy measures the uncertainty in $X$ that remains if we know the value of $Y$.

			The quantity that quantifies how much information is shared between two sources is
			\be
				I(X:Y) = H(X)+H(Y)-H(X,Y)
			\ee
			and is usually referred to as the \emph{mutual information}. 
			\index{information!mutual@$I(X:Y)$: mutual information}
			Two sources which have zero mutual information are independent. 
			
			The mutual information plays a key role in the characterization the information carrying 
			capacity of memoryless channels. Together the compression and channel capacity formulas are the two pillars
			of Shannon's information theory. In this thesis we focus mainly on compression problems and 
			refer the reader interested in channel capacities to the classic texts \cite{CT91,CK81}.
			\index{channel capacity}
			
			In order to get a better intuitive understanding of the conditional entropy and the mutual information,
			we often use a Venn-like diagram to represents them as in Figure~2--1. 
			\begin{figure}[ht] 	
				\begin{center}
					\input{figures/FIGVENN.pic} 
					\FigureCaptionOpt{Graphical representation of the conditional entropy and the mutual information.}{
									  Graphical representation of the conditional entropy and the mutual information.}
				\end{center}
				\label{fig:diagram}
			\end{figure}
			
			Furthermore, one can define the \emph{conditional mutual information} by conditioning the mutual information
			formula on a third system $Z$.
			\bea			\index{information!conditional}
				I(X:Y|Z)		&=&		H(X|Z) + H(Y|Z) - H(XY|Z)\\
								&=&		H(XZ) + H(YZ) - H(XYZ) - H(Z).
			\eea
			The conditional mutual information measures the correlations between $X$ and $Y$
			that are not shared with the variable $Z$.

		\bigskip
		\subsection{Slepan-Wolf coding}			\label{subsection:classical-slepian-wolf}
			
			Next we turn to the compression of two correlated sources $X$ and $Y$ distributed according to $p(x,y)$.
			If the two sources can be encoded together, then according to the Shannon's theorem a compression rate of
			$H(X,Y)$ can be achieved.
			The more interesting problem requires the sources to be encoded separately without communication between
			the encoders. 
			This is known as the Slepian-Wolf source coding problem \cite{SW73}.
			
			Using the coding scheme suggested by Slepian and Wolf, we can compress at rates $(R_X,R_Y)$ 
			for $X$ and $Y$ respectively if they satisfy the inequalities 
			\be	\label{slepian-wolf-inequalities}
			\begin{aligned}
				R_X 		&>	 	\ H(X|Y),	\\
				R_Y 		&>		\ H(Y|X),	\\
				R_X+R_Y		&>		\ H(XY).	
			\end{aligned}
			\ee
			This set of	inequalities corresponds to an achievable \emph{rate region} in the $(R_X,R_Y)$-plane,
			as illustrated in Figure~\ref{fig:classicalSlepainWolfRegion}. 
							\index{rate!region}
			
			\bigskip
			\begin{figure}[ht]   \begin{center}
		        {\small
		        \input{figures/ClassicalSlepianWolfRedion.pst}             
		        }
		        \end{center}
		        \FigureCaptionOpt{The classical Slepian-Wolf rate region.}{
			        			  The classical Slepian-Wolf rate region. The points $\alpha$ and $\beta$
			        			  are two corner points of the region.}
		        \label{fig:classicalSlepainWolfRegion}
		    \end{figure}				
			
			To prove that the Slepian-Wolf rate region is achievable, we only need to show protocols which achieve
			the rates of the two corner points $\alpha$ and $\beta$. Any rate pair on the line between the two corner 
			points can be achieved by \emph{time sharing}. All other points in the rate
			region can be obtained by \emph{resource wasting}.
			
			The proof that the corner points are achievable relies on a coding scheme based on random bins 
			and the properties of jointly typical sequences.  \index{typical!jointly}
			A string $(x^n,y^n)$ is \emph{jointly typical} if $x^n$ is typical according to $p(x)$, $y^n$ is typical 
			according to $p(y)$ and $(x^n,y^n)$ is typical according to $p(x,y)$.
			To encode, we will randomly assign to each string $x^n$ an index $i(x^n) \in \{1,2,\ldots,2^{nR_X}\}$.
			Similarly, to each $y^n$ we assign an index $j(x^n) \in \{1,2,\ldots,2^{nR_Y}\}$. 
			The decoding operation takes the received indices $(i,j)$ and tries to reproduce 
			a copy of the original string $(\hat{x}^n,\hat{y}^n)$.
			
			In the case of point $\alpha$ from Figure~\ref{fig:classicalSlepainWolfRegion}, the rates correspond to
			\bea
				R_X &=& H(X) + \epsilon_1,		\label{rateX}\\
				R_Y &=&	H(Y|X) + \epsilon_2.	\label{rateY}
			\eea
			in the limit where $\epsilon_1$ and $\epsilon_2$ go to zero.
			According to Shannon's source coding theorem (Theorem~\ref{thm:shannon-source-coding}), the rate 
			of equation \eqref{rateX} is sufficient to faithfully decode the string $x^n$,
			i.e. with high probability $\hat{x}^n = x^n$.
			The decoder then has to find the string $y^n$ which is jointly typical with the decoded $\hat{x}^n$
			and this is possible provided the rate $R_Y$ is greater than the conditional entropy $H(Y|X)$.
			The coding scheme for point $\beta$ is analogous.
			

			The multiparty version of the Slepian-Wolf problem was considered in \cite{W74,C75}.
			In the multiparty case, we have not two but $m$ sources $X_1,X_2,\ldots,X_m$ which are to be encoded
			separately and decoded by a common receiver. 
			We want to know the optimal rate tuple $(R_1,R_2,\ldots,R_m)$ at which we can compress the
			corresponding sources such that the information can be recovered faithfully after decoding.
			It is shown in \cite{C75} that the rates have to satisfy the following set of inequalities
			\be	\label{multiparty-slepian-wolf-inequalities}
				\sum_{k \in \K} R_k \ 	>	\ H\!\left(X_{\K} | X_{\Kbar}\right),
			\ee		
			for all $\K \subseteq \{1,2,\ldots,m\}$, $\Kbar = \{1,2,\ldots,m\} \setminus \K$ and
			$X_\K := \{X_i : i \in \K\}$.
			Note that the two-party inequalities \eqref{slepian-wolf-inequalities} are a special case
			of the more general multiparty result.


	\bigskip
	\bigskip
	\section{Quantum information theory}
	
		The fundamental ideas of quantum information theory are analogous to those of classical information theory. 
		In addition to the classical sources and channels, we simply introduce a new set of fundamental 
		building blocks in our studies.
		These \emph{quantum resources} governed by the laws of quantum mechanics can exhibit strange and non-intuitive
		behaviour but can nevertheless be studied with the techniques of information theory.

		\bigskip
		\subsection{Quantum states}
			
			The fundamental principles of quantum mechanics are simple enough to be explained in the space
			available on the back of an envelope, but to truly understand the implications of these
			principles takes years of training and effort.
			We assume the reader is familiar with basic notions of quantum mechanics \cite{sakurai,NC04}. 
			This section will focus on specific notions and notation that are used in quantum information theory.
			
			We will denote quantum systems by uppercase roman letters like $A,B,R$ and the corresponding 
			Hilbert spaces as $\cH^A, \cH^B, \cH^R$ with respective dimensions $d_A,d_B,d_R$.
		    We denote pure states of the system $A$ by \emph{ket}s: $\ket{\ph}^A$
		    and \emph{density matrices} as $\ph^A$.		
			Because of the probabilistic interpretation of quantum mechanics, all kets have unit norm and all
			density matrices are positive and with unit trace.
		    We will refer to both kets and density matrices as \emph{states}.
		    
			We use the partial trace operator to model partial knowledge of a state.
			Given a bipartite state $\rho^{AB}$ shared between Alice and Bob, we say that Alice holds in her lab
			the reduced density matrix: $\rho^A = \Tr_B\rho^{AB}$, where $\Tr_B$ denotes a partial trace over 
			Bob's degrees of freedom.
			In general the state produced in this manner will be \emph{mixed} -- a classical probability distribution
			over states.
			
			Conversely, any mixed state $\sigma^A \in$\footnote{Strictly speaking, we should say 
			$\sigma^A \in D(\cH^A)$ where $D(\cH^A)$ is the set of density matrices over $\cH^A$. 
			We will use this economy of notation consistently.}$\cH^A$
			 can be \emph{purified} to a fictitious 
			larger Hilbert space. 
			That is, we imagine a corresponding pure state $\ket{\sigma}^{AR} \in \cH^A \otimes \cH^R$
			such that taking the partial trace over the $R$ system gives the original state: 
			$\Tr_R\left(\proj{\sigma}^{AR}\right)=\sigma^A$. 
			The purification procedure is often referred to as escaping to the \emph{Church of the larger
			Hilbert space} in literature.

		\bigskip			
		\subsection{von Neumann entropy}
		
			Analogously to classical information theory, we quantify the information content of quantum systems 
			by using an entropy function.
			
										\index{entropy!von-neumann-entropy@$H(A)_\rho$: von Neumann entropy}
			\begin{definition}[von Neumann Entropy] 
				Given the density matrix $\rho^A \in \cH^A$, the expression
				\be
					H(A)_\rho=-\Tr\left(\rho^A\log\rho^A\right)
				\ee
				is known as the \emph{von Neumann entropy} of the state $\rho^A$. 
			\end{definition}
			
			Certain texts use the alternate notation $S(A)_\rho$ for the von Neumann entropy to distinguish it 
			from the classical Shanon entropy, but we choose not to make this distinction here.
			This overloading of notation is warranted since the von Neumann entropy is in fact the 
			Shannon entropy of the eigenvalues of the state. 
			Given the spectral decomposition of the state $\rho^A = \sum_i \lambda_i \proj{e_i}$, we 
			can calculate $H(A)_\rho=-\Tr\left(\rho^A\log\rho^A\right) = - \sum_i \lambda_i \log \lambda_i$.
			The von Neumann entropy of a pure state is zero, since it has only a single eigenvalue.
			
			For bipartite states $\rho^{AB}$ we can also define the quantum conditional entropy
			\be
				H(A|B)_\rho 	:= 		H(AB)_\rho - H(B)_\rho					\label{cond-entrpy} 
			\ee
			where $H(B)_\rho = -\Tr\left(\rho^B\log\rho^B\right)$ is the entropy of the reduced density matrix
			$\rho^B = \Tr_A\!\left( \rho^{AB}\right)$. In the same fashion we can also define the 
			quantum mutual information information
			\be
				I(A;B)_\rho 	:=		H(A)_\rho + H(B)_\rho - H(AB)_\rho 
			\ee
			and in the case of a tripartite system $\rho^{ABC}$ we define the conditional mutual information 
			as 
			\bea
				I(A;B|C)_\rho 	&:=&	H(A|C)_\rho + H(B|C)_\rho - H(AB|C)_\rho \label{cond-mut-info} \\
								&=&		H(AC)_\rho + H(BC)_\rho - H(ABC)_\rho - H(C)_\rho.
			\eea
			
			\noindent It can be shown that $I(A;B|C)$ is strictly non negative for any state $\rho^{ABC}$.
			The formula $I(A;B|C)\geq 0$ can also be written in the form
			\be	\label{strong-subadditivity}
				H(AC) + H(BC) 	\geq	H(C) + H(ABC).
			\ee
			This inequality, originally proved in \cite{LR73}, is called the \emph{strong subadditivity} of von Neumann 
			entropy and forms an important building block of quantum information theory.
			\index{strong subadditivity}

			On the surface, it may appear to the reader that quantum information theory has nothing new to offer except 
			a rewriting of the classical formulas in a new context.
			This observation is highly misleading.
			We present the following example to illustrate some of the new aspects of quantum information theory.
			
			\begin{example}	\label{example:EPR-pair}
				Consider the $\Phi^{+}$\! Bell state 
				\be
					\ket{\Phi}^{AB} = \tfrac{1}{\sqrt{2}}(\ket{00}^{AB}+\ket{11}^{AB}).
				\ee
				This state exhibits a form of quantum correlation called \emph{entanglement} that is fundamentally
				different from classical correlation.
				The associated density matrix is $\Phi^{AB} = \proj{\Phi}^{AB}$, which has
				the reduced density matrices $\Phi^A = \Phi^B = \tfrac{1}{2}(\proj{0} + \proj{1})$.
				
				Next we calculate the entropy of the two subsystems $A$, $B$ and the system as a whole 
				\be
					H(A)_\Phi = 1, 	\qquad
					H(B)_\Phi = 1, 	\qquad
					H(AB)_\Phi = 0,
				\ee
				since $\Phi^A,\Phi^B$ are maximally mixed and $\ket{\Phi}^{AB}$ is pure.
				Using these results, it is now simple to calculate the conditional entropy
				\be	\label{cond-ent-can-be-neg}
					H(A|B)	=	H(AB) - H(B)	= -1 \textup{ [bits]},
				\ee
				and the mutual information
				\be	\label{mutual-inf-can-be-2}
					I(A;B)	=	H(A)+H(B) - H(AB)	= 2 \textup{ [bits]}.
				\ee
			\end{example}
			
			Equation \eqref{cond-ent-can-be-neg} illustrates one of the key differences between classical information
			theory and quantum information theory: the fact that conditional entropy can be negative.
			How can we interpret negative values as uncertainties? Also, it is not immediately clear what we mean by
			conditioning on a quantum system in the first place. These issues will be discussed in some detail in 
			Section~\ref{section:state-merging} where we will give the conditional entropy an operational interpretation.
			
			In classical information theory, the mutual information between two binary sources attains its maximal value
			of $1$ when the two sources are perfectly correlated. As we can see from equation 
			\eqref{mutual-inf-can-be-2}, in the quantum world two qubits can be, in some sense,
			\emph{more than perfectly correlated} and have mutual information as much as $2$ bits!

		\bigskip
		\subsection{Quantum resources}	\label{subsection:quantum-resources}
		
			The current trend in quantum information theory is to look at communication tasks as inter-conversions
			between clearly defined information resources. 
			To render the resource picture generic, we always imagine a scenario in which two localized parties,
			usually called Alice and Bob, want to perform a certain communication task.
			Local computation will be regarded as free of cost in order to focus on the communication aspects 
			of the task. 
			
			An example of a classical communication resource is the \emph{noiseless channel} from 
			Alice to Bob, denoted $[c\to c]$. 
			The symbol $[c \to c]$ represents the ability to send one bit of information from Alice to Bob.
			A related classical resources is the \emph{noisy channel}, denoted $\{c\to c\}$ which is usually 
			modeled as a mapping $\cN^{X \to Y}$, described by a conditional probability $p(Y=y|X=x)$ where
			$X$ is the input variable sent by Alice and $Y$ the random variable received by Bob.
			The noiseless channel $[c\to c]$ is, therefore, a special case of the general channel $\{c\to c\}$ 
			with the identity mapping $\cN = \id^{X \to Y}$ from $X$ to $Y$.
			Another classical resource denoted $[cc]$ represents a random bit shared between Alice and Bob.
			
			Quantum information theory introduces a new set of resources. 
			In analogy to the classical case, we have the \emph{noiseless quantum channel} $[q \to q]$ which 
			represents the ability to transfers one \emph{qubit}, a generic two dimensional quantum system, 
			from Alice to Bob. \index{qubit}
			A \emph{noisy quantum channel}, $\{q \to q\}$, is modeled by a mapping $\cN^{A \to B}$ 
			which takes density matrices in $\cH^A$ to density matrices in $\cH^B$.
			The mapping $\cN$ is a \emph{quantum operation}: a completely positive trace preserving (CPTP) 
			operator \cite{NC04}.
			\index{CPTP@CPTP: completely positive trace preserving}
			
			One key new resource of quantum information theory is the maximally entangled state shared between 
			Alice and Bob 
			\be
				\ket{\Phi}^{AB} = \tfrac{1}{\sqrt{2}}(\ket{00}^{AB}+\ket{11}^{AB}),
			\ee
			which we denote $[qq]$.  
			Note that, since local operations are allowed for free in our formalism, any state 
			$\ket{\Phi'}^{AB} = U^A\! \otimes\!U^B \ket{\Phi}^{AB}$ where $U^A,U^B$ are local unitary operations
			is equivalent to $\ket{\Phi}^{AB}$.
			Entanglement is a fundamental quantum resource because it cannot be generated by local operations and
			classical communication (LOCC). 
			\index{LOCC@LOCC: local operations and classical communication}
			The precise characterization of entanglement has been a great focal point of research in the last decade.
			For an in depth review of the subject we refer the readers to the excellent papers \cite{VP98, HHHH}.
			
			Entanglement forms a crucial building block for quantum information theory because it can be used 
			to perform or assist with many communication tasks.
			In particular, two of the first quantum protocols that ever appeared involve \emph{ebits}, 
			or entangled bits.  \index{ebit}
			The \emph{quantum teleportation} protocol \cite{teleportation} uses entanglement and two bits of classical
			communication to send a quantum state from Alice to Bob
			\be	\tag{TP}	\label{teleport}
				[qq] + 2[c \to c]	\ \ \geq \ \ 	[q \to q],
			\ee
			while the \emph{superdense coding} protocol \cite{superdense} uses entanglement to send two classical
			bits of information with only a single use of a quantum channel
			\be	\tag{SC}	\label{superdense}
				[qq] + [q \to q]	\ \ \geq \ \ 	2[c \to c].
			\ee
			The above \emph{resource inequalities} indicate that the resources on the left hand side can be 
			used to simulate the resource on the right hand side.
			\index{resource!inequality}

			The two protocols \eqref{teleport} and \eqref{superdense} are only the tip of the iceberg: 
			there are many more protocols and fundamental results in quantum information theory that can be 
			written as resource inequalities.
			In Section~\ref{section:family-of-protocols} we will introduce some of them and the relationships 
			that exist between them.
			

		\bigskip		
		\subsection{Distance measures}		\label{subsection:distance-measures}

			In order to describe the ``distance'' between two quantum states we use the notions of 
			\emph{trace distance} and \emph{fidelity}.
			\index{trace distance}
		    The trace distance between quantum states $\sigma$ and $\rho$ is
			\be
				TD(\rho, \sigma) 	:= \|\rho-\sigma\|_1 = \mathrm{Tr}|\rho-\sigma|
			\ee
			where  $|X| = \sqrt{X^{\dagger}X}$.
			
			\index{fidelity}
			The fidelity between two pure states is simply the square of their inner product
			\be
				F(\ket{\varphi}, \ket{\psi}) = \left\vert \braket{\varphi}{\psi} \right\vert^2.
			\ee
			
			The most natural generalization of this notion to mixed states $\rho$, $\sigma$ is
			the formula
	    	\be
	    		F(\rho, \sigma) = \mathrm{Tr}\left(\sqrt{\sqrt{\rho}\sigma\sqrt{\rho}}\right)^2.
	    	\ee
		    Two states that are very similar have fidelity close to 1 whereas states with little similarity 
		    will have low fidelity.
		    
		    Note that some texts, (ex: \cite{NC04}) define the trace distance with an extra normalization factor
			of $\tfrac{1}{2}$ and write the fidelity without the square. These differences of convention 
			do not affect any of our findings but are important to point out to avoid confusion.
		
			
			The trace distance and fidelity measures are related, that is if two states $\rho$ and $\sigma$
			are close in one measure they are also close in the other \cite{Fuchs}. 
			More precisely, the quantities $TD$ and $F$ satisfy the following inequalities
			\begin{align}
				1-\sqrt{F} \ \	&\leq \ \tfrac{1}{2}TD 	\ \leq \ \ \sqrt{1-F},	\\
				1 - TD	   \ \ 	&\leq \ \ \ F 	\ \  	\ \leq \ \ 1 - \tfrac{TD^2}{4}.
			\end{align}
			Thus, if for certain states $F \geq 1-\epsilon$, then $TD \leq 2\sqrt{\epsilon}$. 
			Also, if $TD \leq \epsilon$, then $F \geq 1 - \epsilon$.

		\bigskip				
		\subsection{Ensemble and entanglement fidelity}	\label{subsection:ent-fid}
		
			\index{i.i.d.@i.i.d.: identical, independently distributed}
			The concept of an identical, independently distributed (i.i.d.) source also exists in quantum information
			theory. However, there are a number of ways we can adapt the concept to the quantum setting so some
			clarifications are in order.
			
			An ensemble $\E=\{p_i, \ket{\psi_i} \}$ is a set of quantum states $\ket{\psi_i}$ which occur with 
			probability $p_i$. 
			One way to describe a quantum source is to specify the states $\ket{\psi_i}$ and the
			corresponding probabilities $p_i$ associated with this source. 
			Using this ensemble characterization we can specify what it means to successfully 
			perform a communication protocol with that source.
			Let $\cN^{A \to \widehat{A}}$ with input $\ket{\psi}^A \in \cH^A$ and output 
			$\sigma^{\widehat{A}} \in \cH^{\widehat{A}}$ be the quantum operation associated with the protocol:
			\be
				\cN(\proj{\psi}) = \sigma^{\widehat{A}}.
			\ee
			To measure how faithfully the input state has been reproduced at the output we calculate the 
			input-output fidelity $F(\ket{\psi}^A, \sigma^{\widehat{A}})$. 
			In order to measure how faithfully the source as a whole is reproduced at the output, 
			we have to average over the input-output fidelities of the ensemble
			\be
				\bar{F}\!\left(\E, \cN\right) :=	\sum_i	p_i F(\ket{\psi_i}, \sigma_i), 
				\qquad \sigma_i = \cN(\proj{\psi_i}).
			\ee
			If we want the source to be preserved perfectly then we require $\bar{F}(\E, \cN)=~1$.
			In general, however, we will be content with approximate transmission where
			\be \label{eqn:MixedFidDef}
				\bar{F}\!\left(\E, \cN\right) \geq 1 - \epsilon
			\ee
			for arbitrary small $\epsilon$.
			It turns out that this way of describing the source may not be practical or desirable since 
			it requires a detailed knowledge of the inner workings of the source --- something that is 
			often impossible to obtain even in theory.
			
			The better way to describe a quantum source is specify only the average density operator 
			$\rho = \sum_i p_i \proj{\psi_i}$ for that source. 
			This characterization could be obtained through \emph{state tomography} \cite{NC04} and 
			does not presuppose any knowledge of the ensemble which generates $\rho$.
			This description is more general because the results we obtain for the density matrix $\rho$ will hold
			for \emph{all} ensembles $\{p_i, \ket{\psi_i} \}$ such that $\rho = \sum_i p_i \proj{\psi_i}$.
			
			This also leads us to an alternative and simpler way of judging the success of a quantum protocol 
			that relies on the idea of the \emph{Church of the larger Hilbert space}.
			Let $\ket{\psi}^{AR}$ be a purification of $\rho^A$ to some reference system $R$. 
			This reference system is entirely fiducial and does not participate in the protocol.
			In the larger Hilbert space $\cH^A\otimes\cH^R$ the $\cN^{A \to \widehat{A}}$ operation acts as
			\be
				\cN^{A\to \widehat{A}}\!\!\otimes\!\id^R \!\left(\proj{\psi}^{AR} \right) = \sigma^{\widehat{A}R}.
			\ee
			
			\noindent The operation is shown as a quantum circuit in Figure~\ref{fig:ent_fid}.
		    \begin{figure}[ht]   \begin{center}
		        \input{figures/figureEntFid.pst}             \end{center}
		        \FigureCaptionOpt{Quantum circuit illustrating the concept of entanglement fidelity.}{
			        			  A quantum circuit which shows $\cN$ acting on the $A$ system while
			        			  the reference, $R$, is left unperturbed.}
		        \label{fig:ent_fid}
		    \end{figure}						
			
			For approximate transmission, we now require the fidelity between the pure input state $\ket{\psi}^{AR}$ 
			and the possibly mixed output state $\sigma^{\widehat{A}R}$ to be high
			\be	\label{eqn:EntFidDef}
				F(\ket{\psi}^{AR}, \sigma^{\widehat{A}R}) 
					= 		\bra{\psi^{AR}} \sigma^{\widehat{A}R} \ket{\psi^{AR}}
					\geq 	1 - \epsilon.
			\ee
			Equation \eqref{eqn:EntFidDef} measures the \emph{entanglement fidelity} of the operation: 
			how well the protocol manages to transfers the $R$-entanglement from the $A$ system to the $\widehat{A}$
			system. \index{entanglement!fidelity}
			It can be shown \cite{EntFid} that if the channel $\cN$ has high entanglement fidelity then the 
			average fidelity  $\bar{F}(\E, \cN)$ will also be high for any ensemble $\E$ such that 
			$\rho^A=\sum_i p_i \proj{\psi_i}$.
			In other words, equation \eqref{eqn:EntFidDef} implies equation \eqref{eqn:MixedFidDef}.
			The entanglement fidelity paradigm has the advantage that the input state to the protocol is
			pure, which makes our analysis much simpler.
			Also, in this paradigm we are certain that any correlations the $A$ system might have with other systems 
			are preserved because of monogamy of entanglement.
					
			In the i.i.d. setting, we operate simultneously on $n$ copies of the same input state $\rho^A$. 
			We denote the tensor product of all the input states by $\rho^{A^n} = \rho^A \otimes \cdots \otimes \rho^A$
			($n$-copies). The quantum operation becomes $\cN^{A^n \to \widehat{A}^n}$ and the output state will be
			$\sigma^{\widehat{A}^n}$. 
			The entanglement fidelity 
			\be	\label{eqn:EntFidDefIID}
				F(\ket{\psi}^{A^nR^n}, \sigma^{\widehat{A}^nR^n}) 
					= 		\bra{\psi^{A^nR^n}} \sigma^{\widehat{A}^nR^n} \ket{\psi^{A^nR^n}} 
					\geq 	1 - \epsilon(n),
			\ee	
			is now a function of $n$, the \emph{block size} of the protocol. 
			Thus, in the i.i.d. setting we say the protocol implemented by $\cN$ succeeds when $\epsilon(n) \to 0$ as 
			$n \to \infty$. 
			More formally, for any required precision $\epsilon_0$, there exists an $N(\epsilon)$ such that for all 
			$n \geq N(\epsilon)$,  there exist $n$-dependent maps $\cN$ such that $\epsilon(n) < \epsilon_0$ in 
			equation~\eqref{eqn:EntFidDefIID}.


\chapter{Results in quantum information theory} 	\label{chapter:results-in-QIT}
	
	This chapter is dedicated to four landmark results in quantum information theory.
	The first of these is Schumacher compression, the quantum version of source coding \cite{Sc95}.
	The second is the 
	\emph{resource framework of quantum information theory} \cite{DHW05b},
	which defines rigorously the properties of quantum protocols and discusses the relationships between 
	them \cite{DHW04}.
	Then, in section~\ref{section:state-merging}, we focus our attention on one protocol 
	for compression of quantum information with side information known as \emph{state merging} \cite{HOW05,HOW05b}.
	Finally, in the last section of this chapter, we discuss in detail the fully quantum Slepian-Wolf (FQSW)
	protocol for state transfer and simultaneous entanglement distillation. 
	To a large extent, the multiparty results in this thesis are a direct generalization of the two-party 
	FQSW protocol, therefore, section~\ref{section:FQSW} is of central importance to the remainder of the argument.

	\bigskip
	\bigskip
	\section{Schumacher compression}		\label{section:schumacher-compression}
		
		If classical information theory is 60 years old \cite{S48} then quantum information theory must be 12 years old.
		Indeed, we can say that Schumacher laid the foundations of quantum information theory with his 1995 
		paper \cite{Sc95} where he showed that the von Neumann entropy, $H(\rho)$,
		plays the analogous role of Shannon entropy for quantum systems. Namely, it has operational interpretation 
		as the number of qubits necessary to convey the information from a quantum source $\rho$.

		\subsection{Typical subspace}		\label{section:quantum-typical}
	
			The notion of a typical set (section~\ref{section:typical-sets}) can easily be generalized to
			the quantum setting.
			Consider a source which produces many copies of the state $\rho^A$ which has spectral decomposition
			$\rho^A = \sum_i \lambda_i \proj{i}$. 
			
														\index{i.i.d.@i.i.d.: identical, independently distributed}
			In the the i.i.d. regime, the state produced by the source is given by 
			$\rho^{A^n} = \rho^{A_1}\otimesc\rho^{A_2}\!\otimes \cdots \otimes \!\rho^{A_n}$
			which can be written as 
			\be
			\begin{aligned}
			 	\rho^{A^n}	&=	\sum_{i^n}	\lambda_{i_1}\!\cdots\!\lambda_{i_n}
			 								\ \proj{i_1}^{A_1}\otimesc\cdots\otimesc \proj{i_n}^{A_n} \\
						   				&=	\sum_{i^n}	\lambda_{i^n}	\proj{i^n}^{A^n}.
			\end{aligned}
			\ee
			
			We now define the \emph{typical projector} as follows
			\begin{eqnarray}	
				\ptyp			 	&=&	\sum_{i^n \in \typ} 	\proj{i_1}^{A_1}\otimesc	
																\proj{i_2}^{A_2}\otimesc 
																\ldots \proj{i_n}^{A_n}  \nonumber \\
									&=&	\sum_{i^n \in \typ} 	\proj{i^n}^{A^n},
			\end{eqnarray}
			where we sum over all the typical sequences $\typ$ with respect to the classical probability distribution
			$p(i) := \lambda_i$. 
			
			We call the support of $\ptyp$, the \emph{typical subspace} of $\cH^{A^n}$ associated with~$\rho^A$.
			\index{typical!subspace}
			The typical subspace, by its construction, inherits the characteristics of the typical set.
			Indeed, $\ptyp$ has the following properties 
			
			\smallskip
			{\renewcommand{\labelenumi}{(\roman{enumi})\ \ \ } 
			\begin{enumerate}
				\item 	$\Tr\left[ \rho^{A^n} \ptyp \right] > 1 - \delta \qquad 
							\forall \delta, \epsilon > 0$\ \  and $n$ sufficiently large.

				\item 	$\Tr\left[\ptyp\right] \leq 2^{n[H(A)_\rho +\epsilon]} \qquad \ 
							\forall \epsilon >0$\ \ and $n$ sufficiently large.
			\end{enumerate} 	}	
			\smallskip
			
			Property (i) says that, for large $n$, most of the states produced by the source will lie mostly inside 
			the typical subspace.
			Property~(ii) is a bound on the size of the typical subspace which follows from the classical bound on 
			the size of the typical set $\typ$.
			These two properties are at the heart of our ability to compress quantum information.


		\bigskip
		\subsection{Quantum compression}

			Analogously to the classical case, we have the notion of a quantum compression rate.
			In the quantum regime, we use the entanglement fidelity (see Section \ref{subsection:ent-fid}) 
			to measure how well the state is reproduced after decoding.
													
													\index{rate!quantum compression}
			\begin{definition}[Quantum compression rate]
				We say a compression rate $R$ for the source $\rho^{A}$ is achievable if for all $\epsilon$, 
				there exists $N(\epsilon)$ such that for $n > N(\epsilon)$, there exist maps:
				\bea
					\E_n:\cH^{A^n}	 	& \rightarrow & 	\m \qquad |\m| = 2^{nR} \\
					\cD_n:\m 	\ \ 	& \rightarrow &		\cH^{\widehat{A}^n}
				\eea
				such that the purification $\ket{\psi}^{A^nR^n}$ of $\rho^{A^n}$ satisfies  
				\be
					F(\ket{\psi}^{A^nR^n}, \sigma^{\widehat{A}^nR^n}) 
						= \bra{\psi} \sigma^{\widehat{A}^nR^n} \ket{\psi}^{A^nR^n} > 1 - \epsilon.
				\ee
				where $\sigma^{\widehat{A}^nR^n} =  \cD_n\!\circ\!\E_n\!\otimes\!\id^{R^n} \left(\proj{\psi}\right)$.
			\end{definition}
			
			\index{compression!quant@quantum, Schumacher}
			\begin{theorem}[Schumacher noiseless coding]
				An i.i.d. quantum source $\rho^A$ can be compressed at a rate $R$ if $R>H(A)_\rho$ and cannot
				if $R<H(A)_\rho$.
			\end{theorem}
			
			The idea behind the Schumacher compression result is simple.
			We encode by performing the measurement 
			\be
				M_\E = \{ \ptyp, \ \id\! -\! \ptyp \}.
			\ee
			If \ptyp occurs, we keep this state since it is typical. 
			Otherwise, if $(\id - \ptyp)$ occurs, we replace the state with some fixed state 
			$\ket{\textrm{err}}$ as an indicator that an error has occurred. 
			The decoding operation $\cD_n$ is the identity operation.
			Property~(i) from the previous section guarantees that the probability of error tends to zero 
			when $n$ becomes large.
			Also, since we only send states within the typical subspace, Property~(ii) gives us a bound 
			on the amount of quantum information necessary to convey this state.

			Note that the compression protocol described above works both for scenarios where the mixed state
			is obtained from a stochastic average over pure states $\rho^A = \sum_i p_i \proj{\psi_i}^A$ 
			and scenarios where the density matrix is part of a larger pure state $\rho^A = \Tr_E \proj{\Psi}^{AE}$.

	\bigskip
	\bigskip
	\section{Quantum protocols as resource inequalities}		\label{section:family-of-protocols}

		Most of the old results of quantum information theory form a loose collection of coding theorems,
		each of them with applications only to one specific communication task.
		Recently, there has been a push to organize these results into a unified framework of resource
		inequalities~\cite{BBPS,DHW04,DHW05b,FQSW}.  \index{resource!inequality}
		A resource inequality is a quantitative statement regarding inter-conversions between clearly
		defined \emph{generic} information resources.
		The key benefit of such a framework is that, like {\sc Lego} blocks, we can build one communication 
		protocol based on another, and generally work at a higher level of abstraction then is possible
		when working with the specifics of each protocol.

		%
		%

		\bigskip
		\subsection{The framework}

			A unified framework for both classical and quantum information theory was developed in \cite{DHW05b}.
			The notions ``resource'' and ``protocol'' are clearly defined as well as the rules for combining and 
			composing them.
			In particular, this framework deals with the class of bipartite, unidirectional communication tasks 
			involving memoryless channels and sources in the i.i.d regime.
								\index{protocol!framework}
			
			Borrowing from the cryptography heritage, the two main participants in the protocols are called
			Alice and Bob. Alice is usually the sender, and performs some encoding operation while Bob
			does the decoding.
			Additionally, the framework introduces two novel participants \emph{Eve} and the \emph{Reference}.
			We use Eve to model information lost to the environment in a noisy channel.
			The reference $R$ is a fiducial purification system which allows us to deal with mixed states in a 
			simple manner as discussed in Section~\ref{subsection:ent-fid}.
			Most important of all, the framework introduces the \emph{Source}, which produces some state $\rho^S$
			and distributes it to the participants before the beginning of the protocol.

			In Section~\ref{subsection:quantum-resources}, we introduced some of the resources of 
			information theory like the noiseless classical channel $[c \to c]$, noisy classical 
			channel $\{c \to c\}$ and the quantum equivalents $[q \to q]$ and $\{q \to q\}$.
			In order to be more precise, we sometimes use a different notation for noisy channels
			\be
				\{q \to q\} \equiv	\rsrc{ \cN }
			\ee
			which explicitly shows the map $\cN$ associated with that channel.
			Note that the angle brackets $\rsrc{ . }$ indicate that we are working in the asymptotic regime
			of many copies of the resource: $\rsrc{\cN} \sim \cN^{\otimes n}$ and 
			$\rsrc{\rho^{AB}} \sim (\rho^{AB})^{\otimes n}$.
			We will denote a \emph{relative resource} as $\rsrc{\cN:\rho^A}$ which is a channel 
			guaranteed to  behave as the channel $\cN$ provided the input state is exactly $\rho^A$. 
			\index{resource!relative}
			
			As a first example of the protocol framework, consider the Schumacher compression result from 
			the previous section. It can be represented by the following resource inequality 
			\be	\label{RI:schumacher-compression}
				\left(H(B)_\sigma+ \delta \right)[q \to q] 	\geq 	\rsrc{ \id^{A \to B}  : \rho^A}
			\ee
			for any $\delta \geq 0$ and where $\sigma^B := \id^{A \to B}(\rho^A)$.
			The above equation indicates that $(H(B) + \delta)$ qubits are sufficient to accurately convey
			the information contained in the state $\rho^A$ to another party.

			%
			%
			%
			%
			%
			%
			%


		\bigskip
		\subsection{The family of quantum protocols}

			Many protocols of quantum information theory deal with the conversion of some
			noisy resource into the corresponding noiseless version possibly with the use of some auxiliary resources.
			It turns out that many of these protocols are related and it is sufficient to prove two protocols of this
			type and all other protocols follow as simple consequences when we apply teleportation \eqref{teleport}
			or superdense coding \eqref{superdense} either before or after these protocols \cite{DHW04}.
			
			The two protocols which generate all the others of this ``family tree'' are called the mother and 
			father protocols.
			The mother protocol takes the static resource $\rsrc{ \rho^{AB} }$ and some quantum communication
			to distill maximally entangled bits. The resource inequality is
			\be	\tag{\female}	\label{RI:mother}
				\rsrc{ \rho^{AB} } + \frac{1}{2}I(A;R)_\psi[q \to q] \geq \frac{1}{2}I(A;B)_\psi[qq],
			\ee	\index{protocol!mother}
			where the entropies are taken with respect to a purification $\ket{\psi}^{ABR}$ of $\rho^{AB}$.
			The mother protocol can be used to derive three ``children''.
			The first of these is entanglement distillation, also known as the hashing inequality~\cite{hashing}.
			We start with equation \eqref{RI:mother} but implement the $[q \to q]$ term as teleportation
			\be
					\frac{1}{2}I(A;R) \bigg( [qq] +  2[c \to c] 	\geq 	[q \to q] \bigg).
			\ee
			After canceling some of the $[qq]$ terms on both sides we obtain 
			\be		\label{RI:hashing}
				\rsrc{ \rho^{AB} } + I(A;R)_\psi[c \to c] \geq I_c(A\rangle B)_\psi[qq],
			\ee
			where $I_c(A\rangle B) = \tfrac{1}{2}I(A;B) - \tfrac{1}{2}I(A;R) = H(B) - H(AB)$. 
			The mother inequality can also be used to derive noisy versions of the teleportation \cite{DHW04}
			and superdense coding protocols \cite{noisySDC}.

			The father protocol takes the dynamic resource of a noisy quantum channel $\rsrc{ \cN^{A \to B} }$
			and some additional entanglement to simulate a noiseless quantum channel.
			Consider a setup where we send half of the state $\ket{\phi}^{AR}$ through the channel $\cN$,
			which we model as an isometric extension $U^{A \to BE}_\cN$ to an environment $E$.
			The resulting state is $\ket{\psi}^{BER} = U^{A \to BE}_\cN\otimesc\id^R \ket{\phi}^{AR}$ 
			and the resource inequality  is 
			\be	\tag{\male}		\label{RI:father}
				\rsrc{ \cN^{A \to B} }	+ \frac{1}{2}I(R;E)_\psi[qq]  \geq	\frac{1}{2}I(R;B)_\psi[q \to q].
			\ee	\index{protocol!father}
			Using the father inequality and the superdense coding result \eqref{superdense}, we can derive 
			the formula for the entanglement-assisted classical capacity of a quantum channel \cite{BSST99}
			\be
				\rsrc{ \cN^{A \to B} : \rho^A } + H(R)_\psi[qq] \geq	 I(R;B)_\psi[c  \to c].
			\ee
			More importantly, we can obtain the important quantum capacity result, the LSD Theorem~\cite{lloyd,shor,D03},
			named after Lloyd, Shor and Devetak:
			\be
				\rsrc{ \cN }	\geq	I_c(R\rangle B)_\psi[q \to q].
			\ee


			Furthermore, it turns out that equation \eqref{RI:mother} and \eqref{RI:father} are related.
			We can obtain one from the other by replacing dynamic resources with static resources and 
			adjusting for the definitions of $A$ and $R$.
			This duality could be a mere coincidence or it could be indicative of some hidden structure.			
			We will see in section~\ref{section:FQSW} that in fact there exists an even bigger mother!
			The FQSW protocol, sometimes called ``the mother of all protocols'', is a quantum
			protocol that generates both the mother and father protocols as well as many other 
			protocols that were not part of the original family tree \cite{FQSW,triangle}.

	\bigskip
	\bigskip
	\section{State merging}							\label{section:state-merging}
		
		Consider a setup where Alice and Bob share the state $\rho^{AB} = \Tr_R \proj{\psi}^{ABR}$.
		We would like to know how much quantum information Alice needs to send to Bob to \emph{merge} her 
		part of the state into Bob's. The problem is illustrated graphically in Figure~\ref{fig:STATEMERGING}.
		\bigskip
		
	    \begin{figure}[ht]   \begin{center}
	        \input{figures/figureSTATEMERGINGdiagram.pic}             \end{center}
	        \FigureCaptionOpt{Diagram of the state merging protocol.}{
		        Pictorial representation of the state merging protocol.
		        Alice's part of $\ket{\psi}^{ABR}$ is merged with Bob's part.
		        In the end, the purification of $R$ is held entirely in Bob's system.}
			\label{fig:STATEMERGING}
	    \end{figure}		
		    
		In the limit of many copies of the state, the rate at which Alice needs to send quantum information 
		to Bob is given by the formula
		\be
			R 	\ >	\ 	H(A|B)_\rho,
		\ee
		provided classical communication is available for free.
		The primitive which optimally achieves this task is called the \emph{state merging protocol} \cite{HOW05,HOW05b}.
		We will discuss this protocol in some detail in section~\ref{subsection:merging-protocol} but before that we
		dedicate some time to the quantum conditional entropy.
		
		\bigskip
		\subsection{Quantum conditional entropy}	\label{subsection:quantum-conditional-entropy}
		
			The classical notion of conditional entropy $H(X|Y)$ is the amount of communication needed to
			convey the information content of the source $X$ given knowledge of the variable $Y$ at the decoder. 
			As we saw in section~\ref{subsection:classical-slepian-wolf}, the conditional entropy is naturally suited 
			to application in the Slepian-Wolf problem of distributed compression.
			
			When we try to adapt the conditional entropy to the quantum world, we run into a number of 
			conceptual difficulties. 
			Indeed, in order to define $H(A|B)_\rho$ for a quantum state $\rho^{AB}$ we need to replace 
			the classical notion of a conditional distribution with some concept better suited 
			to density matrices \cite{adami1,CondMutInfo}.
			A more pragmatic approach is to simply mimic the form of equation~\eqref{conditional-classical-formula}
			from Section~\ref{subsection:multiple-sources} 
			and write the conditional entropy as a difference of two regular entropies 
			\be
					H(A|B)_\rho		:=		H(AB)_\rho	-	H(B)_\rho.
			\ee
			In this way, we obtain a \emph{formula} for the conditional entropy but still lack an \emph{interpretation}.
			The situation is complicated further by the fact that the quantum conditional entropy can take on 
			negative values seemingly indicating that it is possible to know more about the global state than about 
			a part of it!
			Also to be explained is the relation between the negative values of the conditional entropy 
			and the presence of quantum entanglement as indicated by the entropic Bell inequalities~\cite{adami2}.

			The interpretation issues around the quantum conditional entropy were finally settled in a satisfactory
			manner in two recent papers \cite{HOW05,HOW05b}, in which the quantum conditional entropy is given an
			operational interpretation in terms of the state merging protocol.

		\bigskip
		\subsection{The state merging protocol}		\label{subsection:merging-protocol}
									
									\index{protocol!state merging}		
			Consider a state $\rho^{AB}$ shared between Alice and Bob and a purification of that 
			state $\ket{\psi}^{ABR}$.
			We want to send Alice's part of the state to Bob by using an unlimited amount of classical communication 
			and as little quantum communication as possible.
			Let $\Phi_K \in \cH^{A_0B_0}, \Phi_L \in \cH^{A_1B_1}$ be two maximally entangled states of rank 
			$K$ and $L$ respectively.
			The state merging protocol takes as inputs the state $\proj{\psi}^{ABR}$ and $\log K$ ebits in the form
			of $\Phi_K$ and applies the quantum operation 
			$\cM\!:\!AA_0\otimes BB_0 \to A_1\otimes B_1 \widehat{B}B$ to produce a state 
			\be
				\sigma^{A_1B_1 \widehat{B}BR} 
					= (\cM \otimes \id^R)\!\left(	\proj{\psi}^{ABR} \otimes \Phi_K \right).
			\ee
			We want the state $\rho^{AB}$ to be transferred entirely to Bob's lab: 
			$\sigma^{\widehat{B}B} \approx \rho^{AB}$.
			In addition, $\log L$ ebits are generated by the protocol if $\sigma^{A_1B_1} \approx \Phi_L$.
			More precisely, we measure the success of the protocol by the entanglement fidelity
			\be
				F\!\left( \sigma^{A_1B_1 \widehat{B}BR}, \Phi^{A_1B_1}_L \otimes \proj{\psi}^{ABR} \right)
				\geq	1 - \epsilon.
			\ee

			In our original formulation of the state transfer task we asked how much quantum \emph{communication}
			from Alice to Bob is necessary, yet in the above formulation we only speak of 
			\emph{entanglement} being consumed and generated. 
			This is so because the two resources become equivalent when unlimited classical communication is allowed:
			\be
				[qq] \equiv [q\to q]		\qquad \qquad	\textrm{(free classical communication)}.
			\ee
			Thus, we can say that $\Phi_K$ is the entanglement consumed by the protocol while $\Phi_L$ is the
			entanglement generated. 
			In the i.i.d. regime, where $\ket{\psi}^{ABR} = \left(\ket{\ph}^{ABR}\right)^{\otimes n}$, we define 
			the \emph{entanglement rate}
			\index{i.i.d.@i.i.d.: identical, independently distributed}
			\be
				R =	\frac{1}{n}\left(\log K - \log L \right),
			\ee
			which can take on both positive and negative values.
			When $R > 0$, the entanglement resource has been consumed by the protocol, 
			but when $R < 0$ the protocol is actually generating entanglement as stated in the following theorem.
			
			\begin{theorem}[Quantum state merging \cite{HOW05b}]	\label{thm:state-merging}
				For a state $\rho^{AB}$ shared by Alice and Bob, the entanglement cost of merging is
				equal to the quantum conditional entropy $H(A|B) = H(AB) - H(B)$.
				When $H(A|B)$ is positive, merging is possible only if $R > H(A|B)$ ebits per input copy are
				provided.
				When $H(A|B)$ is negative, the merging is possible by local operations and classical communication
				and moreover, $R < - H(A|B)$ maximally entangled states are obtained per input copy.
			\end{theorem}
			
			We can express the state merging protocol as a resource inequality
			\be	\label{RI:state-merging}
				\rsrc{ U^{S \to AB} : \rho^S } \ + \  H(A|B)_\psi[q \to q]	\ \ 
				\geq \ \
				\rsrc{\id^{S \to B}  : \rho^S} \textup{\ \ \ \ (free $[c \leftrightarrow c]$)}
			\ee
			where $U^{S \to AB}$ is an isometry, $\rho^{AB} = U^{S \to AB}(\rho^S)$ that splits the state 
			produced by the source between Alice \& Bob while $\id^{S\to B}$ gives the state directly to Bob.
			The net effect of $\rsrc{U^{S \to AB} : \rho^S }$ on the left hand side and $\rsrc{\id^{S \to B}:\rho^S}$
			on the right, is the state transfer resource informally defined 
			\be
				\rsrc{\id^{A \to \widehat{B}} : \rho^{AB}} 
					:= 	\rsrc{\id^{S \to B}:\rho^S} - \rsrc{U^{S \to AB} : \rho^S }.
			\ee
			According to equation \eqref{RI:state-merging},  this resource can be an asset or a liability 
			depending on the sign of $H(A|B)$.

			%
			%
			
			The state merging protocol has numerous applications.
			It can be used to study the quantum capacity of multiple access channels, 
			entanglement distillation\cite{DW05},  entanglement of assistance\cite{EntAss} and distributed compression.
			The latter of these is of particular relevance to the subject of this thesis since it is
			a quantum generalization of the Slepian-Wolf problem discussed in 
			section~\ref{subsection:classical-slepian-wolf}.
			Alice and Bob have to individually compress their shares of a state $\rho^{AB}$ and transmit them to 
			common receiver, Charlie.
			We allow unlimited classical communication and rates $R_A$, $R_B$ of quantum communication to Charlie.
			The rate region for quantum distributed compression is given by the inequalities
			\be					\label{state-merging-inequalities}
				\begin{aligned}
					R_A 		&>	 	\ H(A|B)_\rho,	\\
					R_B 		&>		\ H(B|A)_\rho,	\\
					R_A+R_B		&>		\ H(AB)_\rho.	
				\end{aligned}
			\ee	
						\index{rate!region}
			
			The rates for quantum distributed compression \eqref{state-merging-inequalities} should be compared 
			with the classical distributed compression rates \eqref{slepian-wolf-inequalities}.
			This is an instance of a general trend in quantum information theory: 
			if classical communication is available for free, the solution to the quantum analogue
			of a given classical communication task is identical the classical solution up to replacement
			of Shannon entropies by von Neumann entropies.
			Many times, however, this ``$H$ goes to $S$ rule'' is only skin deep 
			and sometimes it does not hold at all.
			
			In the next section we will give the details of the fully quantum Slepian-Wolf (FQSW) protocol, 
			which is a generalization of state merging where no classical communication is allowed.
			In the light of this, a detailed proof of the state merging protocol has been omitted for the 
			sake of brevity and since it follows from the more powerful FQSW protocol.

	\bigskip
	\bigskip
	\section{The fully quantum Slepian-Wolf protocol} 		\label{section:FQSW}

		The fully quantum Slepian-Wolf protocol \cite{FQSW} is a procedure for simultaneous quantum state transfer 
		and  entanglement distillation.  \index{FQSW@FQSW: fully quantum Slepian-Wolf}
		It can be thought of as the quantum version of the classical Slepian-Wolf protocol but, unlike the
		state merging protocol considered above, no classical communication is allowed. 
	    This FQSW protocol generates nearly all the other protocols of quantum information theory 
	    as special cases, yet despite its powerful applications it is fairly simple to implement.
		
		\bigskip
	    \begin{figure}[hb]   \begin{center}
	        \input{figures/figureFWSQdiagram.pic}             \end{center}
	        \FigureCaptionOpt{Diagram of the $ABR$ correlations before and after the FQSW protocol.}{
		        Diagram representing the $ABR$ correlations before and after the FQSW protocol.
		        Alice manages to decouple completely from the reference $R$.
		        The $\widehat{B}$ system is isomorphic to the original $AB$: it is the purification of~$R$. }
	        \label{fig:FQSW}
	    \end{figure}
	    	
	    The state $\ket{\psi}^{ABR}=\left(\ket{\ph}^{ABR}\right)^{\otimes n}$ is shared between Alice, Bob and 
	    a reference system $R$.
	    The FQSW protocol describes a procedure for Alice to transfer her $R$-entanglement to Bob while
	    at the same time generating ebits with him.
	    Alice can accomplish this by encoding and sending part of her system, denoted $A_1$, to Bob.
	    The state after the protocol can approximately be written as 
	    $\ket{\Phi}^{A_2\widetilde{B}}(\ket{\ph}^{R\widehat{B}})^{\otimes n}$, 
	    where the systems $\widetilde{B}$ and $\widehat{B}$ are held in Bob's lab while $A_2$ remains with Alice.
	    The additional product, $\ket{\Phi}^{A_2\widetilde{B}}$, is a maximally entangled state shared between 
	    Alice and Bob. 
	    Figure \ref{fig:FQSW} illustrates the entanglement structure before and after the protocol.

		\bigskip
		\subsection{The protocol} \label{subsec:FQSW-protocol}
		
					\index{protocol!FQSW: fully quantum Slepian-Wolf}
			The protocol relies on an initial compression step and the mixing effect of random unitary operations 
			for the encoding. 
			We assume that, before the start of the protocol,  Alice and Bob have pre-chosen a random unitary operation
			$U_A$. 
			Equivalently, they could have shared random bits which they use to locally generate the same unitary
			operation.
	    	
		    The protocol, represented graphically in Figure \ref{fig:AliceActions}, consists of the following steps:
		    \begin{enumerate}
		        \item   Alice performs Schumacher compression on her system $A$ to obtain the output system $A^S$.
		
		        \item   Alice then applies a random unitary $U_A$ to $A^S$.
		
		        \item   Next, she splits her system into two parts: $A_1A_2=A^S$ with $d_{A_1} = 2^{nQ_A}$ and
		                \be
		                    Q_A  >  \frac{1}{2}I(A;R)_\ph.
		                \ee
		                She sends the system $A_1$ to Bob.
		
		        \item   Bob, in turn, performs a decoding operation $V_B^{ {A_1B} \to \widehat{B}\widetilde{B} }$ 
		        		which splits his system into a $\widehat{B}$ part purifying $R$ and a $\widetilde{B}$ part 
		        		which is fully entangled with Alice.
		    \end{enumerate}
		
		    \begin{figure}[ht]  \begin{center}
		        \input{figures/figureAliceActions.pic} \end{center}
		        \FigureCaptionOpt{Circuit diagram for the FQSW protocol.}{
		            Circuit diagram for the FQSW protocol. 
		            First we Schumacher compress the $A$ system, then we apply the random unitary encoding $U_A$.
		            At the receiving end Bob applies a decoding operation $V$. }
		        \label{fig:AliceActions}
		    \end{figure}
	
		    The best way to understand the mechanism behind this protocol is by thinking about destroying correlations.
		    If, at the end of the protocol, Alice's system $A_2$ is nearly decoupled from the reference in the sense that
		    $\sigma^{A_2 R} \approx \sigma^{A_2}\otimes \sigma^{R}$, then Alice must have succeeded in
		    sending her $R$ entanglement to Bob because it is Bob alone who then holds the $R$ purification.
		    We can therefore guess the lower bound on how many qubits Alice will have to send before she can
		    decouple from the reference. Originally, Alice and R share $I(A;R)_\ph$ bits of information per copy of
		    $\ket{\ph}^{ABR}$.   
		    Since one qubit can carry away at most two bits of 
			quantum mutual information, this means that the minimum rate 
		    at which Alice must send qubits to Bob is
		    \be
		        Q_A     >  \frac{1}{2}I(A;R)_\ph. 
		    \ee
	
		    It is shown in \cite{FQSW} that this rate is achievable in the limit of many copies of the state.
		    Therefore the FQSW protocol is optimal for the state transfer task.
			More formally the decoupling process is described by the following theorem:
			
			\bigskip
			
			\begin{theorem}[One-shot decoupling theorem from \cite{FQSW}] \label{thm:oneShotMother} \ \\
			Let $\sigma^{A_2 R}(U) = \Tr_{A_1}[ (U \otimes \id^R) \psi^{A^{\!S}R} (U^\dagger \otimes \id^R) ]$ 
			be the state remaining on $A_2 R$ after the unitary transformation $U$ has been applied to $A^S = A_1 A_2$.
			Then
			\be \label{eqn:oneShotMother}
			   \int_{\UU(A)} \Big\| \sigma^{A_2 R}(U) - \frac{\id^{A_2}}{d_{A_2}} \otimes \sigma^R \Big\|_1^2 \, dU
			   \leq
			   \frac{d_{A^{\!S}} d_R}{d_{A_1}^2} \Tr[(\psi^{A^{\!S}R})^2] .
			\ee
			\end{theorem}
			
			\bigskip
			
			This theorem quantifies how close to decoupled the $A_2$ and $R$ systems are if a random unitary operation is
			applied to the $A^S=A_1A_2$ system. There are several important observations to make in relation to the above
			inequality.
			First, we note that for a given state $\ket{\psi}^{ABR}$, the dimensions of the systems $A^S$ and $R$ 
			as well as the purity $\Tr[(\psi^{A^SR})^2]$ are fixed numbers over which Alice has no control. 
			Alice can, however, choose the dimension of the subsystem she sends
			to Bob, $d_{A_1}$, and influence how decoupled she is from the reference.
			By making making $d_{A_1}$ sufficiently large, Alice can thus make the right hand side of 
			\eqref{eqn:oneShotMother} tend to zero.

			Second, the fact that Alice holds something very close to a maximally mixed state 
			$\id/d_{A_2}$ indicates that Bob can, by an appropriate choice of decoding operation $V_B$, 
			establish a maximally entangled state $\ket{\Phi}^{A_2\widetilde{B}}$ with Alice.
			These ebits generated between Alice and Bob are a useful side-effect of the protocol that is
			similar to the entanglement generated by the state merging protocol.

			All that remains now is to specify $d_{A_1}$, the dimension of the system sent to Bob, 
			in terms of entropic quantities of the input state.
			This can be done in the the limit where $n$, the number of copies of the state goes to infinity.
			Using the properties of typical subspaces, we can we can make the right hand side of equation
			\eqref{eqn:oneShotMother} tend to zero provided the rate $Q_A \equiv \frac{1}{n}\log d_{A_1}$ 
			satisfies \cite{FQSW}:
			\be	\label{eqn:FQSW}
			   	Q_A     \geq  \frac{1}{2}I(A;R)_\ph + \delta
			\ee
			for any $\delta>0$.

		\bigskip
		\subsection{The FQSW resource inequality}	\label{subsec:FQSW-RI}
	
			In the spirit of section~\ref{section:family-of-protocols} above, we can succinctly express the 
			effects of the fully-quantum Slepian-Wolf protocol as a resource inequality
			\be	\label{RI:FQSW}	
			    \rsrc{U^{S \to AB}: \ph^S } + \tfrac{1}{2}I(A;R)_\ph[q \to q]
			    \ \geq \ 
			    \tfrac{1}{2}I(A;B)_\ph[qq] + \rsrc{ \id^{S \to \widehat{B}} : \ph^S }
			\ee
			which is read:      {\it   Given the state $\ket{\ph}^{ABR}$ and $\tfrac{1}{2}I(A;\!R)$ qubits of
			                           communication from Alice to Bob we can obtain the state 
			                           $\ket{\ph}^{R\widehat{B}}$ while also purifying $\tfrac{1}{2}I(A;B)$ ebits. }

			As previously announced, the FQSW protocol is more powerful than the state merging protocol of
			section~\ref{section:state-merging} since it generates it as a special case.
			Indeed, when we implement the quantum communication $[q \to q]$ of equation~\eqref{RI:FQSW} as 
			teleportation according to equation~\eqref{teleport} 
			\be	\label{recycling-in-FQSW}
				\tfrac{1}{2}I(A;R)[qq]	+ I(A;R)[c \to c] 	\ \geq \ \tfrac{1}{2}I(A;R)[q \to q].
			\ee
			We now ``recycle'' the entanglement produced by the protocol. 
			The factor in front of $[qq]$ is is going to be $\tfrac{1}{2}I(A;R) - \tfrac{1}{2}I(A;B) = H(A|B)$ and the 
			overall resource inequality becomes  
			\be
			    \rsrc{U^{S \to AB}: \ph^S } + H(A|B)_\ph[qq] + I(A;R)_\ph[c \to c]
			    \ \geq \ 
			    \rsrc{ \id^{S \to \widehat{B}} : \ph^S },			
			\ee
			which is exactly the state merging resource inequality \eqref{RI:state-merging}, when we also
			account for the classical communication cost.

			The FQSW inequality generates the mother inequality \eqref{RI:mother} by discarding the additional 
			resource $\rsrc{ \id^{S \to \widehat{B}} : \ph^S }$ on the right hand side.
			Moreover it was recently shown that by the \emph{source-channel duality}, the FQSW protocol can 
			be used to generate the father protocol \eqref{RI:father} 
			and by \emph{time reversal duality} the FQSW protocol leads to the fully quantum reverse 
			Shannon (FQRS) protocol~\cite{triangle}. 
			Other notable results related to the FQSW protocol are the recent results for 
			broadcast channels \cite{DH06}, and the generalization of the FQSW task called 
			\emph{quantum state redistribution}, which uses side information both at the encoder and the
			decoder~\cite{DY06,DY07}.
			

		\bigskip
		
		In addition to the its powerful protocol generating faculties, the FQSW protocol has applications 
		to the distributed compression problem for quantum systems.
		Indeed, the original FQSW paper \cite{FQSW} partially solves the distributed compression problem in the two-party
		case by providing upper and lower bounds on the set of achievable rates.
		In Chapter~\ref{chapter:multiparty-distributed-compression} we will present our results on 
		the multiparty version of the same problem.
		For the sake of continuity, the reader may wish to skip 
		Chapter~\ref{chapter:multiparty-quantum-information} on a first reading 
		of the thesis since it is a self-contained exposition on the multiparty squashed entanglement, 
		which only comes into play relatively late in the distributed compression chapter.



\chapter{Multiparty quantum information}	\label{chapter:multiparty-quantum-information}

	Many of the protocols of information theory deal with multiple senders and multiple receivers.
	As a whole, however, \emph{network information theory}, the field which studies general multiparty 
	communication scenarios is not yet fully developed even for classical systems \cite{CT91}.
	Quantum network information theory, which deals with quantum multipartite communication, is also under
	active development \cite{butterfly,DH07,AHS07} and, thanks to the no-cloning properties of quantum information,
	sometimes admits simple solutions \cite{butterfly}.
	On the other hand, a full understanding of quantum network theory will require a precise characterization 
	of multiparty entanglement, a task which is far from completed \cite{LSSW,DCT99,CKW00,BPRST99}.
	Nevertheless, we can hope that years from now we will have a rigorous and complete theory of multiparty 
	information theory in the spirit of the two-party protocols framework \cite{DHW05b}.

	One step toward the development of a multiparty information theory would be to generalize the concept of mutual
	information $I(A;B)$ to more than two parties.
	The mutual information, the information that two systems $A$ and $B$ have in common, can be written as
	\be	\label{protoTypeOne}
		I(A;B)	=  H(A)	- H(A|B).
	\ee
	The above formula is interpreted as a reduction of the total uncertainty of $A$ by the amount that 
	is not common to $B$. What is left is the uncertainty that is shared.

	Another way to write the mutual information is
	\be	\label{protoTypeTwo}
		I(A;B)	=  H(A)+H(B) - H(AB),
	\ee	
	which adds both entropies (double-counting the entropy that is common) and then subtracts the total 
	entropy. Equation \eqref{protoTypeTwo} is a measure of how different from independent the variables $A$ and 
	$B$ are. 
	Both of these interpretations of the mutual information can be generalized to the multiparty case.

	One way to define the multiparty mutual information for three variables $A,B$ and $C$ is by mimicking
	the form of equation \eqref{protoTypeOne} above and define
	\be
		I_\cap(A;B;C)	:=	I(A;B)	- I(A;B|C).
	\ee	
	The motivation behind this formula is to subtract from the mutual information $I(A;B)$ any terms that 
	are due to $AB$-only correlations and not true tripartite correlations.
	The expanded form of the restrictive mutual information is
	{\small
	\be	\label{typeOne}	\nonumber
		I_\cap(A;B;C) =	H(A)+H(B)+H(C) -H(AB)-H(AC)-H(BC) + H(ABC).
	\ee}
	This is the information shared by \emph{all} three parties and corresponds to the region labeled ``g'' 
	in  Figure~\ref{fig:TripleVenn}.
	This form of multiparty mutual information was defined in \cite{CondMutInfo} but has not yet proved
	useful in applications. 
	Also, $I_\cap(A;B;C)$ can take on negative values \cite{multisquash}, which are difficult to interpret.

    \begin{figure}[ht]  \begin{center} \input{figures/figureTripleVenn.pic} \end{center}
		\FigureCaptionOpt{Entropy diagram for the multiparty information.}{
						  Entropy diagram for three parties $A$, $B$ and $C$.}
		\label{fig:TripleVenn}
	\end{figure}

	Another approach is to define the multiparty mutual information in the spirit of \eqref{protoTypeTwo},
	as the measure of how different from independent the three variables are
	\be	\label{typeTwo}
		I_\cup(A;B;C) :=	H(A)+H(B)+H(C) - H(ABC).
	\ee
	In terms of the regions in Figure~\ref{fig:TripleVenn}, we have $I_\cup(A;B;C) = d+e+f+2g$.
	This form of the mutual information is naturally connected to the relative entropy\cite{NC04} and 
	also satisfies the chain-like property
	\be
		I_\cup(A;B;C)	=	I(A;B)	+ 	I(AB;C),
	\ee
	which indicates how the multiparty information is affected when we introduce a new system.

						\index{information!multiparty}
    In this chapter we will investigate some of the properties of the inclusive multiparty information
    $I_\cup(A;B;C)$, henceforth referred to simply as mutual information $I(A;B;C)$.
    Our work on the multiparty information will also allow us to naturally extend the notion of squashed 
    entanglement~\cite{CW04} to the multiparty scenario.
    The multiparty squashed entanglement, discussed in Section~\ref{sec:squashed-entanglement}, turns out to 
    be a measure of multipartite entanglement with excellent properties and clear and intuitive interpretation.
    It finds application in the proof of Theorem~\ref{thm:THM2}, the outer bound on the rate region for 
    distributed compression.

	\pagebreak


\bigskip
\bigskip
\section{Multiparty information} \label{sec:multiparty-information}

	We begin with the definition of the multiparty quantum information for $m$ parties.

	    	\index{information!multiparty}
    \begin{definition}[Multiparty information]  \label{Jformation}
        Given the state $\rho^{X_1X_2\ldots X_m}$ shared between $m$ systems, we define the multiparty information
        as: 
        \bea
            I(X_1;X_2; \cdots; X_m)_\rho    &:=& H(X_1) + H(X_2) + \cdots + H(X_m) - H(X_1X_2 \cdots X_m)
                                            \nonumber \\
                                            &=& \sum_{i=1}^m H(X_i)_\rho - H(X_1X_2 \cdots X_m)_\rho
        \eea
    \end{definition}

    The subadditivity inequality for quantum entropy ensures that the multiparty information
    is zero if and only if $\rho$ has the tensor product form
    $\rho^{X_1} \otimes \rho^{X_2} \otimes \cdots \otimes \rho^{X_m}$.

    The conditional version of the multiparty mutual information is obtained by replacing all the entropies by
    conditional entropies
    \bea
        I(X_1;X_2; \cdots; X_m|E)_\rho  &:=& \sum_{i=1}^m H(X_i|E) - H(X_1X_2 \cdots X_m|E) \nonumber \\
                                        &=&      \sum_{i=1}^m H(X_iE) - H(X_1X_2 \cdots X_mE) - (m-1)H(E) \nonumber \\
                                        &=&      I(X_1;X_2; \cdots; X_m;E) - \sum_{i=1}^m I(X_i;E).
    \eea
    This definition of multiparty information has appeared previously in \cite{Lindblad,RHoro,GPW05} and
    more recently in \cite{multisquash}, where many of its properties were investigated.


    Next we investigate some formal properties of the multiparty information which will be
    useful in our later analysis.

    \begin{lemma}[Merging of multiparty information terms] \label{mergingJ}
        Arguments of the multiparty information can be combined by subtracting their mutual information:
        \be
            I(A;B;X_1;X_2; \cdots;X_m) - I(A;B) = I(AB;X_1;X_2; \cdots ;X_m).
        \ee
    \end{lemma}
    \begin{proof} This identity is a simple calculation. 
    It is sufficient to expand the definitions and cancel terms.
	   \begin{align}
	        I(A;&B;X_1;X_2; \cdots;X_m) - I(A;B)    =  \nonumber \\
	            & = H(A)\!+\!H(B)\!+\!\!\sum\!H(X_i) -\!H(ABX_1X_2 \cdots X_m)\!-\!H(A)\!-\!H(B)\!+\!H(AB)   \nonumber \\
	            & = H(AB) + \sum H(X_i) - H(A,BX_1X_2 \cdots X_m)   \nonumber  \\
	            & = I(AB;X_1;X_2; \cdots ;X_m). \nonumber
	    \end{align} 
	\end{proof}


    Discarding a subsystem inside the conditional multiparty information cannot lead it to increase.
    This property, more than any other, justifies its use as a measure of correlation.
    \begin{lemma}[Monotonicity of conditional multiparty information] \label{cond-monotonicity}
        \be
            I(AB;X_1;\cdots X_m| E) \geq I(A;X_1;\cdots X_m| E)
        \ee
    \end{lemma}
    \begin{proof}
        This follows easily from strong subadditivity of quantum entropy (SSA).
        \begin{align*}
            I(AB&;X_1;X_2; \ldots;X_m|E) =  \\
            = &\        H(ABE) + \sum_i H(X_iE) - H(ABX_1X_2 \ldots X_mE) -mH(E)    \nonumber   \\
            = &\        H(ABE) + \sum_i H(X_iE) - H(ABX_1X_2 \ldots X_mE) -mH(E)    + \nonumber \\
              &\ \quad  \underbrace{H(AE) - H(AE)}_{=0} \quad
                        + \quad \underbrace{H(AX_1X_2 \ldots X_mE) - H(AX_1X_2 \ldots X_mE)}_{=0} \nonumber \\
            = &\        H(AE) + \sum_i H(X_iE) - H(AX_1X_2 \ldots X_m) -mH(E)  + \nonumber \\
              &\ \quad \underbrace{\left[
                                        H(ABE) + H(AX_1X_2 \ldots X_mE) - H(AE) - H(ABX_1X_2 \ldots X_mE)
                                    \right]}_{\geq 0 \; \mbox{\tiny by SSA}} \nonumber \\
        \end{align*}             
        \begin{align*}
            \geq &\        H(AE) + \sum_i H(X_iE) - H(AX_1X_2 \ldots X_mE) -mH(E)  \\
            = &\        I(A;X_1;X_2 \ldots X_m|E)
        \end{align*}
    \end{proof}

    We will now prove a multiparty information property that follows from a more general chain rule,
    but is all that we will need for applications.

    \begin{lemma}[Chain-type Rule] \label{useful-chain-rule}
    \be
        I(AA';\XX|E) \geq I(A;\XX|A'E)
    \ee
    \end{lemma}
    \begin{proof}
    \begin{align*}
           I(AA'&;\XX|E)        = \\
            =&\ \       H(AA'E)+ \sum_{i=1}^m H(X_iE) - H(AA'\XcX) - mH(E)  \\
            =&\ \       I(A;\XX|A'E) + \sum_{i=1}^m \left[ H(A'E) + H(X_iE)- H(E) - H(A'X_iE)  \right] \nonumber \\
         \geq&\ \       I(A;\XX|A'E).
    \end{align*}
    The inequality is true by strong subadditivity.         
    \end{proof}

    \myparagraph{Remark} It is interesting to note that we have two very similar reduction-of-systems formulas derived
    from different perspectives. From Lemma \ref{cond-monotonicity} (monotonicity of the multiparty information)
    we have that
        \be
            I(AB;\XX|E) \geq I(A;\XX|E),
        \ee
    but we also know from Lemma \ref{useful-chain-rule} (chain-type rule) that
        \be
            I(AB;\XX|E) \geq I(A;\XX|BE).
        \ee
    The two expressions are inequivalent; one is not strictly stronger than the other.
    We use both of them depending on whether we want to keep the deleted system around for conditioning.


	\bigskip
	\bigskip	
	\section{Multiparty squashed entanglement} 	\label{sec:squashed-entanglement}

    Using the definition of the conditional multiparty information from the previous
    section, we can define a multiparty squashed entanglement analogous to the bipartite version
    \cite{tucci-1999,tucci-2002,CW04}.
    The multiparty squashed entanglement has recently been investigated independently by Yang et al. \cite{multisquash}.

	\index{entanglement!squashed@$\Esq$: squashed entanglement}
    \bigskip
    \begin{definition}[Multiparty squashed entanglement]  \label{squashedEntanglement}
        Consider the density matrix $\rho^{X_1X_2\ldots X_m}$ shared between $m$ parties. We define the multiparty 
        squashed entanglement in the following manner
        \bea
            \Esq(X_1;X_2; \ldots; X_m)_\rho
                        &:=&        \frac{1}{2}\inf_E
                                        \left[
                                            \sum_{i=1}^m H(X_i|E)_\rhot - H(X_1X_2 \cdots X_m|E)_\rhot
                                        \right] \nonumber \\
                        &=&             \frac{1}{2}\inf_E  I(X_1;X_2;\cdots;X_m|E)_\rhot
        \eea
        where the minimization happens over all states of the form $\rhot^{X_1X_2\ldots X_mE}$ such that
        $\Tr_E\!\left(\rhot^{X_1X_2\ldots X_mE}\right) = \rho^{X_1X_2\ldots X_m}$. (We say $\rhot$
        is an \emph{extension} of $\rho$.)
    \end{definition}
    The dimension of the extension system $E$ can be arbitrarily large, which is in part what makes calculations
    of the squashed entanglement very difficult except for simple systems.
    The motivation behind this definition is that we can include a copy of all classical correlations inside
    the extension $E$ and thereby eliminate them from the multiparty information by conditioning.
	Since it is impossible to copy quantum information, 
	we know that taking the infimum over all possible
	extensions $E$ we will be left with a measure of the purely quantum correlations.
	%
	%
    The definition of \Esq as a minimization over a conditional mutual information is motivated by the classical 
    cryptography notion of \emph{intrinsic information} which provides a bound on the secret-key
    rate~\cite{MW99,C05,CW04}.



	\bigskip
    \myparagraph{Example:} It is illustrative to calculate the squashed entanglement for separable states, which
    are probabilistic mixtures of tensor products of local pure states.
    Consider the state \vspace{-0.15cm}
    \be
        \rho^{X_1X_2\ldots X_m} = \sum_j p_j        \proj{\alpha_j}^{X_1} \otimes
                                                    \proj{\beta_j}^{X_2} \otimes \cdots
                                                    \proj{\zeta_j}^{X_m},                     \nonumber
                                                    \vspace{-0.3cm}
    \ee
    which we choose to extend by adding a system $E$ containing a record of the index $j$ as follows
    \be
        \rhot^{X_1X_2\ldots X_mE} = \sum_j p_j      \proj{\alpha_j}^{X_1} \otimes
                                                    \proj{\beta_j}^{X_2} \otimes \cdots
                                                    \proj{\zeta_j}^{X_m}  \otimes
                                                    \proj{j}^E.                   \nonumber
                                                    \vspace{-0.3cm}
    \ee
    When we calculate conditional entropies we notice that for any subset $\K \subseteq \{1,2,\ldots m\}$,
    \be
        H(X_\K | E)_\rhot = 0.
    \ee
    Knowledge of the classical index leaves us with a pure product state for which all the relevant entropies are zero.
    Therefore, separable states have zero squashed entanglement:
    \be	\nonumber
        \Esq(X_1;X_2; \ldots ;X_m)_\rho
                    = \frac{1}{2}\left[ \sum_i^m H(X_i|E)_\rhot -H(X_1X_2 \ldots X_m|E)_\rhot \right]
                    = 0.
    \ee

\bigskip

    We now turn our attention to the properties of $\Esq$.
    Earlier we argued that the squashed entanglement measures purely quantum contributions
    to the mutual information between systems, in the sense that it is zero for all separable states.
    In this section we will show that the multiparty squashed entanglement cannot increase under the action of
    local operations and classical communication, that is, that $\Esq$ is an LOCC-monotone.
    We will also show that \Esq has other desirable properties; it is convex, subadditive and continuous.


    \begin{proposition}
        The quantity $\Esq$ is an entanglement monotone, i.e. it does not increase 
        on average under local quantum operations and classical communication (LOCC). 
    \end{proposition}

    \begin{proof}
    In order to show this we will follow the argument of \cite{CW04}, which in turn follows the approach
    described in \cite{Vid00}.  We will show that \Esq has the following two properties:
    \begin{enumerate}	\label{requirementsForLOCC}
        \item   Given any unilocal quantum instrument $\E_k$
                (a collection of completely positive maps such that $\sum_k\!\E_k$ is trace preserving \cite{DL70})
                and any quantum state $\rho^{X_1\ldots X_m}$, then
                \be
                    \Esq(X_1;X_2;\ldots X_m)_\rho \geq \sum_k p_k \Esq(X_1;X_2;\ldots X_m)_{\rhot_k}
                \ee
                where
                \be
                    p_k = \Tr\ \E_k(\rho^{X_1\ldots X_m})
                    \quad \textrm{and} \quad
                    \rhot_k^{X_1\ldots X_m}=\frac{1}{p_k}\E_k(\rho^{X_1\ldots X_m}).
                \ee
        \item   $\Esq$ is convex.
    \end{enumerate}


    Without loss of generality, we assume that $\E_k$ acts on the first system.
    We will implement the quantum instrument by appending to $X_1$
    environment systems $X_1'$ and $X_1''$ prepared in standard pure states,
    applying a unitary $U$ on $X_1X_1'X_1''$, and then tracing out over $X_1''$.
    We store $k$, the classical record of which $\E_k$ occurred, in the $X_1'$ system.
    More precisely, for any extension of $\rho^{X_1X_2\cdots X_m}$ to $X_1X_2\cdots X_mE$,
        \be
            \rho^{X_1X_2\ldots X_mE}    \mapsto
            \rhot^{X_1X_1'X_2\ldots X_mE} := \sum_k \ \E_k\!\! \otimes\!\!
                    I_E \!\left( \rho^{X_1X_2\ldots X_mE} \right)\otimes\ket{k}\bra{k}^{X_1'}.
        \ee
    The argument is then as follows:
    \bea
      \frac{1}{2} I(X_1;X_2;\ldots X_m|E)_\rho  &=&
                \frac{1}{2}I(X_1X_1'X_1'';X_2;\ldots ;X_m|E)_\rho \label{refone}    \\
        &=&     \frac{1}{2}I(X_1X_1'X_1'';X_2;\ldots ;X_m|E)_{\rhot}    \label{reftwo} \\
        &\geq&  \frac{1}{2}I(X_1X_1';X_2;\ldots ;X_m|E)_{\rhot} \label{refthree} \\
        &\geq&  \frac{1}{2}I(X_1;X_2;\ldots; X_m|EX_1')_{\rhot} \label{reffour} \\
        &=&     \frac{1}{2}\sum_k p_k I(X_1;X_2;\ldots ;X_m|E)_{\rhot_k}    \label{reffive} \\
        &\geq&  \sum_k p_k \Esq\left(X_1;X_2;\ldots ;X_m \right)_{\rhot_k} \label{refsix}
    \eea

    The equality \eqref{refone} is true because adding an uncorrelated ancilla does not change the entropy of the system.
    The transition $\rho \rightarrow \rhot$ is unitary and doesn't change entropic quantities so
    \eqref{reftwo} is true.
    For \eqref{refthree} we use the monotonicity of conditional multiparty information, Lemma \ref{cond-monotonicity}.
    In \eqref{reffour} we use the chain-type rule from Lemma \ref{useful-chain-rule}.
    In \eqref{reffive} we use the index information $k$ contained in $X_1'$.
    Finally, since $\Esq$ is the infimum over all extensions, it must be no more than the particular extension $E$,
    so \eqref{refsix} must be true.
    Now since the extension $E$ in \eqref{refone} was arbitrary, it follows that
    $\Esq({X_1;X_2;\ldots ;X_m})_\rho \geq \sum_k p_k \Esq\left(X_1;X_2;\ldots; X_m \right)_{\rhot_k}$ which completes
    the proof of Property 1.

    To show the convexity of $\Esq$, we again follow the same route as in \cite{CW04}.
    Consider the states $\rho^{X_1X_2\ldots X_m}$ and $\sigma^{X_1X_2\ldots X_m}$ and their
    extensions $\rhot^{X_1X_2\ldots X_mE}$ and $\sigmat^{X_1X_2\ldots X_mE}$ defined over the same system $E$.
    We can also define the weighted sum of the two states
    $\tau^{X_1X_2\ldots X_m} = \lambda\rho^{X_1X_2\ldots X_m} + (1-\lambda)\sigma^{X_1X_2\ldots X_m}$
    and the following valid extension:
    \be
        \tilde{\tau}^{X_1X_2\ldots X_mEE'} = \lambda\rho^{X_1X_2\ldots X_mE}\otimes\proj{0}^{E'}
                                            + (1-\lambda)\sigma^{X_1X_2\ldots X_mE}\otimes\proj{1}^{E'}.
    \ee
    Using the definition of squashed entanglement we know that
    \begin{align}
        \Esq(X_1;X_2;\ldots &; X_m)_\tau \nonumber \\
                            &\leq\      \frac{1}{2} I(X_1;X_2;\ldots; X_m|EE')_{\tilde{\tau}}   \nonumber \\
                            &=\         \frac{1}{2}\left[ \lambda I(X_1;X_2;\ldots ;X_m|E)_\rhot
                                            + (1-\lambda)I(X_1;X_2;\ldots ;X_m|E)_\sigmat \right].  \nonumber
    \end{align}
    Since the extension system $E$ is completely arbitrary we have
    \be
        \Esq(X_1;\ldots; X_m)_\tau
        \leq
        \lambda\Esq(X_1;\ldots; X_m)_\rho + (1-\lambda)\Esq(X_1;\ldots; X_m)_\sigma, \nonumber
    \ee
    so \Esq is convex.

    We have shown that \Esq satisfies both Properties 1 and 2 from page~\pageref{requirementsForLOCC}. 
    Therefore, it must be an entanglement monotone.
    \end{proof}

\bigskip
\myparagraph{Subadditivity on Product States}
    Another desirable property for measures of entanglement is that they should be additive or at least subadditive on
    tensor products of the same state. Subadditivity of \Esq is easily shown from the properties of
    multiparty information.

    \begin{proposition}     \label{subadditiveTensorProducts}
    $\Esq$ is subadditive on tensor product states, i.e.
        \be \label{subadditivity}
            \Esq\!\left( {X_1Y_1;\ldots;X_mY_m} \right)_\rho \leq
                    \Esq\!\left( {X_1;\ldots;X_m} \right)_\rho
                    +  \Esq\!\left( {Y_1;\ldots;Y_m} \right)_\rho
        \ee
        where $\rho^{X_1Y_1X_2Y_2\ldots X_mY_m}=\rho^{X_1X_2\ldots X_m}\otimes\rho^{Y_1Y_2\ldots Y_m}$.
    \end{proposition}
    \begin{proof}
    Assume that $\rho^{X_1X_2\ldots X_mE}$ and $\rho^{Y_1Y_2\ldots Y_mE'}$ are extensions. Together they form an
    extension $\rho^{X_1Y_1X_2Y_2\ldots X_mY_mEE'}$ for the product state.
    \begin{align}
        2\Esq\big(X_1&Y_1;X_2Y_2;\ldots;X_mY_m\big)_\rho \nonumber \\
            &\leq\      I(X_1Y_1;X_2Y_2;\ldots;X_mY_m|EE') \nonumber \\
            &=\         \sum_i H(X_iY_iEE') - H(X_1Y_1X_2Y_2\ldots X_mY_mEE')
                                    - (m-1)H(EE') \nonumber
    \end{align}
    \begin{align}
            &=\         I(X_1;X_2;\ldots;X_m|E) + I(Y_1;Y_2;\ldots;Y_m|E'). \qquad \ \ \ 
    \end{align}
    The first line holds because the extension for the $XY$ system that can be built by combining
    the $X$ and $Y$ extensions is not the most general extension.
    The proposition then follows because the inequality holds for all extensions of $\rho$ and $\sigma$.
    \end{proof}

    The question of whether \Esq is additive, meaning superadditive in addition to subadditive, remains an open problem.
    Indeed, if it were possible to show that correlation between the $X$ and $Y$ extensions is unnecessary
    in the evaluation of the squashed entanglement of $\rho \otimes \sigma$, then \Esq would be additive.
    This is provably true in the bipartite case~\cite{CW04} but the same method does not seem to work
    with three or more parties.


\bigskip
\myparagraph{Continuity}
    The continuity of bipartite \Esq was conjectured in \cite{CW04} and proved by Alicki and Fannes in \cite{AF04}.
    We will follow the same argument here to prove the continuity of the multiparty squashed entanglement.
    The key to the continuity proof is the following lemma which makes use of an ingenious geometric construction.

    \begin{lemma}[Continuity of conditional entropy \cite{AF04}] \label{continuityOfConditional}
    Given density matrices $\rho^{AB}$ and $\sigma^{AB}$ on the space $\calH^A \otimes \calH^B$ such that
        \be
            \| \rho - \sigma \|_1 = \frac{1}{2}\Tr | \rho - \sigma | \leq \epsilon,
        \ee
    it is true that
        \be
            \left| H(A|B)_\rho - H(A|B)_\sigma \right| \leq 4\epsilon \log d_A + 2h(\epsilon)
        \ee
    where $d_A=\dim \calH^A$ and $h(\epsilon)=-\epsilon\log\epsilon - (1-\epsilon)\log(1-\epsilon)$ is the
    binary entropy.
    \end{lemma}

    This seemingly innocuous technical lemma makes it possible to prove the continuity of \Esq in spite
    of the unbounded dimension of the extension system.
    \begin{proposition}[\Esq is continuous]
        For all states $\rho^{X_1X_2\ldots X_m}$, $\sigma^{X_1X_2\ldots X_m}$  with 
        \mbox{$\| \rho - \sigma \|_1 \leq \epsilon$},
        \mbox{$\| \Esq(\rho)-\Esq(\sigma) \| \leq \epsilon'$} where $\epsilon'$ depends on $\epsilon$ and vanishes as
        $\epsilon \rightarrow 0$.
    \end{proposition}
    The precise form of $\epsilon'$ can be found in equation \eqref{epsilonprime}.

    \begin{proof}
        Proximity in trace distance implies proximity in fidelity distance \cite{Fuchs}, in the sense that
        \be
            F(\rho^{X_1X_2\ldots X_m}, \sigma^{X_1X_2\ldots X_m})   \geq        1 -\epsilon,
        \ee
        but by Uhlmann's theorem \cite{U76} this means that we can find purifications $\ket{\rho}^{X_1X_2\ldots X_mR}$
        and  $\ket{\sigma}^{X_1X_2\ldots X_mR}$ such that
        \be
            F(\ket{\rho}^{X_1X_2\ldots X_mR}, \ket{\sigma}^{X_1X_2\ldots X_mR})     \geq        1 -\epsilon.
        \ee
        Now if we imagine some general operation $\Lambda$ that acts only on the purifying system $R$
        \bea
            \rho^{X_1X_2\ldots X_mE}    &=&     (I^{X_1X_2\ldots X_m}\otimes\Lambda^{R \rightarrow E})
                                                \proj{\rho}^{X_1X_2\ldots X_mR}   \\
            \sigma^{X_1X_2\ldots X_mE}  &=&     (I^{X_1X_2\ldots X_m}\otimes\Lambda^{R \rightarrow E})
                                                \proj{\sigma}^{X_1X_2\ldots X_mR}
        \eea
        we have from the monotonicity of fidelity for quantum channels that
        \be
            F({\rho}^{X_1X_2\ldots X_mE}, {\sigma}^{X_1X_2\ldots X_mE})
                \geq F(\ket{\rho}^{X_1X_2\ldots X_mR}, \ket{\sigma}^{X_1X_2\ldots X_mR})
                    \geq        1 -\epsilon,
        \ee
        which in turn implies \cite{Fuchs} that
        \be
            \| \rho^{X_1X_2\ldots X_mE} - \sigma^{X_1X_2\ldots X_mE} \|_1   \leq    2\sqrt{\epsilon}.
        \ee
        Now we can apply Lemma \ref{continuityOfConditional} to each term in the multiparty information to obtain
        \begin{align}
        \Big| I(X_1&;X_2;\ldots X_m|E)_\rho - I(X_1;X_2;\ldots X_m|E)_\sigma \Big| \nonumber \\
            &\leq\  \sum_{i=1}^m    \Big| H(X_i|E)_\rho - H(X_i|E)_\sigma \Big| \nonumber \\
            & \qquad\qquad\quad    +\ \Big| H(X_1X_2\ldots X_m|E)_\rho  - H(X_1X_2\ldots X_m|E)_\sigma \Big| \nonumber \\
            &\leq\  \sum_{i=1}^m \left[ 8\sqrt{\epsilon} \log d_i + 2h( 2\sqrt{\epsilon}) \right]
                    + 8\sqrt{\epsilon} \log\left( \prod_{i=1}^m d_i \right)
                    + 2h( 2\sqrt{\epsilon} ) \nonumber \\
            &=\     16\sqrt{\epsilon}\log\left( \prod_{i=1}^m d_i \right)
                    + (m+1)2h( 2\sqrt{\epsilon} )
                    =: \epsilon' \label{epsilonprime}
        \end{align}
        where $d_i=\dim {\calH}^{X_i}$ and $h(.)$ is as defined in Lemma \ref{continuityOfConditional}.
        Since we have shown the above inequalities for \emph{any} extension $E$ and the quantity $\epsilon'$ vanishes
        as $\epsilon \rightarrow 0$, we have proved that \Esq is continuous.
    \end{proof}

	\bigskip
	\bigskip
	\section{Example calculations of \Esq}   \label{subsec:examples}

		Below we give several examples of simple systems where \Esq is calculated to gain intuition about how it behaves.
		As a first step we verify that \Esq is zero for states that are manifestly not entangled.
	
		\myparagraph{Example 1: Fully decoupled state}
		Given the state $\rho_1^{X_1X_2\ldots X_m} = \rho^{X_1} \otimes \rho^{X_2} \otimes\cdots\otimes\rho^{X_m} =
		\bigotimes_1^m \rho^{X_i}$ the mutual information for this state is:
		   \be
		       I(X_1;X_2; \ldots ;X_m) = \sum_i^m H(X_i) - \underbrace{H(X_1,X_2, \ldots, X_m)}_{=\sum_i^m H(X_i)} = 0
		   \ee
		which is to be expected since the state is a tensor product and cannot contain entanglement.

		In the next example we look at more complicated states where \Esq is non-zero but simple to calculate.
		
		\myparagraph{Example 2: Partially separable state}
		Now consider a state which is separable on all systems except for two. We write
		\be     \rho_2^{ABX_1X_2\ldots X_m} = \sum_j p_j
		                                               \proj{\alpha_j}^{AB} \otimes
		                                               \proj{\beta_j}^{X_1} \otimes \cdots
		                                               \proj{\zeta_j}^{X_m}
		                                               \nonumber
		\ee
		\noindent and an extension $E$ that records the index $j$. For this extension we will have:
		{ \small
		\begin{align*}
		   I(A;B;X_1;&\ldots ;X_m|E)       = \\
		   &=\ \ \              H(A|E) + H(B|E) + \sum_i^m H(X_i|E) - H(ABX_1\cdots X_m|E)\\
		   &=\ \ \              H(AE)\!+\!H(BE)\!+\!\!\sum_i^m H(X_iE)\!-\!H(ABX_1\cdots X_mE)\!-\!(m\!+\!1)H(E) \\
		   &\geq^{(1)}     		H(AE) + H(BE) - H(ABE) - H(E) \\
		   &=\ \ \              I(A;B|E) \\
		   &\geq\ \ \          	2 \Esq(A;B)_\rho
		\end{align*} }
		To show (1) we repeatedly used the strong subadditivity property of von~Neumann entropy
		\be
		   -H(E) - H(EY_1 \cdots Y_m) \geq - H(EY_m) - H(EY_1 \cdots Y_{m-1}).
		\ee
		Thus we have shown that for partially separable states, \Esq of the whole is at least as much as
		its non-separable part. 

		\myparagraph{Example 3: \Esq for the GHZ and W states}
		Consider the $m$-party GHZ state 
		$\ket{GHZ}^{X_1X_2\cdots X_m} = \frac{\ket{0}^{\otimes m} + \ket{1}^{\otimes m}}{\sqrt{2}}$
		and the $m$-party W state 
		$\ket{W}^{X_1X_2\cdots X_m}=\frac{1}{\sqrt{m}}\sum_{i=0}^{m-1} \vert\hat{i}\rangle$,
		where $\vert\hat{i}\rangle = \vert 0 \cdots 0 1 \underbrace{0 \cdots 0}_{i}\rangle$.
		In particular, the three-party GHZ and W states correspond to
		\be	\nonumber
		   \ket{GHZ} = \frac{\ket{000} + \ket{111}}{\sqrt{2}}  \qquad  \rm{and}
		                                                       \qquad  \ket{W}=\frac{1}{\sqrt{3}}
		                                                               \left( \ket{001}+\ket{010}+\ket{100} \right).
		\ee
		
		\noindent The squashed entanglement of the the general GHZ state is
		\bea
	      \Esq(X_1;X_2;\cdots;X_m)_{GHZ}	&=&		\frac{1}{2}\inf_E I(X_1;X_2;\cdots;X_m|E)	\nonumber\\ \nonumber
	   								&=&		\frac{1}{2}I(X_1;X_2;\cdots;X_m) \qquad\qquad\textrm{(pure state)}\\
	   								&=&     \frac{1}{2}\bigg[
	   												\sum_i^m \!\!\!\underbrace{H(X_i)}_{ \text{max. mixed} }
	                                                - \ \ \underbrace{H(X_1\ldots X_m)}_{ \text{pure}}
											\bigg]  \nonumber\\
									&=&     \frac{m}{2} \nonumber
		\eea
		
		\noindent For the W state, the 1-qubit reduced systems are of the form
		\be
		   \Tr_{X_2\ldots X_m}\left(\proj{W}\right) =
		                                                   \left(\begin{array}{cc}
		                                                       \frac{m-1}{m}       &   0               \\
		                                                       0                   &   \frac{1}{m}
		                                                   \end{array}\right)
		\ee
		and so the squashed entanglement for the W state is given by the formula
		\begin{eqnarray*}
		   \Esq(X_1;X_2;\cdots;X_m)_W 	&=&	\frac{1}{2}I(X_1;X_2;\cdots;X_m)   \\
		   								&=&	\frac{1}{2}\bigg[ 
	   												\sum_i^m H(X_i)
	                                                - \underbrace{H(X_1\ldots X_m)}_{ =0 }
											\bigg]  \\
	   							   	&=&	 \frac{m}{2}\log_2\left( \frac{m}{(m-1)^{\frac{(m-1)}{m}}}\right) \\
	   								&=&	  \frac{1}{2}\log_2\left( \frac{m^m}{(m-1)^{(m-1)}}\right) \\
	   								&=&	  \frac{1}{2} \log_2 m + O(1) \quad  <\!< \frac{m}{2}.
		\end{eqnarray*}
		We can see that the GHZ state is maximally multiparty entangled whereas the W state contains 
		very little multiparty entanglement.



\chapter{Multiparty distributed compression}	\label{chapter:multiparty-distributed-compression}

	
	    Distributed compression of classical information, as discussed in
	    Section~\ref{subsection:classical-slepian-wolf}, involves many parties collaboratively encoding their 
	    classical sources $X_1,X_2\cdots X_m$ and sending the information to a common receiver \cite{C75}.
	    In the quantum setting, the parties are given a quantum state $\ph^{A_1A_2\cdots A_m} \in \cH^{A_1A_2\cdots A_m}$
	    and are asked to individually compress their shares of the state and transfer them
	    to the receiver while sending as few qubits as possible \cite{ADHW04}. 
	    %
		%
	    We have already discussed a version of quantum distributed compression in 
	    Section~\ref{section:state-merging} where we used shared entanglement and classical communication 
	    to accomplish the task \cite{HOW05b}. 
	    In this chapter, we consider the fully quantum scenario where only quantum communication is used and 
	    classical communication is forbidden.

	    In our analysis, we work in the case where we have many copies of the input state, so that
	    the goal is to send shares of
	    the purification $\ket{\psi}^{A_1A_2\cdots\!A_mR}=(\ket{\ph}^{A_1A_2\cdots\!A_mR})^{\otimes n}$,
	    where the $A_i$'s denote the $m$ different systems and $R$ denotes the reference system, which does not
	    participate in the protocol.
	    A word on notation is in order. 
	    We use $A_i$ to denote both the individual system associated with state $\ph$ as well the $n$-copy
	    version $A_i^{\otimes n}$ associated with $\psi$; the intended meaning should be clear from the context.
	    We also use the shorthand notation $A=A_1A_2\cdots\!A_m$ to denote all the senders.
	    
	    The objective of distributed compression is for the participants to transfer their $R$-entanglement 
	    to a third party Charlie as illustrated in Figure \ref{fig:mpFQSWdiagram}.
	    As discussed in Section~\ref{subsection:ent-fid}, preserving the $R$-entanglement means our protocol
	    has high \emph{entanglement fidelity} \cite{EntFid} which guarantees that we can transfer
	    the state $\ph^{A_1A_2\cdots\!A_m}$, but also preserve all the correlations this state has with the 
	    rest of the world.
	    
	    \bigskip
	    
	    \begin{figure}[ht] \begin{center}
	        {\footnotesize
	        \input{figures/figurempFQSWdiagram.pic} }

	        \FigureCaptionOpt{	Representation of the quantum correlations in the multiparty distributed 
	        					compression protocol.}{
	            Pictorial representation of the quantum correlations between the systems at three stages of the protocol.
	            Originally the state $\ket{\psi}$ is shared between $A_1A_2\cdots A_m$ and $R$.
	            The middle picture shows the protocol in progress.
	            Finally, all systems are received by Charlie and $\ket{\psi}$ is now shared between
	            Charlie's systems $\widehat{A}_1\widehat{A}_2 \cdots \widehat{A}_m$ and $R$.            }
	        \label{fig:mpFQSWdiagram}
	    \end{center}
	    \end{figure}

		An equivalent way of thinking about quantum distributed compression is to say that the 
		participants are attempting to decouple their systems from the reference $R$ solely by 
		sending quantum information to Charlie.
		Indeed, if we assume that originally $R$ is the purification of $A_1A_2\cdots A_m$, 
		and at the end of the protocol there are no correlations between 
		the remnant $W$ systems (see Figure \ref{fig:mpFQSWdiagram}) and $R$, 
		then the purification of $R$ must have been transferred to 
		Charlie's laboratory since none of the original information was discarded.

	    To perform the distributed compression task, each of the senders independently encodes her share
	    before sending part of it to Charlie.
	    The encoding operations are modeled by quantum operations (CPTP maps) $\E_i$ with outputs $C_i$ of dimension
	    $2^{nQ_i}$. \index{encoding operation@$\E_i$: encoding operation}
	    Once Charlie receives the systems that were sent to him, he will apply a decoding operation $\cD$, with 
	    output system 
	    $\widehat{A}=\widehat{A}_1\widehat{A}_2 \ldots \widehat{A}_m$ isomorphic to the original $A=A_1A_2\ldots A_m$.
	 
	 	\bigskip				\index{rate!region}
	    \begin{definition}[The rate region] \label{achievable}
	    We say that a rate tuple $\vec{Q} = (Q_1,Q_2,\ldots,Q_m)$ is achievable if for all $\epsilon > 0$ there exists
	    $N(\epsilon)$ such that for all $n\geq N(\epsilon)$ there exist $n$-dependent maps $(\E_1,\E_2,\ldots,\E_m,\cD)$
	    with domains and ranges as in the previous paragraph for which the fidelity
	    between the original state, $\ket{\psi}^{A^nR^n}=\left(\ket{\ph}^{A_1A_2\cdots\!A_mR}\right)^{\otimes n}$,
	    and the final state,   ${\sigma}^{\widehat{A}_1\widehat{A}_2 \ldots \widehat{A}_mR}={\sigma}^{\widehat{A}^nR^n}$,
	    satisfies
	    { \small
	    \begin{equation} \nonumber 		
	        F\!\left( \ket{\psi}^{A^nR^n}\!\!\!, {\sigma}^{\widehat{A}^nR^n} \right) \!=\!
	        \phantom{.}^{\widehat{A}^nR^n}\!\!\bra{\psi}
	            (\cD \circ (\E_1 \ox \cdots \ox \E_{m}))(\psi^{A^nR^n})
	            \ket{\psi}^{\widehat{A}^nR^n} \geq 1 - \epsilon.
	    \end{equation} }
	    We call the closure of the set of achievable rate tuples the rate region.
	    \end{definition}



	\bigskip
	\bigskip
	\section{The multiparty FQSW protocol} \label{subsec:protocol}
	    Like the original FQSW protocol, the multiparty version relies on Schumacher compression and the mixing
	    effect of random unitary operations for the encoding.
	    The only additional ingredient is an agreed upon permutation of the participants.
	    The temporal order in which the participants will perform their encoding is of no importance.
	    However, the permutation determines how much information each participant is to send to Charlie. 
	
		\bigskip		\index{protocol!multiparty FQSW}
	    For each permutation $\pi$ of the participants, the protocol consists of the following steps:
	    \begin{enumerate}
	        \item   Each Alice-$i$ performs Schumacher compression on her system $A_i$ reducing its
	                effective size to the entropy bound of roughly $H(A_i)$ qubits per copy of the state.
	        \item   Each participant applies a known, pre-selected random unitary to the compressed system.
	        \item   Participant $i$ sends to Charlie a system $C_i$ of dimension $2^{nQ_i}$ where
	                \be
	                    Q_i > \frac{1}{2}I(A_i;A_{\K_i}R)_\ph
	                \ee
					where $\K_i = \{ \pi\!(j) : j>\pi^{\neg 1}(i) \}$ is the set of participants who come
					after participant $i$ according to the permutation.
	        \item   Charlie applies a decoding operation $D$ consisting of the composition of the decoding maps
	                $\cD_{\pi\!(m)} \circ \cdots \circ \cD_{\pi\!(2)} \circ \cD_{\pi\!(1)}$ defined by the individual 
	                FQSW steps in order to recover $\sigma^{\widehat{A}_1\widehat{A}_2 \ldots \widehat{A}_m}$
	                nearly identical to the original $\psi^{A_1A_2\cdots\!A_m}$ and purifying $R$.
	    \end{enumerate}

		Note that, in order to perform the decoding operation $\cD$, Charlie needs to know which random unitaries
		which were used in the individual encoding operations $\E_i$. We assume this information is shared before 
		the beginning  of the protocol in addition to the permutation $\pi$.
		


	\subsection{Statement of results} \label{subsec:statemenet-of-results}
	
	    This section contains our two main theorems about multiparty distributed compression.
	    In Theorem \ref{thm:THM1} we give the formula for the set of achievable rates using the multiparty FQSW protocol
	    (sufficient conditions).
	    Then, in Theorem \ref{thm:THM2} we specify another set of inequalities for the rates $Q_i$ which must be true
	    for any distributed compression protocol (necessary conditions).
	    In what follows, we consistently use $\K \subseteq \{ 1,2,\ldots m \}$ to denote any subset of the senders in
	    the protocol.
	
	
	    \begin{theorem}     \label{thm:THM1}
	    Let $\ket{\ph}^{A_1A_2\cdots A_mR}$ be a pure state. If the inequality
	        \be \label{inner-bound}
	            \sum_{k\in \K} Q_k \geq  
	            	\frac{1}{2} \left[ \sum_{k\in \K}\!\left[H(A_k)_\ph\right]  +  H(R)_\ph -  H(RA_{\K})_\ph \right]
	        \ee
	    holds for all $\K \subseteq \{1,2,\ldots,m\}$, then the rate tuple $(Q_1,Q_2,\cdots,Q_m)$ is achievable
	    for distributed compression of the $A_i$ systems.
	    \end{theorem}
	
	    Because Theorem \ref{thm:THM1} expresses a set of sufficient conditions for the protocol to succeed, we say
	    that these rates are contained in the rate region. The proof is given in the next section.

	    In the $m$-dimensional space of rate tuples $(Q_1,Q_2,\cdots,Q_m) \in \RR^m$, the inequalities
	    \eqref{inner-bound} define a convex polyhedron \cite{poly} whose facets are given by the corresponding
	    hyperplanes, as illustrated in Figure~\ref{fig:graph3d}.
	    More specifically, the rate region is a supermodular polyhedron \cite{E69}, which means that 
	    it has some special properties that will help us in the proof of Theorem~\ref{thm:THM1}.
	
	    \begin{figure}[ht]
	        \begin{center}      \includegraphics[width=4in]{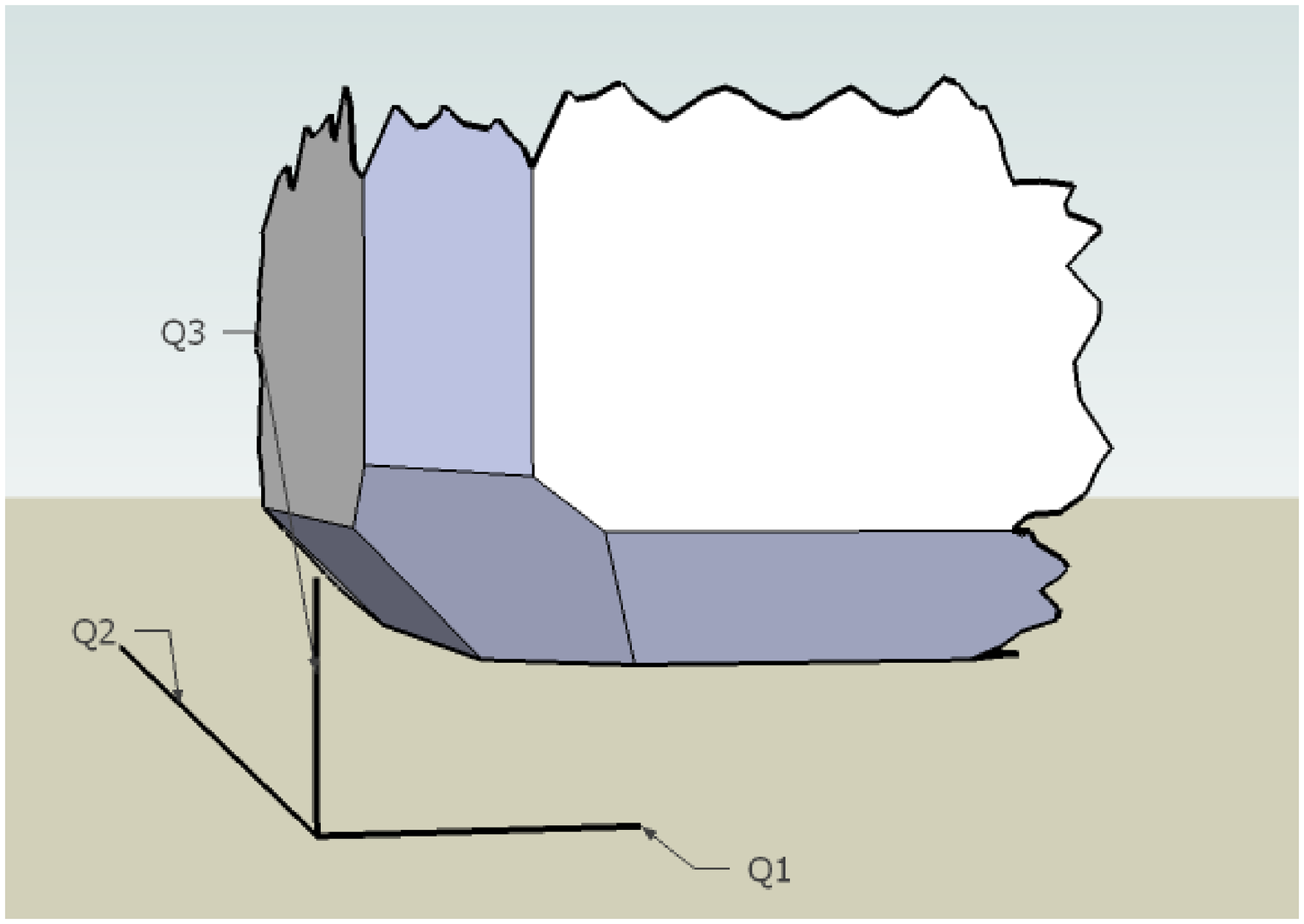}     \end{center}
	        \FigureCaptionOpt{The rate region for the multiparty FQSW protocol with three senders.}{
							  The rate region for the multiparty FQSW protocol with three senders.}
	        \label{fig:graph3d}
	    \end{figure}

	    In order to characterize the rate region further we formulate Theorem~\ref{thm:THM2},
	    an outer bound on the rates that must be satisfied for \emph{all} distributed compression protocols. 
	    \begin{theorem}     \label{thm:THM2}
	    Let $\ket{\ph}^{A_1A_2\cdots A_mR}$ be a pure state input to a distributed compression protocol which
	    achieves the rate tuple $(Q_1,Q_2,\ldots,Q_m)$, then it must be true that
	        \be \label{outer-bound}
	            \sum_{k\in \K} Q_k
	            \geq
	            \frac{1}{2} \left[ \sum_{k\in \K}\!\left[H(A_k)_\ph\right]  +  H(R)_\ph -  H(RA_{\K})_\ph \right]
	                        - \Esq( A_{k_1};A_{k_2};\ldots ;A_{k_{|\K|}} )_\ph,
	        \ee
	    for all $\K \subseteq \{1,2,\ldots,m\}$, where \Esq is the multiparty squashed entanglement.
	    \end{theorem}
		
		\bigskip	
	    The multiparty squashed entanglement 
	    was defined in Section~\ref{sec:squashed-entanglement} above. 
	    \index{entanglement!squashed@$\Esq$: squashed entanglement}
	    

	    Notice that Theorems \ref{thm:THM1} and \ref{thm:THM2} both provide bounds of the same form and
	    only differ by the presence of the \Esq term.
	    The rate region is squeezed somewhere between these two bounds as illustrated in Figure \ref{fig:2Dregion}.
	    
	    \begin{figure}[ht]  \begin{center}
	        \def\JPicScale{0.8}
	        {\footnotesize
	        \input{figures/figure2Dregion.pic} \def\JPicScale{1} 
	        }
	        \end{center}
			\FigureCaptionOpt{  Two dimensional diagram showing the inner and outer bound on the rate region.}{
								A two dimensional diagram showing the inner bound from Theorem \ref{thm:THM1} 
								and the outer bound from Theorem \ref{thm:THM2}. The boundary of the real rate 
								region must lie somewhere in between. }	             
	        \label{fig:2Dregion}
	    \end{figure}

	    For states which have zero squashed entanglement, the inner and outer bounds on the region coincide so that
	    in those cases our protocol is an optimal solution to the multiparty distributed compression problem.
	
%
%

	\bigskip
	\bigskip
	\section{Proof of inner bound}	\label{sec:THMIproof}
	
	    The multiparty fully quantum Slepian-Wolf protocol can be constructed directly \cite{HW06} or
	    through the repeated application of the two-party FQSW protocol \cite{FQSW}. 
		We choose the latter approach here in order to illustrate the power of the FQSW protocol as 
		a building block for more complex protocols.
		To complete the proof we will have to ``stitch together'' different achievable points 
		using some concepts from the theory of polyhedra~\cite{poly}.
		The multiparty rate region has a complex but regular geometry so it is important that we 
		use the right language to describe it.
		The geometry of multiparty rate regions has previously been discussed in \cite{E69,TH98}.
							\index{rate!region}
	
	    For every permutation $\pi \in S_m$ of the $m$ senders, there is a different rate tuple
	    $\vec{q}_\pi = (Q_1,Q_2,\ldots,Q_m)_\pi \in \RR^m$ which is achievable in the limit of many copies of the state.
	    By time-sharing we can achieve any rate that lies in the \emph{convex hull} of these points.
	    We will show that the rate region for an input state $\ket{\ph}^{A_1\cdots A_mR}$ can equivalently be described
	    by the set of inequalities from Theorem \ref{thm:THM1}, that is
	    \be \label{inner-bound-biss}
	        \sum_{k\in \K} Q_k
	        \geq
	        \frac{1}{2} \left[ \sum_{k\in \K}\!H(A_k)_\ph  +  H(R)_\ph -  H(RA_{\K})_\ph \right] =: C_\K
	    \ee
	    where $\K \subseteq \{1,2,\ldots,m\}$ ranges over all subsets of participants and $C_\K$ is the name we give to
	    the constant on the right hand side of the inequality.
	    The proof of Theorem \ref{thm:THM1} proceeds in two steps. First we show the set of rate tuples
	    $\{\vec{q}_\pi\}$ is contained in the rate region and then we prove that the set of inequalities
	    \eqref{inner-bound-biss} is an equivalent description of the rates obtained by time sharing and 
	    resource wasting of the rates $\{\vec{q}_\pi\}$.

	
	    Consider the $m$-dimensional space of rate tuples $(Q_1,\cdots,Q_m) \in \RR^m$.
	    We begin by a formal definition of a corner point $\qq_\pi$.
	    \begin{definition}[Corner point]
	        Let $\pi \in S_m$ be a permutations of the senders in the protocol.
	        The corresponding rate tuple $q_\pi=(Q_1,Q_2,\ldots,Q_m)$ is a corner point if
	        \be     \label{individual-rates}
	            Q_{\pi\!(k)}    =   \frac{1}{2}I(A_{\pi\!(k)};A_{\pi\!(k+1)}\cdots A_{\pi\!(m)}R)
	        \ee
	        where the set $A_{\pi\!(k+1)}\cdots A_{\pi\!(m)}$ denotes all the systems which come after $k$ in
	        the permutation $\pi$.
	    \end{definition}
	
	    We define $\Q := \{ \qq_\pi : \pi \in S_m \}$, the set of all corner points. 
	    Clearly, $\vert \Q \vert \leq m!$ but since some permutations might lead to the same rate tuple, 
	    the inequality may be strict.
	
	    \begin{lemma}   \label{q-point}
	        The set of corner points, $\Q = \{ \qq_\pi : \pi \in S_m \}$, is contained in the rate region.
	    \end{lemma}
	
	    \begin{proof}[Proof sketch for Lemma \ref{q-point}]
	        We will now exhibit a protocol that achieves one such point. In order to simplify the notation,
	        but without loss of generality, we choose the reversed-order permutation $\pi=(m, \ldots, 2,1)$.
	        This choice of permutation corresponds to Alice-$m$ sending her information first and Alice-$1$ sending last.
	
	        We will repeatedly use the FQSW protocol is order to send the $m$ systems to Charlie:
	        \begin{enumerate}
	            \item       The first party Schumacher compresses her system $A_m$ and sends it to Charlie.
	                        She succeeds provided
	                        \be \nonumber
	                            Q_m     \geq    \frac{1}{2}I(A_m;A_1A_2\ldots A_{m-1}R) + \delta= H(A_m) + \delta
	                        \ee
	                        for any $\delta > 0$.
	                        The above rate is dictated by the FQSW inequality \eqref{eqn:FQSW} because
	                        we are facing the same type of problem except that the ``reference'' consists
	                        of $R$ as well as the remaining participants $A_1A_2\cdots A_{m-1}$.
	                        The fact that the formula reduces to $Q_m > H(A_m)$ should also be expected
	                        since there are no correlations that the first participant can take advantage of;
	                        she is just performing Schumacher compression.
	
	            \item       The second party also faces an instance of an FQSW problem.
	                        The task is to transmit the system $A_{m-1}$ to Charlie, who is now assumed to hold $A_m$.
	                        The purifying system consists of  $A_1A_2\cdots A_{m-2}R$.
	                        According to inequality (\ref{eqn:FQSW}) the rate must be
	                        \be \nonumber
	                            Q_{m-1} \geq    \frac{1}{2}I(A_{m-1};A_1A_2\cdots A_{m-2}R) + \delta
	                        \ee
	                        for any $\delta > 0$.
	
	            \item       The last person to be merging with Charlie will have a purifying system consisting
	                        of only $R$. Her transfer will be successful if
	                        \be \nonumber
	                            Q_1     \geq    \frac{1}{2}I(A_1;R) + \delta
	                        \ee
	                        for any $\delta > 0$.
	        \end{enumerate}
	
	        On the receiving end of the protocol, Charlie will apply the decoding map $\cD$ consisting of the
	        composition of the decoding maps $\cD_1 \circ \cD_2 \circ \cdots \circ \cD_m$ defined by the individual
	        FQSW steps to recover the
	        state $\sigma^{\widehat{A}_1\widehat{A}_2 \cdots \widehat{A}_m}$, which will be such that
	        the fidelity between $\ket{\psi}^{A^n R^n}$ and $\sigma^{\hat{A}^n R^n}$ is high, essentially by the triangle
	        inequality.
	        Finally, because we can make $\delta$ arbitrarily small, the rate tuple $(Q_1,\cdots,Q_m)$, with
	        \be
	            Q_k     =   \frac{1}{2}I(A_k; A_{1}\cdots A_{k\!-\!1}R),
	        \ee
	        must be contained in the rate region. The same argument applies for each permutation $\pi \in S_m$, leading
	        to the conclusion that the full set $\Q$ is contained in the rate region.
	    \end{proof}
	
	    Each one of the corner points $\qq_\pi$ can also be described by an equivalent set of equations involving
	    sums of the rates.
	    \begin{lemma} \label{lemmaWithL}
	        The rate tuple $(Q_1,Q_2,\ldots,Q_m)$ is a corner point if and only if for some $\pi \in S_m$ and for
	        all $l$ such that $1\leq l \leq m$,
	        {\small
	        \be \label{equationWithL}
	            \sum_{m-l+1 \leq k \leq m} Q_{\pi\!(k)} =    \frac{1}{2}
	                                                        \left[
	                                                            \sum_{m-l+1 \leq k \leq m} 
	                                                            \!\!\!\!\!\!\!\!H(A_{\pi\!(k)}) + H(R)
	                                                            - H(A_{\pi\![m-l+1, m]}R)
	                                                        \right]
	                                                        = C_{\pi\![m-l+1, m]}
	        \ee}
	        where $A_{\pi\![m-l+1, m]} := A_{\pi\!(m-l+1)}A_{\pi\!(m-l+2)}\cdots A_{\pi\!(m)}$ denotes
	        the last $l$ participants according to the permutation $\pi$.
	    \end{lemma}
	    \begin{proof}[Proof of Lemma \ref{lemmaWithL}]
	        The proof follows trivially from Lemma \ref{q-point} by considering sums of the rates.
	        If we again choose the permutation $\pi=(m, \ldots, 2,1)$ for simplicity, we see that the sum of the
	        rates of the last $l$ participants is
	        \begin{align}
	            Q_1 +  \cdots + Q_l         &= \frac{1}{2}\bigg[ I(A_1;R)
	                                                    + I(A_2;A_1R) +\cdots
	                                                    + I(A_l;A_1 \cdots A_{l-1}R) \bigg] \nonumber \\
	                            &=  \frac{1}{2}\bigg[
	                                                        \sum_{1 \leq k \leq l} H(A_k)  + H(R) -  H(A_1\cdots A_lR)
	                                                        \bigg] = C_{12\ldots l}.
	        \end{align}
	        A telescoping effect occurs and most of the inner terms cancel so we are left with a system of equations
	        identical to \eqref{equationWithL}. Moreover, this system is clearly solvable for
	        the individual rates $Q_k$. The analogous simplification occurs for all other permutations.
	    \end{proof}

	
	    So far, we have shown that the set
	    of corner points $\Q$ is contained in the rate region of the multiparty fully quantum Slepian-Wolf protocol.
	    The convex hull of a set of points $\Q$ is defined to be
	    \be \label{conv-def}
	        conv(\Q):= \left\{  \vec{x}  \in \RR^m :\
	                            \vec{x} = \sum \lambda_i \vec{q}_i, \
	                            \vec{q}_i \in \Q,\
	                            \lambda_i \geq 0,\
	                            \sum \lambda_i = 1
	                    \right\}.
	    \ee
	    Because of the possibility of time-sharing between the different corner points, the entire convex hull
	    $conv(\Q)$ must be achievable.
	    Furthermore, by simply allowing any one of the senders to waste resources, we know that if a rate tuple
	    $\qq$ is achievable, then so is $\qq + \vec{w}$ for any vector $\vec{w}$ with nonnegative coefficients.
	    More formally, we say that any $\qq + cone(\vec{e}_1, \vec{e}_2, \ldots, \vec{e}_m)$ is also inside the
	    rate region, where $\{\vec{e}_i\}$ is the standard basis for $\RR^m$:
	    $\vec{e}_i = (\underbrace{0, 0, \ldots, 0, 1}_{i}, 0, 0)$ and
	    \vspace{-0.4cm}
	    \be \label{cone-def}
	        cone(\vec{e}_1,\cdots,\vec{e}_m) :=
	                \left\{     \vec{x}  \in \RR^m : \
	                            \vec{x} = \sum \lambda_i \vec{e}_i,\
	                            \lambda_i \geq 0 \right\}.
	    \ee
	    Thus, we have demonstrated that the set of rates
	    \be
	        P_\mathcal{V} := conv(\Q) + cone(\vec{e}_1,\cdots,\vec{e}_m)
	    \ee
	    is achievable. 
	    To complete the proof of Theorem \ref{thm:THM1}, we will need to show that $P_\mathcal{V}$ has an equivalent
	    description as
	    \be \label{inner-bound-bisss}
	        P_\mathcal{H}:=     \left\{ (Q_1,\cdots,Q_m) \in \RR^m \ :\
	                                        \sum_{k\in \K} Q_k  \geq  C_\K ,
	                                        \forall \K \subseteq \{1,2,\ldots,m\}               \right\},
	    \ee
	    where the constants $C_\K$ are as defined in equation \eqref{inner-bound-biss}.
	    This equivalence is an explicit special case of the Minkowski-Weyl Theorem on convex polyhedra.
		
		\bigskip	
	    \begin{theorem}[Minkowski-Weyl Theorem] \cite[p.30]{poly}
	    For a subset $P \subseteq \RR^m$, the following two statements are equivalent:
	    \begin{itemize}
	        \item   $P$ is a $\mathcal{V}$-polyhedron: the sum of a convex hull of a finite set of points
	                $\Pp = \{ \vec{p}_i\}$ plus a conical combination of vectors $\Ww = \{\vec{w}_i\}$
	                \be \label{vpolyhedron}
	                    P = conv(\Pp) \ +\  cone(\Ww)
	                \ee
	                where $conv(\Pp)$ and $cone(\Ww)$ are defined in \eqref{conv-def} and \eqref{cone-def} respectively.
	
	
	        \item   $P$ is a $\mathcal{H}$-polyhedron: an intersection of $n$ closed halfspaces
	                \be \label{hpolyhedron}
	                    P = \{ \vec{x} \in \RR^m : A\vec{x} \geq \vec{a} \}
	                \ee
	                for some matrix $A \in \RR^{n\times m}$ and some vector $\vec{a} \in \RR^n$.
	                Each of the $n$ rows in equation \eqref{hpolyhedron} defines one halfspace.
	    \end{itemize}
	    \end{theorem}
	
	    \vspace{0.45cm}

	    \myparagraph{Preliminaries}
	    Before we begin the equivalence proof in earnest, we make two useful observations which will be
	    instrumental to our subsequent argument.
	    First, we prove a very important property of the constants $C_{\K}$ which will dictate
	    the geometry of the rate region.
	    \begin{lemma}[Superadditivity] \label{lemma-intersection-union}
	        Let $\K, \cL \subseteq \{1,2,\ldots,m\}$ be any two subsets of the senders. Then
	        \be \label{lower-order}
	            C_{\K\cup\cL} + C_{\K\cap\cL} \  \geq\  C_{\K} + C_{\cL}.
	        \ee
	    \end{lemma}
	    \begin{proof}[Proof of Lemma \ref{lemma-intersection-union}]
	        We expand the $C$ terms and cancel the $\frac{1}{2}$-factors to obtain
	        \begin{align} \nonumber
	        \begin{aligned}
	            &\sum_{k\in {\K\cup\cL} }\!H(A_k)  +  H(R) -  H(RA_{\K\cup\cL}) \\[-2mm]
	            &\ +  \sum_{k\in {\K\cap\cL} }\!H(A_k)  +  H(R) -  H(RA_{\K\cap\cL})
	        \end{aligned}
	                           &\geq
	        \begin{aligned}
	            &\sum_{k\in \K}\!H(A_k)  +  H(R) -  H(RA_{\K}) \\[-2mm]
	            &\ + \sum_{k\in \cL}\!H(A_k)  +  H(R) -  H(RA_{\cL}).
	        \end{aligned}
	        \end{align}
	        After canceling all common terms we find that the above inequality is equivalent to
	        %
	        \be                                                                                   %
	            H(RA_{\K}) +  H(RA_{\cL})
	                  \ \ \ \geq \ \ \
	                                H(RA_{\K\cup\cL}) + H(RA_{\K\cap\cL}),
	        \ee
	        which is true by the strong subadditivity (SSA) inequality of quantum entropy~\cite{LR73}.
	    \end{proof}

	    As a consequence of this lemma, we can derive an equivalence property for the saturated inequalities.
	    \begin{corollary} \label{meet-join-lemma}
	        Suppose that the following two equations hold for a given point of $P_\mathcal{H}$:
	        \be
	            \sum_{k\in \K} Q_k = C_{\K}             \qquad \text{and} \qquad  \sum_{k\in \cL} Q_k = C_{\cL}.
	        \ee
	        Then the following equations must also be true:
	        \be
	            \sum_{k\in \K\cup\cL} Q_k = C_{\K\cup\cL} \qquad  \text{and} 
	            \qquad  \sum_{k\in \K\cap\cL} Q_k = C_{\K\cap\cL}.
	        \ee
	    \end{corollary}
	    \begin{proof}[Proof of Corollary \ref{meet-join-lemma}]
	        The proof follows from the equation
	        \be \label{chained-eqn-meet-join-lemma}
	            \sum_{k\in \K} Q_k + \sum_{k\in \cL} Q_k
	            =           C_{\K} + C_{\cL}
	            \ \leq\     C_{\K\cup\cL} + C_{\K\cap\cL}
	            \ \leq\     \sum_{k\in \K\cup\cL} Q_k + \sum_{k\in \K\cap\cL} Q_k
	        \ee
	        where the first inequality comes from Lemma \ref{lemma-intersection-union}.
	        The second inequality is true by the definition of $P_\mathcal{H}$ since $\K\cup\cL$ and $\K\cap\cL$
	        are subsets of $\{1,2,\ldots,m\}$.
	        Because the leftmost terms and rightmost terms are identical, we must have equality
	        throughout equation \eqref{chained-eqn-meet-join-lemma}, which in turn implies the the union and the
	        intersection equations are saturated.
	    \end{proof}
	
		An important consequence of Lemma~\ref{lemma-intersection-union}  is that it implies that the polyhedron
		$P_\mathcal{H}$ has a very special structure. It is known as a supermodular polyhedron or
		contra-polymatroid. 
		The fact that $conv(Q) = P_\mathcal{H}$ was proved by Edmonds \cite{E69},
		whose ingenious proof makes use of linear programming duality. Below we give an elementary proof 
		that does not use duality.

	    A \emph{vertex} is a zero-dimensional face of a polyhedron. 
		A point $ \bar{Q} = ( \bar{Q}_1,\bar{Q}_2,\ldots,\bar{Q}_m ) \in P_\mathcal{H} \subset \RR^m$
		is a vertex of $P_\mathcal{H}$ if and only if it is the unique solution of a set of linearly
		independent equations
	    \be     \label{original-set-of-equations}
	            \sum_{k\in \cL_i} Q_k        =  C_{\cL_i}, \qquad \qquad 1 \leq i \leq m
	    \ee
	    for some subsets $\cL_i \subseteq \{1,2,\ldots,m\}$. 
	    In the remainder of the proof we require only a specific consequence of linear independence,
	    which we state in the following lemma.
	    \begin{lemma}[No co-occurrence] \label{no-co}
	        Let $\cL_i \subseteq \{1,2,\ldots,m\}$ be a collection of $m$ sets such that the system
	        \eqref{original-set-of-equations} has a unique solution. 
	        Then there is no pair of elements $j$, $k$ such that $j \in \cL_i$ if and only if  $k \in \cL_i$  for all $i$.
	    \end{lemma}
	    \begin{proof}
	        If there was such a pair $j$ and $k$, then the corresponding columns of the left hand side of 
	        \eqref{original-set-of-equations} would be linearly dependent.
	    \end{proof}
	
	    Armed with the above tools, we will now show that there is a one-to-one correspondence between the
	    corner points $\Q$ and the vertices of the $\mathcal{H}$-polyhedron $P_\mathcal{H}$.
	    We will then show that the vectors that generate the cone part of the $\mathcal{H}$-polyhedron
	    correspond to the resource wasting vectors $\{\vec{e}_i\}$.
	
	    \vspace{0.45cm}

	    \myparagraph{Step 1: $\Q \subseteq vertices(P_\mathcal{H})$}
	    We know from Lemma \ref{lemmaWithL} that every point $\qq_\pi \in \Q$ satisfies the $m$ equations
		\begin{align}	\label{special-set-of-equations}
	        \sum_{m-i+1 \leq k \leq m} Q_{\pi\!(k)} &=   C_{\pi\![m-i+1, m]}, & 1\leq &i \leq m.
		\end{align}
	
	
	
		The equations \eqref{special-set-of-equations} are linearly independent since the left hand
		side is triangular, and have the form of the inequalites in \eqref{inner-bound-bisss}
		that are used to define $P_\mathcal{H}$. They have the unique solution:
		\begin{align}	\label{solution}
		 Q_{\pi\!(m)} &=  C_{\pi(m)}  &	
		 Q_{\pi\!(i)} &=   C_{\pi\![i, m]} - C_{\pi\![i+1, m]}, \qquad 1\leq i \leq m-1.
		\end{align}
		We need to show that this solution satisfies all the inequalities used to define
		$P_\mathcal{H}$ in \eqref{inner-bound-bisss}. We proceed by induction on $|\K|$.
		The case $|\K|=1$ follows from \eqref{solution} and the superadditivity property
		\eqref{lower-order}.
		For $|\K| \ge 2$ we can write $\K= \{\pi(i)\} \cup \K'$ for some 
		$\K' \subseteq \{ \pi(i+1), \pi(i+2),\ldots,\pi(m) \}$. Then
		\beas
		\sum_{k \in \K} Q_k 	&=& 	Q_{\pi(i)} + \sum_{k \in \K'} Q_k \\
								&\ge&  	C_{\pi\![i, m]} - C_{\pi\![i+1, m]} +  \sum_{k \in \K'} Q_k \\
							&\ge&  	C_{\pi\![i, m]} - C_{\pi\![i+1, m]} + C_{\K'} \qquad\qquad \textrm{(induction)} \\
								&\ge& 	C_{\K}
		\eeas
		where we again used superadditivity to get the last inequality.

	    \vspace{0.45cm}

	    \myparagraph{Step 2: $vertices(P_\mathcal{H}) \subseteq \Q$} In order to prove the opposite inclusion,
	    we will show that every vertex of $P_\mathcal{H}$ is of the form of Lemma
	    \ref{lemmaWithL}. More specifically, we want to prove the following proposition.
	    \begin{proposition}[Existence of a maximal chain] \label{flag-existence}
	        Every vertex of $P_\mathcal{H}$, that is, the intersection of $m$ linearly independent hyperplanes
	        \begin{align}
	            							\qquad\qquad\qquad\qquad\qquad 
	            \sum_{k\in \cL_i} Q_k         &=  C_{\cL_i},       & 1\leq &i \leq m,\\[-3mm]
	        \intertext{	defined by the family of sets $\{\cL_i; \, 1 \leq i \leq m\}$ can be 
	        			described by an equivalent set of equations}
	            							\qquad\qquad\qquad\qquad\qquad 
	            \sum_{k\in \K_i} Q_k         &=  C_{\K_i},       & 1\leq &i \leq m,
	        \end{align}
	        for some family of sets distinct $\K_i \subseteq \{1,2,\ldots,m\}$ that form a \emph{maximal chain} 
	        in the sense of
	        \be
	            \emptyset = \K_0 
	            \subset \K_1 
	            \subset \K_2 
	            \subset \cdots 
	            \subset \K_{m-1} 
	            \subset \K_m = \{1,2,\ldots,m\}.
	        \ee
	    \end{proposition}
	    Since there exists a permutation $\pi$ such that $\forall i,\ \pi\![m-i+1, m] = \K_i$ this implies that
	    all the vertices of $P_\mathcal{H}$ are in $\Q$. The main tool we have have at our disposal in order to prove
	    this proposition is Corollary \ref{meet-join-lemma}, which we will use extensively.
	
	
	
	    \begin{proof}[Proof of Proposition \ref{flag-existence}]
	
	    Let $\{ \cL_i \}_{i=1}^{m}$ be the subsets of $\{1,2,\ldots,m\}$ for which the inequalities are saturated
	    and define $\cL^\mathcal{S}_i := \cL_i \cap \mathcal{S}$, the intersection of $\cL_i$ with some set
	    $\mathcal{S} \subseteq \{1,2,\ldots,m\}$.
	
	    \noindent Construct the directed graph $G = (V, E)$, where:
	    \begin{itemize}
	        \item   $V = \{1,2,\ldots,m\}$, i.e. the vertices are the numbers from $1$ to $m$;
	
	        \item   $E = \left\{(j, k) \ :\  (\forall i) \; j \in \cL_i \implies k \in \cL_i \ \right\}$,
	                i.e. there is an edge from vertex $j$ to vertex $k$ if  whenever vertex $j$ occurs
	                in the given subsets, then so does vertex $k$.
	    \end{itemize}
	    Now $G$ has to be acyclic by Lemma \ref{no-co}, so it has a topological sorted order. 
	    Let us call this order $\nu$.  Let $\K_0 = \emptyset$ and let
	    \be
	        \K_l = \{\nu_{m-l+1}, \ldots, \nu_m \}
	    \ee
	    for $l \in \{ 1, \ldots , m \}$.
	    The sets $\K_l$, which consist of the last $l$ vertices according to the ordering $\nu$,
	    form a maximal chain $\K_0 \subset \K_1 \subset \cdots \subset \K_{m-1} \subset \K_m$
	    by construction.
	
	    We claim that all the sets $\K_l$ can be constructed from the sets $\{ \cL_i \}$ by using unions
	    and intersections as dictated by Corollary \ref{meet-join-lemma}.
	    The statement is true for $\K_m=\{1,2,\ldots,m\}$ because every variable must appear in some constraint
	    equation, giving $\K_m = \cup_i \cL_i$.
	    The statement is also true for $\K_{m-1}=\{\nu_2, ..., \nu_m\}$ since the vertex $\nu_1$ has no
	    in-edges in $G$ by the definition of a topological sort, which means that
	    \be \label{build-Km1}
	      \K_{m-1} = \bigcup_{\nu_1 \notin \cL^{\K_m}_i} \cL^{\K_m}_i.
	    \ee
	    %
	    %
	    %
	    For the induction statement, let $l \in \{m-1, \ldots, 2, 1\}$ and
	    assume that $\K_l = \bigcup_{i} \cL^{\K_{l}}_i$. 
	    Since the vertex $\nu_{m-l}$ has no in-edges in the induced subgraph generated by the vertices $\K_l$
	    by the definition
	    of the topological sort, $\K_{l-1}$ can be obtained from the union of all the sets not containing $\nu_{m-l}$:
	    \be
	        \K_{l-1} = \bigcup_{\nu_{m-l} \notin \cL^{\K_l}_i} \cL^{\K_l}_i.
	    \ee
	    In more detail, we claim that for all $\omega \neq \nu_{m-l} \in \K_{l-1}$ there exists $i$ such that
	    $\nu_{m-l} \not\in \cL_i^{\K_l}$ and $\omega \in \cL_i^{\K_l}$. If it were not true, that would imply the existence
	    of $\omega \neq \nu_{m-l} \in \K_{l-1}$ such that for all $i$, $\nu_{m-l} \in \cL_i^{\K_l}$ or
	    $\omega \not\in \cL_i^{\K_l}$. This last condition implies that whenever $\omega \in \cL_i^{\K_l}$
	    it is also true that $\nu_{m-l} \in \cL_i^{K_l}$,
	    which corresponds to an edge $(\omega,\nu_{n-l})$ in the induced subgraph.
	    \end{proof}
	
	    We have shown that every vertex can be written in precisely the same form as Lemma \ref{lemmaWithL}
	    and is therefore a point in $\Q$.
	    This proves $vertices(P_\mathcal{H}) \subseteq \Q$, which together with the result of Step 1, implies
	    $vertices(P_\mathcal{H}) = \Q$.

	    \vspace{0.45cm}

	    \myparagraph{Step 3: Cone Part}
	    The final step is to find the set of direction vectors that correspond to the cone part of $P_\mathcal{H}$.
	    The generating vectors of the cone are all vectors that satisfy the homogeneous versions of
	    the halfspace inequalities \eqref{hpolyhedron}, which in our case gives
	    \be
	        \sum_{k \in \K} Q_k \geq 0
	    \ee
	    for all $\K \subset \{ 1, 2, \ldots, m \}$. These inequalities are satisfied if and only if $Q_k \geq 0$
	    for all $k$. We can therefore conclude that the cone part of $P_\mathcal{H}$ is
	    $cone(\vec{e}_1,\vec{e}_2,\ldots,\vec{e}_m)$.
	
	    \vspace{0.45cm}

	    This completes our demonstration that $P_\mathcal{V}$ is the $\mathcal{V}$-polyhedron description of the
	    $\mathcal{H}$-polyhedron $P_\mathcal{H}$.
	    Thus we arrive at the statement we were trying to prove; if the inequalities
	    \be
	        \sum_{k\in \K} Q_k
	                    \geq    C_{\K}
	                    =       \frac{1}{2} \left[ \sum_{k\in \K}\!H(A_k)_\ph  +  H(R)_\ph -  H(RA_{\K})_\ph \right]
	    \ee
	    are satisfied for any $\K \subseteq \{1,2,\ldots,m\}$, then the rate tuple $(Q_1,Q_2,\cdots,Q_m)$ is 
	    inside the rate region. This completes the proof of Theorem \ref{thm:THM1}.
	
	
		%

%
%

	\bigskip
	\bigskip	
	\section{Proof of outer bound}	\label{sec:THMIIproof}    
	
	    We want to show that \emph{any} distributed compression protocol which works must satisfy all of the 
	    inequalities \eqref{outer-bound} from Theorem \ref{thm:THM2}. 
	    In order to prove this, we will use some of the properties of multiparty information and 
	    squashed entanglement.  
	    We break up the proof into three steps.
	
		\bigskip	
	    \myparagraph{Step 1: Decoupling Formula}
	    We know that the input system $\ket{\psi}^{A^nR^n}$ is a pure state.  
	    If we account for the Stinespring dilations of each encoding and decoding operation, then we can 
	    view any protocol
	    as implemented by unitary transformations with ancilla and waste.
	    Therefore, the output state (including the waste systems) should also be pure.
								\index{encoding operation@$\E_i$: encoding operation}
								\index{waste system@$W_i$: waste system associated with $\E_i$}	
	    More specifically, the encoding operations are modeled by CPTP maps $\E_i$ with outputs $C_i$ 
	    of dimension $2^{nQ_i}$.
	    In our analysis we will keep track of the purification (waste) systems $W_i$ of the the Stinespring dilations
	    $\E_i$, so the evolution as a whole will be unitary. 
	    \[
	    \Qcircuit @C=1em @R=.7em {
	       \lstick{A_i}      & \multigate{1}{\E_i} & \qw & \rstick{C_i \quad  \leftarrow\text{to Charlie}} \qw \\
	       \lstick{ \ket{0}} & \ghost{\E_i}        & \qw & \rstick{W_i \quad  \! \leftarrow\text{waste}} \qw
	    }
	    \]

	    Once Charlie receives the systems that were sent to him, he will apply a decoding CPTP map $\cD$ with output
	    system 
	    $\widehat{A}=\widehat{A}_1\widehat{A}_2 \ldots \widehat{A}_m$ isomorphic to the original $A=A_1A_2\ldots A_m$.
	    \[
	    \Qcircuit @C=1em @R=.7em {
	       \lstick{\bigcup_{i}^m C_i}           & \multigate{1}{\cD}      & \qw &
	            \rstick{\widehat{A}_1\cdots\widehat{A}_m \quad \leftarrow\text{near-purification of $R$}} \qw \\
	       \lstick{ \ket{0}}    & \ghost{\cD}   & \qw & \rstick{W_C \qquad\qquad \!\!\!\leftarrow
	       \text{Charlie's waste}} \qw
	    }
	    \]

	    In what follows we will use Figure \ref{fig:mpFQSW} extensively in order to keep track of the evolution and 
	    purity of the states at various points in the protocol.
	
	    \begin{figure}[ht]  \begin{center}
	        \input{figures/figureBIGmpFQSW.pst} \end{center}
			\FigureCaptionOpt{ Detailed diagram of the distributed compression circuit.}{
	            A general distributed compression circuit diagram showing the encoding operations $\E_i$ with
	            output systems $C_i$ (compressed data) and $W_i$ (waste). The decoding operation takes all
	            the compressed data $\bigotimes_i\!C_i$ and applies the decoding operation $\cD$ to output a state
	            ${\sigma}^{\widehat{A}^nR^n}$ which has high fidelity with the original $\ket{\psi}^{A^nR^n}$. }
	        \label{fig:mpFQSW}
	    \end{figure}
	
	    The starting point of our argument is the fidelity condition (Definition~\ref{achievable}) for successful
	    distributed compression, which we restate below for convenience
	    \be
	        F\left( \ket{\psi}^{A^nR^n}\!,\ {\sigma}^{\widehat{A}^nR^n} \right)     \geq    1 - \epsilon
	    \ee
	    where $\ket{\psi}^{A^nR^n} = \left( \ket{\ph}^{A_1A_2\cdots A_mR} \right)^{\otimes n}$ 
	    is the input state to the protocol and $\sigma^{\widehat{A}^nR^n}$ 
	    is the output state of the protocol.
	    Since $\sigma^{\widehat{A}^nR^n}$ has high fidelity with a rank one state, it must have one large eigenvalue
	    \be
	        \lambda_{\rm max}(\sigma^{\widehat{A}^nR^n}) \geq 1 - \epsilon.
	    \ee
	    Therefore, the full output state $\ket{\sigma}^{\widehat{A}R^nW_1\!\cdots W_m W_C}$ has Schmidt decomposition
	    of the form
	    \be
	        \ket{\sigma}^{\widehat{A}^nR^nW_1\!\cdots W_m W_C}
	            =
	                \sum_i \sqrt{\lambda_i} \ket{e_i}^{\widehat{A}^nR^n} \!\!\otimes \ket{f_i}^{W_1\!\cdots W_m W_C},
	    \ee
	    where $\ket{e_i},\ket{f_i}$ are orthonormal bases and $\lambda_1=\lambda_{\rm max} \geq 1 - \epsilon$.

	    Next we show that the output state $\ket{\sigma}^{\widehat{A}^nR^nW_1\!\cdots W_m W_C}$
	    is very close in fidelity to a totally decoupled state 
	    $\sigma^{\widehat{A}^nR^n}\otimes\sigma^{W_1\cdots W_mW_C}$,
	    which is a tensor product of the marginals of $\ket{\sigma}$ on the subsystems
	    ${\widehat{A}^nR^n}$ and ${W_1\cdots W_mW_C}$:
	    \begin{align}
	        F\big(\ket{\sigma}^{\widehat{A}^nR^nW_1\!\cdots W_m W_C}&, \
	            \sigma^{\widehat{A}^nR^n}\otimes\sigma^{W_1\cdots W_mW_C} \big)  = \nonumber \\
	            &=\ \Tr\left[ \ket{\sigma}\!\!\bra{\sigma}^{\widehat{A}^nR^nW_1\!\cdots W_m W_C}\left(
	                \sigma^{\widehat{A}^nR^n}\otimes\sigma^{W_1\cdots W_mW_C}
	              \right)\right] \nonumber \\
	            &=\ \sum_{i}    \lambda^3_i  \geq\  (1-\epsilon)^3 \geq 1 - 3\epsilon. \label{fidelity_close}
	    \end{align}
	    Using the relationship between fidelity and trace distance \cite{Fuchs}, we can transform \eqref{fidelity_close}
	    into the trace distance bound
	    \be
	        \left\| \ket{\sigma}\!\!\bra{\sigma}^{\widehat{A}^nR^nW_1\!\cdots W_m W_C} -
	                    \sigma^{\widehat{A}^nR^n} \otimes \sigma^{W_1\!\cdots W_m W_C} \right\|_1 \leq 2\sqrt{3\epsilon}.
	    \ee
	    By the contractivity of trace distance, the same equation must be true for any subset of the systems.
	    This bound combined with the Fannes inequality implies that the entropies taken with respect to
	    the output state are nearly additive:
	    \bea
	        \big\vert H(R^n W_\K)_\sigma    \ -\  H(R^n)_\sigma + H(W_\K)_\sigma    \big\vert
	        &\leq &     2\sqrt{3\epsilon} \log(d_{R^n}d_{W_\K}) + \eta(2\sqrt{3\epsilon}) \nonumber \\
	        &\leq &     2\sqrt{3\epsilon} \log(d_{A^n}d_{A^{2n}_\K}) + \eta(2\sqrt{3\epsilon}) \nonumber \\
	        &\leq &     2\sqrt{3\epsilon}\ n \log(d^3_{A}) + \eta(2\sqrt{3\epsilon}) \nonumber \\
	        &=:& 		f_1(\epsilon,n).  \label{wastereference}
	    \eea
	    for any subset $\K \subseteq \{ 1,2\ldots m \}$ with $\epsilon \leq \frac{1}{12e^2}$ and $\eta(x)=-x\log x$.
	    In the second line we have used the fact that $d_A = d_R$ and exploited the fact that
	    $d_{W_\K}$ can be taken less than or equal to $d_{A^{2n}_\K}$, the maximum size of an environment
	    required for a quantum operation with inputs and outputs of dimension no larger than $d_{A^{n}_\K}$.

		\bigskip	
	    \myparagraph{Step 2: Dimension Counting}
	    The entropy of any system is bounded above by the logarithm of its dimension.
	    In the case of the systems that participants send to Charlie, this implies that
	    \be \label{c-dimension-bound}
	        n \sum_{k\in \K} Q_k \geq H(C_\K)_{\psi'}.
	    \ee
	    We can use this fact and the diagram of Figure \ref{fig:mpFQSW} to bound the rates $Q_i$.
	    First we add $H(A_\Kbar)_{\psi} = H(A_\Kbar)_{\psi'}$ to both sides of equation (\ref{c-dimension-bound}) and
	    obtain the inequality
	    \be     \label{useful2}
	        H(A_\Kbar)_{\psi}  +  n \sum_{k\in \K} Q_k
	        \geq
	        H(C_\K)_{\psi'} +H(A_\Kbar)_{\psi'} \geq H(C_\K A_\Kbar)_{\psi'}.
	    \ee
	    For each encoding operation, the input system $A_i$ is unitarily related to the outputs $C_iW_i$ so we can write
	    \be   \label{thisistight}
	        H(A_i)_\psi = H(W_iC_i)_{\psi'} \leq H(W_i)_{\psi'} + H(C_i)_{\psi'} \leq H(W_i)_{\psi'} + nQ_i,
	    \ee
	    where in the last inequality we have used the dimension bound $H(C_i) \leq nQ_i$.
	    If we collect all the $Q_i$ terms from equations (\ref{useful2}) and (\ref{thisistight}), 
	    we obtain the inequalities
	    \bea
	        n\sum_{i\in\K} Q_i  &\geq&      H(C_\K A_\Kbar)_{\psi'} - H(A_\Kbar)_{\psi} \label{useful2rewrite} \\
	        n\sum_{i\in\K} Q_i  &\geq&  \sum_{i\in\K}  H(A_i)_\psi  - \sum_{i\in\K} H(W_i)_{\psi'}. \label{sumoftight}
	    \eea
	    Now add equations (\ref{useful2rewrite}) and (\ref{sumoftight}) to get
	    \begin{align}
	    2 n\sum_{i\in\K} Q_i    &\geq^{\ \!\ \ \ }\quad \sum_{i\in\K}  H(A_i)_\psi  - \sum_{i\in\K} H(W_i)_{\psi'}
	                                                +  H(C_\K A_\Kbar)_{\psi'} - H(A_\Kbar)_\psi \nonumber \\
	                            &=^{(1)}\quad       \sum_{i\in\K}  H(A_i)_\psi  - \sum_{i\in\K} H(W_i)_{\psi'}
	                                                +  H(W_\K R^n)_{\psi'} - H(R^n A_\K)_\psi \nonumber  \\
	                            &\geq^{(2)}\quad    \sum_{i\in\K}  H(A_i)_\psi  - \sum_{i\in\K} H(W_i)_{\psi'}
	                                                +  H(W_\K)_{\psi'} + H(R^n)_{\psi'} \nonumber \\[-0.5cm]
								& \hspace{7.8cm} 		- H(R^n A_\K)_\psi - f_1(\epsilon,n) \nonumber  \\
	                            &=^{\ \ }\quad      \left[ \sum_{i\in\K} H(A_i) + H(R^n) - H(R^n A_\K) \right]_\psi
	                                                +  H(W_\K)_{\psi'} \nonumber \\[-0.5cm]
								& \hspace{7.8cm}		- \sum_{i\in\K} H(W_i)_{\psi'} - f_1(\epsilon,n), 
													\label{dirty-rate-bound}
	    \end{align}
	    where the equality $\!\!\phantom|^{(1)}$ comes about because the two systems
	    $\ket{\psi}^{A_\K A_\Kbar R^n}$ and $\ket{\psi'}^{C_\K W_\K A_\Kbar R^n}$ are pure.
	    The inequality (\ref{wastereference}) from Step 1 was used in $\!\!\phantom|^{(2)}$.

		\bigskip	
	    \myparagraph{Step 3: Squashed Entanglement}
	    We would like to have a bound on the extra terms in equation \eqref{dirty-rate-bound} that does not depend on the
	    encoding and decoding maps.
	    We can accomplish this if we bound the waste terms $\sum_{i \in \K} H(W_i)_{\sigma}  - H(W_\K)_{\sigma}$
	    by the squashed entanglement $2\Esq( A_{k_1};\cdots;A_{k_l} )_\psi$ of the input state
	    for each $\K = \{k_1,k_2,\ldots,k_l\} \subseteq \{1,\ldots,m\}$ plus some small
	    corrections.
	    The proof requires a continuity statement analogous to \eqref{wastereference},
	    namely that
	    \be
	        \big\vert H(W_i) - H(W_i|R)\big\vert \leq   f_2(\epsilon,n)     \label{conditioning-on-R}
	    \ee
	    where $f_2$ is some function such that $f_2(\epsilon,n)/n \rightarrow 0$ as  $\epsilon \rightarrow 0$.
	    The proof is very similar to that of \eqref{wastereference} so we omit it.

	    Furthermore, if we allow an arbitrary transformation  $\cN^{R \to E}$ to be applied to the $R$ system, we will
	    obtain some general extension but the analog of equation (\ref{conditioning-on-R}) will remain true
	    by the contractivity of the trace distance under CPTP maps. We can therefore write:
	    \begin{align*} 
	        \sum_{i\in\K} &H(W_i)_\psi   -  H(W_\K)_\psi   \\
	             &\leq     \sum_{i\in\K} H(W_i|E)  - H(W_\K|E)
	                                                                     + [|\K|+1]f_2(\epsilon,n) \\
	                    &=                I(W_{k_1};W_{k_2};\ldots;W_{k_l};E)
	                          - I(W_{k_1};E) - \!\!\!\!\sum_{i\in \{ \K  \setminus k_1\} } I(W_i;E) + f'_2(\epsilon,n)\\
	                    &=^{(1)}      I(W_{k_1}E;W_{k_2};\ldots;W_{k_l})
	                                        - \sum_{i\in \{ \K  \setminus k_1\} } I(W_i;E)  + f'_2(\epsilon,n)\\
	                    &\leq^{(2)}   I(A_{k_1}E;W_{k_2};\ldots;W_{k_l})
	                                        - \sum_{i\in \{ \K  \setminus k_1\} } I(W_i;E)  + f'_2(\epsilon,n)\\
	                    &=^{(1)}      I(A_{k_1};W_{k_2};\ldots;W_{k_l},E) - I(A_{k_1};E)
	                                        - \sum_{i\in \{ \K  \setminus k_1\} } I(W_i;E)  + f'_2(\epsilon,n)\\
	                    &\leq^{(3)}       I(A_{k_1};A_{k_2};\ldots;A_{k_l};E) - \sum_{i\in\K} I(A_i;E)    
	                    	+ f'_2(\epsilon,n)\\
	                    &\leq             I(A_{k_1};A_{k_2};\ldots;A_{k_l}|E) + f'_2(\epsilon,n),
	    \end{align*}
	    where we have used the shorthand $f'_2(\epsilon,n) = [|\K|+1]f_2(\epsilon,n)$ for brevity.
	    Equations marked $\!\!\phantom|^{(1)}$ use Lemma \ref{mergingJ} and inequality $\!\!\phantom|^{(2)}$ comes
	    about from Lemma \ref{cond-monotonicity}, the monotonicity of the multiparty information.
	    Inequality $\!\!\phantom|^{(3)}$ is obtained when we repeat the steps for $k_2,\ldots,k_l$.
	    The above result is true for any extension $E$ but we want to find the tightest possible lower bound for the
	    rate region so we take the infimum over all possible extensions $E$ thus arriving at the definition of squashed
	    entanglement.
	    \index{entanglement!squashed@$\Esq$: squashed entanglement}

	    \ \\
	
	    \noindent Putting together equation \eqref{dirty-rate-bound} from Step 2 and the bound from Step 3 we have
	    \begin{align*}
	        2 n\sum_{i\in\K} Q_i    &\geq   \left[ \sum_{i\in\K} H(A_i) + H(R^n) - H(R^n A_\K) \right]_\psi \\[-0.4cm]
	        						&\hspace{4.9cm} 	- \left(\sum_{i\in\K} H(W_i)_{\psi'}-H(W_\K)_{\psi'} \right) 
	        										- f_1(\epsilon,n) \nonumber \\
	                                &\geq   \left[ \sum_{i\in\K} H(A_i) + H(R^n) - H(R^n A_\K) \right]_\psi \\[-0.4cm]
	                                &\hspace{4.5cm}	-  2\Esq( A_{k_1};\cdots;A_{k_l} )_\psi
	                                				- f_1(\epsilon,n) - f'_2(\epsilon,n). 
	    \end{align*}
	    We can simplify the expression further by using the fact that $\ket{\psi} = \ket{\ph}^{\otimes n}$
	    to obtain
	    \begin{align*}
	        \sum_{k\in \K} Q_k
	        &\geq   \frac{1}{2}\left[ \sum_{k\in \K} H(A_k) + H(R) - H(RA_{\K}) \right]_\ph	\\[-0.4cm]
	        &\hspace{4.4cm} 	- \Esq({A_{k_1};A_{k_2};\ldots A_{k_l}})_\ph	
	                        - \frac{f_1(\epsilon,n)}{2n}
	                        - \frac{f'_2(\epsilon,n)}{2n} 
	    \end{align*}
	    where the we used explicitly the additivity of the entropy for tensor product states and the subadditivity of
	    squashed entanglement demonstrated in Proposition \ref{subadditiveTensorProducts}.
	
	    Theorem \ref{thm:THM2} follows from the above since $\epsilon > 0$ was arbitrary and the sum
	    $(f_1(\epsilon,n) + f'_2(\epsilon,n))/n \rightarrow 0$ as  $\epsilon \rightarrow 0$. \qed

	\pagebreak
	\section{Discussion}		\label{section:discussion}

		The multiparty fully quantum Slepian-Wolf protocol is an optimal solution to the distributed compression
		problem for separable states, i.e. states of the form
		\be	\nonumber
			\ph^{X_1\cdots X_m}	=	\sum_i	p_i	\ph_i^{X_1} \otimesc \ph_i^{X_2} \otimes\!\cdots \otimes \ph_i^{X_m},
		\ee
		because $\Esq =0$ for such states.
		For general states, we have provided an outer bound on the set of achievable rates based on the
		multiparty squashed entanglement.
		In this section, we outline some other aspects of the multiparty FQSW protocol and its relation to other
		protocols.

																		\index{information!multiparty}
		First, we note that there is an alternative, more compact way of writing the rate sum inequalities of 
		Theorem~\ref{thm:THM1} and Theorem~\ref{thm:THM2}. Consider the inequalities of
		the inner bound \eqref{inner-bound} reproduced below:
		\be	\label{inner-bound-discussion}
           \sum_{k\in \K}\! Q_k \ \geq \ 
	           	\frac{1}{2}\!\left[ \sum_{k\in \K}\!H(A_k)  +  H(R) -  H(RA_{\K}) \right],	
	           	\quad \forall\K \subseteq \{1,\ldots,m\}.
		\ee
		The term on the right hand side can be expressed as a multiparty information
		\be	\label{mi-inner-bound}
           \sum_{k\in \K} Q_k \geq  
	           	\frac{1}{2}I(A_{{\displaystyle ;}\K};R),	\qquad \forall\K \subseteq \{1,\ldots,m\},
		\ee
		where $I(A_{{\displaystyle ;}\K};R)$ is the multiparty information of all the members of $\K$ and $R$.
		The multiparty information function is naturally suited to the multiparty distributed compression problem.

		When only two parties are involved ($m=2$), the inequalities in \eqref{inner-bound-discussion}
	    reduce to the two-party bounds on distributed compression presented in \cite{FQSW}:
	    \be
	    \begin{aligned} 
	            Q_1       &\geq \frac{1}{2} I({A_1};R), \\
	            Q_2       &\geq \frac{1}{2} I(A_2;R), \label{eqn:region-bis} \\
	            Q_1 + Q_2 &\geq \frac{1}{2}\left[ H({A_1}) + H({A_2})+ H({A_1}{A_2}) \right].
	    \end{aligned}
	    \ee
	    However, we now understand the mystery behind the expression that looks like the mutual information
	    with a reversed sign: it is simply the form $\frac{1}{2}I(A_1;A_2;R)$, where $H(R)=H(A_1A_2)$
	    and $H(A_1A_2R)=0$.
	    The outer bound inequalities \eqref{outer-bound} similarly reduce to the corresponding expressions in 
		the FQSW paper \cite{FQSW} with the multiparty squashed entanglement being replaced by the original 
		two-party squashed entanglement of \cite{CW04}.

								\index{protocol!state merging}
		Another observation concerns the classical communication cost of the protocol.
		If we move away from the ``fully quantum'' regime and allow classical communication between
		the senders and the receiver we can achieve better rates.
		We do this by recycling the entanglement generated by the FQSW protocol.
		For two parties, the combination of multiparty FQSW of equation \eqref{eqn:region-bis} with teleportation
		reproduces  the state merging results of equation \eqref{state-merging-inequalities} 
		\be					\label{state-merging-inequalities-bis}
			\begin{aligned}
				R_A 		&>	 	\ H(A|B)_\rho,	\\
				R_B 		&>		\ H(B|A)_\rho,	\\
				R_A+R_B		&>		\ H(AB)_\rho.	
			\end{aligned}
		\ee	

		Finally we note that the multiparty FQSW protocol can be operated backwards in time to produce
		an optimal reverse Shannon theorem for the quantum broadcast channel \cite{DH07}.


\chapter{Possible applications to the black hole information paradox}	\label{chapter:applications-to-black-holes}

	There are very few physical systems that require both the application of the principles of general relativity and
	of quantum mechanics in order to understand them. Black holes fall into this category.
	Classically, a black hole is a region of space where the gravity is so strong that nothing can escape its 
	pull -- not even light.
	However, according to a certain semi-classical calculation performed by Hawking \cite{hawking2}, black holes 
	emit thermal radiation at a very slow rate. 
	Thus, while it may take a very long time, all the mass/energy that fell into the black hole will eventually 
	be released back into the universe and the black hole will evaporate.
	
	This scenario poses a serious problem known as the black hole information paradox.
	Consider a universe originally in the pure state $\ket{\textrm{Universe}}$ which collapses onto itself to form a
	black hole. After a very long time, the black hole evaporates completely to leave behind a universe 
	filled with thermal radiation, which corresponds to the maximally mixed state.
	Herein lies the paradox: an initially pure state has evolved to a mixed state --- something which 
	violates the laws of unitary evolution so central to all of quantum theory. 
	
	Does gravity lead to non-unitary evolution or is general relativity incomplete?
	Over the last 30 years, many preeminent physicists have had something to say about this question
	and yet this paradox still defies explanation \cite{preskill,thaschen,page2}.
	What is worse is that the more we think about the information paradox the more we realize that it is 
	not an ``unwarranted extrapolation from an untrustworthy approximation''\cite{preskill} but rather a true 
	paradox of physics that cannot be explained yet.
	True paradoxes of this kind are indicators that the scientific theories we use 
	do not provide a complete description of reality.
	
	The black hole information paradox is yet to be explained in a satisfactory manner by modern physics
	and perhaps will not be until a theory of quantum gravity is developed.
	Recently, however, interesting contributions to the black hole information problem have been made by people 
	from within the quantum information community \cite{blackadami,blackeradami,blackhayden,locking,hsu}.
	In the last chapter of this thesis, we present a curious and counter-intuitive result about the nature of
	purifications and then use this observation to make a speculative comment about black holes with highly 
	mixing internal dynamics.

	\bigskip
	\section{Polygamy of purification}
		
		In a closing remark of the original FQSW paper \cite{FQSW}, the authors make a very interesting observation
		about the nature of quantum purifications which we will refer to as \emph{polygamy of purification}.
		Consider three parties --- Alice, Bob and Ron who share the quantum state
		\be
			\ket{\psi}^{A^nB^nR^n}	=	\left( \ket{\Phi}^{A_BB}\otimes\ket{\Phi}^{A_RR} \right)^{\otimes n},
		\ee
		where $\ket{\Phi}$ denotes the maximally entangled state 
		$\ket{\Phi} = \tfrac{1}{\sqrt{2}}(\ket{00}+\ket{11})$.
		In other words, Alice shares $n$ entangled states with Ron and another $n$ maximally entangled states 
		with Bob. The entanglement structure is illustrated in Figure~\ref{fig:weird} a).

		\begin{figure}[ht]
			\begin{center} 
			\includegraphics[width=5.2in]{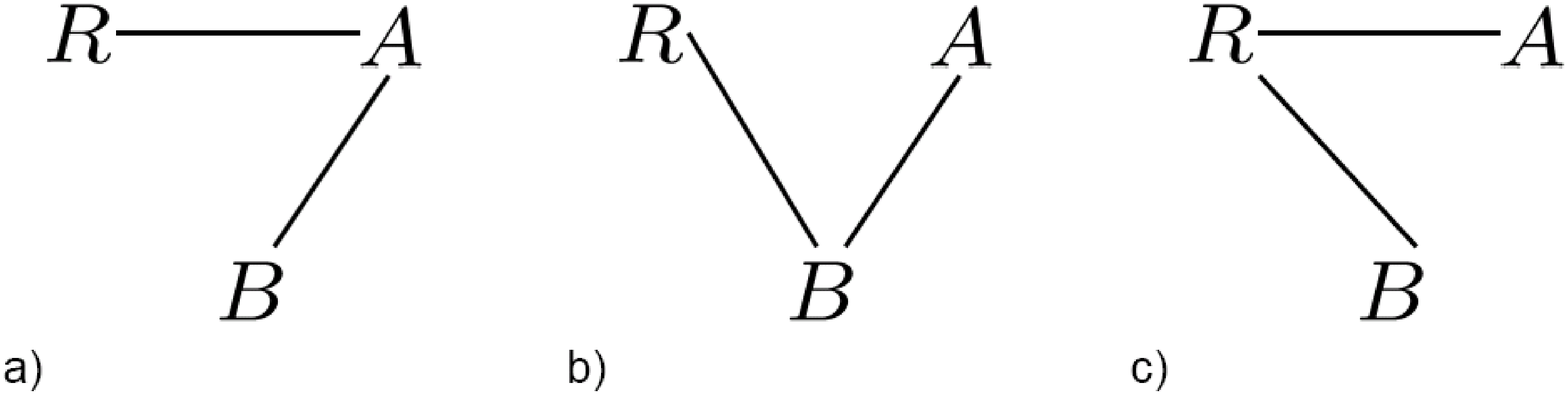}
			\end{center}
			\vspace{0.30cm}
		 	\FigureCaptionOpt{Transfer of quantum correlations between three parties.}{
							  Transfer of quantum correlations between three parties:
							  (a) The original $AR$ and $AB$ entanglement. 
							  (b) The effect of Alice sending the system $A_1$ to Bob. 
							  	  She is completely decoupled from the $R$ system.
							  (c)~Alternatively, Alice can send the \emph{same} $A_1$ system to  
							      Ron and completely decouple from Bob!}
		 	\label{fig:weird}
		\end{figure}
		
		Now, we tell Alice to perform the standard FQSW task, that is, to transfer her $R$ entanglement to Bob.
		Suppose that Alice performs the standard FQSW protocol in order to accomplish the 
		entanglement transfer.\!\!\footnote{Since in our setup the $R$-entangled part of her system is clearly
		identifiable, another approach for Alice could be to simply take the $n$ Ron-entangled qubits and send
		them to Bob.} 
		She applies a random unitary to the system $A^n$ and then sets aside a subsystem $A_1$ of 
		dimension $d_{A_1}$ where 
		\be
			\log d_{A_1} \geq \tfrac{1}{2}I(A;R)_{\phi} = n	\ \ \textrm{[qubits]}
		\ee
		as required by equation \eqref{eqn:FQSW} for the FQSW protocol.
		Sending the system $A_1$ to Bob will successfully decouple Alice from
		Ron and lead to the entanglement configuration illustrated in Figure~\ref{fig:weird} b).
		
										\index{entanglement!polygamy}
		Note, however, that the encoding operation was not specifically targeting Bob.
		Indeed, if the same $A_1$ system is sent to Ron instead, we would transfer the Bob entanglement to him
		and obtain the configuration of Figure~\ref{fig:weird} c).
		The polygamy of purification, therefore, is the observation that it is possible for a single 
		quantum system $A_1$ to contain the purification of more than one other system!

	\bigskip
	\section{Random internal dynamics for black holes}
		
		Recently, the results of the FQSW protocol were connected to the black hole information 
		paradox~\cite{blackhayden}. 
		The question studied is not about the evolution of the universe as whole but something more specific.
		If we drop half of pure state $\ket{\ph}^{AB_1}$ into a black hole, denoted $B_2$, how long will 
		it take for the its purification to come out? 
		
		Under the assumption that the internal dynamics of the black hole correspond to a random unitary operation,
		a situation which was considered previously in \cite{page2}, we can give an answer to this question 
		since it corresponds to an FQSW-type of problem except for the Schumacher compression step.
		We model the internal black hole dynamics as a random unitary $U_B$ which takes the system $B=B_1B_2$ 
		to an isomorphic system $B'R$, where $R$ is released as radiation and $B'$ is what
		remains of the black hole.
		The rest of the universe is denoted $U$ and no assumptions are made about its size.
		The situation is illustrated in Figure~\ref{fig:blackhole1}.
			
		\bigskip
	    \begin{figure}[ht]   \begin{center}
	        \input{figures/figureBlackHoleOne.pst}             \end{center}
	        \FigureCaptionOpt{Black hole before and after emitting the radiation system $R$.}{
		        			  {\bf a)} Black hole before the radiative process has taken place.
		        			  The purification of the $A$ system, $B_1$, is somewhere inside the black hole.
		        			  The system $U$ denotes the rest of the universe, i.e. everything that is not $A$ or~$B$.
		        			  \ {\bf b)} After the black hole emits the radiation chunk $R$ the remainder of 
		        			  the black hole is labeled $B'$.
		        			  }
	        \label{fig:blackhole1}
	    \end{figure}								
		
		Inspired by the FQSW results, we can say that if the dimension of the radiated system satisfies
		\be	\label{rateBlackHole}
			\log d_R \geq	\tfrac{1}{2}I(A;B) =  \tfrac{1}{2}I(A;B_1) = H(A)
		\ee
		then, with high probability, it will contain the purification of the $A$ system.
		This is because we can think of the black hole as an active entity mixing its internal degrees
		of freedom.
		
		In the current setup, we do not have the luxury of working in the i.i.d. regime so the statements we
		make are nothing more than inspired hand waving arguments. 
		Nevertheless, our calculation leads us to speculate that the purification information of a 
		specific system will come out fairly fast and independently of the size of the black hole.
		In fact, since the system we labeled $A$ was arbitrary, the purification of all 
		subsystems of the universe with the same dimension comes out with the radiation $R$!
		This is not be so surprising since we already know about the polygamy of purification.
		Nevertheless, even if the purification of any particular system of interest comes out quickly, we still
		have to wait until all of the black hole evaporates to recover the the purification of the whole universe,
		so the original black hole paradox remains.
			
		It is not clear what we mean when we say that the black hole has ``internal dynamics''.
		To assume that something interesting happens at the horizon is OK perhaps, but aren't black holes
		supposed to trap systems forever?

	\bigskip
	\section{Lost subsystem problem}
		
		Consider now a similar situation to the above but this time the black hole consists of 
		two systems $B_2L$, where the $L$ system is ``lost''; nothing ever leaves $L$.
		Half of a pure state $\ket{\ph}^{AB_1}$ is dropped into the black hole which is assumed to
		have random unitary dynamics on the space $B=B_1B_2$ from which a system $R$ is emitted.
		Once more we label $B'$ the remainder of the black hole as illustrated in Figure~\ref{fig:blackhole2}.
		
	    \begin{figure}[ht]   \begin{center}
	        \input{figures/figureBlackHoleTwo.pst}             \end{center}
	        \FigureCaptionOpt{Black hole which contains a lost subsystem $L$.}{
		        			  {\bf a)} The lost subsystem $L$ is part of the black hole $BL$.
	        				  The system $U$ denotes the rest of the universe.
	        				  \ {\bf b)} The black hole has released radiation $R$ from the $B$ subsystem.
	        				  The remainder of the black hole is $B'L$.}
	        \label{fig:blackhole2}
	    \end{figure}
		
		We would like to know how big the $R$ system has to be in order for the purification of $A$ to come out.
		This time, there are two active ``participants'': $B$ and $L$, so the multiparty FQSW results have to be
		considered. 
		Thus, in order for the purification of $A$ to come out the dimension of the radiated systems have to
		satisfy
		\bea
			\log d_R			&\geq&		\tfrac{1}{2}I(B;A) =  \tfrac{1}{2}I(B_1;A) = H(A),	\nonumber \\
			\log d_{R_L}		&\geq&		\tfrac{1}{2}I(L;A) =  0,	\\
		 \log d_R +\log d_{R_L}	&\geq&		\tfrac{1}{2}I(L;B;A)= H(A)+ \tfrac{1}{2}I(B_2;L). \nonumber
		\eea 
		where $d_{R_L}$ is the dimension of the system released by the lost system.

		At first sight, all seems to be in order since the requirement $\log d_{R_L} \geq 0$ is satisfied.
		The inequality for the sum of the rates, however, adds an extra requirement for $d_R$.
		To see the purification of $A$ come out we will have to wait until 
		\be
			\log d_R	\ >\ 	\max\{H(A), \ H(A)+ \tfrac{1}{2}I(B_2;L) \}.
		\ee
		
		Thus, if the are any significant correlations between the $B_2$ and $L$ parts of the black hole the 
		information will \emph{not} not come out quickly.
		This result is very interesting because the purification of $A$ will be slow to come out
		even though it is held in the $B$ part of the black hole and hasn't completely fallen into the $L$ system.

\chapter{Conclusion}

	This thesis has been an expedition into the field of quantum information science with many twists and turns.
	We began by introducing the fundamental principles of classical information theory and their extensions 
	to the quantum realm. 
	Armed with the basics, we were ready to approach some of last decade's important results in quantum information
	theory with the aim of getting readers from outside the field up to speed.

		
	We then attacked the multiparty distributed compression problem with the most powerful weapon available
	in our arsenal: the fully quantum Slepian-Wolf protocol.
	The construction of the multiparty distributed compression protocol is conceptually simple.
	It consists of sequential applications of the two-party FQSW protocol with careful accounting of the 
	information theoretic quantities at each step. 
	However, in order to achieve rigorous proofs of the bounds on the multiparty rate region, 
	we had to wage a heavy battle in difficult but interesting terrain. 
	
	To achieve a rigorous proof of Theorem~\ref{thm:THM1}, the inner bound on the rate region, we had to 
	dig into the geometry of convex polyhedra in $m$-dimensional space. 
	The proof we obtained uses a sufficient level of mathematical abstraction so as to apply to other problems 
	in information theory involving multiparty rate regions proved in terms of achievable points 
	but expressed instead in terms of facet inequalities. 
	Indeed, our proof is valid for all supermodular rate regions, that is, all rate region specified by 
	a set of inequalities 
	\be 
		\sum_{k \in \K}	R_\K	\geq	C_\K,	\qquad \quad  \forall \K \subseteq \{1,\ldots,m\}
	\ee
	for which the constants $C_\K$ satisfy the supermodular condition 
	$C_{\K\cup\cL} + C_{\K\cap\cL}  \geq  C_{\K} + C_{\cL}$.
	In particular, the rate regions for the classical multiparty Slepian-Wolf problem~\cite{W74,C75}
	and the multiparty state merging protocol~\cite{HOW05b} fall into this category because of strong subadditivity.

	Also, in order to prove Theorem~\ref{thm:THM2}, the outer bound on the rate region, it was necessary 
	to formulate a definition of the multiparty information and from it derive a multiparty generalization 
	of the squashed entanglement. 
	In the chapter dedicated to the multiparty squashed entanglement, we showed that it is a continuous,
	convex and subadditive measure of entanglement --- all desirable but rare properties in the multiparty case.

	Some open problems remain which could form fruitful directions for future investigations.
	The additivity of the multiparty squashed entanglement is an important conjecture that 
	was recently proved in an updated version of \cite{multisquash}, which now includes W. Song in the author list.
	As for the distributed compression problem, we have fully solved the problem only for separable states.
	Perhaps a different correction term exists for the outer bound? 
	If we find states for which we can calculate \Esq analytically or numerically
	we could use them to further probe the shape of the outer bound.
	Of course, the black hole information paradox remains an open problem since it hasn't been solved 
	by our toy-model observations. 
	
														\index{QIT@QIT: quantum information theory}
	And so, we add the new \emph{weapon of mass decoupling} to the ever growing collection
	of quantum information theory protocols derived from the nearly-universal building block of two-party FQSW. 
	At the time of writing of this thesis, this collection contains entanglement distillation, 
	channel simulation, communication over quantum broadcast channels, and many others.
	In fact, even the more general state redistribution \cite{DY07} result can be 
	obtained from the FQSW protocol~\cite{O07}.

\printindex{Glossary of technical terms}{Glossary of technical terms:}%
{For the convenience of the reader, we have collected in this section all technical terms and abbreviations that
were used in the document. The first page reference points to the page where the concept is defined. \ \\ }

\bibHeading{References}
\bibliography{thesis}
\bibliographystyle{unsrt}

%
%

\end{document}

%% file: include/header.tex
\font\gensymbols=drgen10
\def\male{{\gensymbols\char"1A}}
\def\female{{\gensymbols\char"19}}

\newcommand{\rsrc}[1]{ \langle #1 \rangle}

\newcommand{\typ}{  \ensuremath{ T^{(n)}_\epsilon  } }
\newcommand{\ptyp}{ \ensuremath{ \Pi_\epsilon^{(n)}} }

\newcommand{\joke}[2]{#2}

\newtheorem{definition}{Definition}[chapter]
\newtheorem{proposition}[definition]{Proposition}
\newtheorem{lemma}[definition]{Lemma}

\newtheorem{theorem}[definition]{Theorem}
\newtheorem{corollary}[definition]{Corollary}

\newtheorem{example}[definition]{Example} 

\newcommand{\nc}{\newcommand}
\nc{\rnc}{\renewcommand} 
\nc{\beq}{\begin{equation}}
\nc{\eeq}{{\end{equation}}} 
\nc{\beqa}{\begin{eqnarray}}
\nc{\eeqa}{\end{eqnarray}} \nc{\lbar}[1]{\overline{#1}}
\nc{\ketbra}[2]{|#1\rangle\!\langle#2|}
\nc{\proj}[1]{|#1\rangle\!\langle #1 |} 
\nc{\avg}[1]{\langle#1\rangle}
\rnc{\max}{\operatorname{max}} 
\nc{\rank}{\operatorname{rank}}
\nc{\conv}{\operatorname{conv}}
\nc{\smfrac}[2]{\mbox{$\frac{#1}{#2}$}}
\nc{\Tr}{\operatorname{Tr}}
\nc{\ox}{\otimes}
\nc{\dg}{\dagger}
\nc{\dn}{\downarrow} \nc{\cA}{{\cal A}} \nc{\cB}{{\cal B}}
\nc{\cC}{{\cal C}} \nc{\cD}{{\cal D}} \nc{\cE}{{\cal E}}
\nc{\cF}{{\cal F}} \nc{\cG}{{\cal G}} \nc{\cH}{{\cal H}}
\nc{\cI}{{\cal I}} \nc{\cJ}{{\cal J}} \nc{\cK}{{\cal K}}
\nc{\cM}{{\cal M}} \nc{\cN}{{\cal N}}
\nc{\cO}{{\cal O}} \nc{\cP}{{\cal P}} \nc{\cR}{{\cal R}}
\nc{\cS}{{\cal S}} \nc{\cT}{{\cal T}} \nc{\cU}{{\cal U}}
\nc{\cX}{{\cal X}} \nc{\cW}{{\cal W}} \nc{\cZ}{{\cal Z}}
\nc{\csupp}{{\operatorname{csupp}}}
\nc{\qsupp}{{\operatorname{qsupp}}} 
\nc{\rar}{\rightarrow} \nc{\lrar}{\longrightarrow}
\nc{\poly}{\operatorname{poly}}
\nc{\polylog}{\operatorname{polylog}}
\nc{\Lip}{\operatorname{Lip}} 

\def\>{\rangle}
\def\<{\langle}

\def\ph{\varphi}

\nc{\RR}{{{\mathbb R}}}
\nc{\CC}{{{\mathbb C}}}  
\nc{\NN}{{{\mathbb N}}}
\nc{\ZZ}{{{\mathbb Z}}}
\nc{\UU}{{{\mathbb U}}}
\nc{\id}{{\operatorname{{\mathbf 1}}}}
\def \Tr{\textup{Tr}}

\def \calH{ \ensuremath{ \cal H } }

\renewcommand{\ket}[1]{\left\vert #1 \right\rangle}
\renewcommand{\bra}[1]{\left\langle #1 \right\vert}
\newcommand{\braket}[2]{\left\langle #1\vert#2 \right\rangle}

\newcommand{\m}{  \ensuremath{\mathcal M} }
\newcommand{\K}{  \ensuremath{\mathcal K} }
\newcommand{\cL}{  \ensuremath{\mathcal L} }
\newcommand{\Q}{  \ensuremath{\mathcal Q} }
\newcommand{\Pp}{  \ensuremath{\mathcal P} }
\newcommand{\Ww}{  \ensuremath{\mathcal W} }
\newcommand{\qq}{  \ensuremath{\vec{q}} }
\newcommand{\Kbar}{{  \ensuremath{ \bar{ {\mathcal K} } }   }}

\newcommand{\otimesc}{\negthickspace\otimes\negthickspace}	

\newcommand{\Esq}{\ensuremath{E_\text{sq}} }

\newcommand{\rhot}{{\ensuremath{\tilde{\rho}}}}
\newcommand{\sigmat}{{\ensuremath{\tilde{\sigma}}}}

\newcommand{\be}{\begin{equation}}
\newcommand{\ee}{\end{equation}}
\newcommand{\bea}{\begin{eqnarray}}
\newcommand{\eea}{\end{eqnarray}}
\newcommand{\bs}{\begin{split}}
\newcommand{\es}{\end{split}}
\renewcommand{\neg}{\mbox{-}}
\newcommand{\XX}{X_1;\ldots;X_m}
\newcommand{\XcX}{X_1,\ldots,X_m}
\newcommand{\E}{   {\cal E} }
\newcommand{\beas}{\begin{eqnarray*}} 
\newcommand{\eeas}{\end{eqnarray*}} 

\newcommand{\myparagraph}[1]{\noindent{\bf #1}$\quad$}

%% file: include/graphviz.tex
\newcommand{\digraph}[3][scale=1]{
\newwrite\dotfile
\immediate\openout\dotfile=#2.dot
\immediate\write\dotfile{digraph #2 {\string#3}}
\immediate\closeout\dotfile
\IfFileExists{#2.ps}
{ \includegraphics[#1]{#2} }
	{ \fbox{ \begin{tabular}{l}
		The file \texttt{#2.ps} hasn't been created from
		\texttt{#2.dot} yet. \\
		Run `\texttt{dot -Tps -o #2.ps #2.dot}' to create it. \\
		Here is a \textsf{bash} loop to process all \textsf{dot} files
		in the current directory: \\
		\texttt{
			for f in *.dot do ;
			dot -Tps -o \$\{f\%dot\}ps \$f ;
			done
		}
		\end{tabular}}
	}
}